\documentclass[acmsmall,screen,nonacm]{acmart}

\setcopyright{none}
\renewcommand\footnotetextcopyrightpermission[1]{}

\AtBeginDocument{%
  }

\usepackage[utf8]{inputenc}
\usepackage[T1]{fontenc}
\usepackage[english]{babel}

\usepackage{microtype}
\usepackage{enumitem}  
\usepackage{relsize}  
\usepackage[noabbrev,capitalize]{cleveref}  
\usepackage{mathpartir}  
\usepackage{thm-restate}
\usepackage[normalem]{ulem}    
\usepackage[clock]{ifsym}  
\usepackage{wrapfig}
\usepackage{xcolor}
\usepackage{pifont}  

\usepackage{xr-hyper}
\usepackage{tikz}
\usetikzlibrary{fit, shapes.multipart, decorations,arrows,arrows.meta,shapes,shapes.multipart,decorations.pathreplacing,calc,positioning,decorations.pathmorphing}

\usepackage[final]{listings}

\newcommand{\lstCodeSize}{\footnotesize\ttfamily} 
\definecolor{keywordcolour}{rgb}{0.5,0,0.35}
\definecolor{greencomments}{rgb}{0,0.5,0}
\definecolor{redstrings}{rgb}{0.9,0,0}
\lstdefinelanguage{scribble}{
  basicstyle=\footnotesize\ttfamily,
  stringstyle=\color{blue},
  showstringspaces=false,
  keywords={mixed,and,as,at,by,catches,choice,continue,do,from,global,import,instantiates,interruptible,local,module,or,par,protocol,rec,role,sig,throws,to,type,with,int,aux},
  morestring=[b]",
  morestring=[b]',
  morecomment=[l][\color{greencomments}]{//},
  literate={->}{{${\rightarrow\ }$}}1 {>}{{$>$}}1 {<}{{$<$}}1 {<=}{{$\leq$}}1 {>=}{{$\geq$}}1 {&&}{{$\land$}}1 {||}{{$\lor$}}1 {!=}{{$\neq$}}1 {=}{{$=$}}1
}

\lstdefinelanguage{Erlang}{
	morekeywords={
		after, and, andalso, band, begin, bnot, bor, bsl, bsr, bxor,
		case, catch, cond, div, end, fun, if, let, not, of, or, orelse,
		receive, rem, try, when, xor, module, export, import, define,
		record, include, behaviour, callback_mode, throw, true, false,
		begin, do, else, query, whereis, internal, cast, self, stop
	},
	sensitive=true,                       
	keywordstyle=\color{blue},            
	comment=[l]{\%},                      
	commentstyle=\color{greencomments},            
	stringstyle=\color{orange},           
	morestring=[b]",                      
	morestring=[b]',                      
	basicstyle=\ttfamily\footnotesize,    
	numbers=left,                         
	numberstyle=\tiny\color{gray},        
	stepnumber=1,                         
	numbersep=5pt,                        
	tabsize=2,                            
	breaklines=true,                      
	breakatwhitespace=false,              
	showstringspaces=false,               
	escapeinside={(*@}{@*)},              
	alsoletter={_},                       
	otherkeywords={->, <-, ==, /=, =<, >=, <, >, ::, =>, ++, --, **, !, ?}, 
	literate=
			{start(}{{\color{keywordcolour}start}\color{black}(}1
			{start_link(}{{\color{keywordcolour}start\_link}\color{black}(}1
			{stop(}{{\color{keywordcolour}stop}\color{black}(}1
			{call(}{{\color{keywordcolour}call}\color{black}(}1
			{reply(}{{\color{keywordcolour}reply}\color{black}(}1
			{enter_loop(}{{\color{keywordcolour}enter\_loop}\color{black}(}1
			{send_event(}{{\color{keywordcolour}send\_event}\color{black}(}1
			{send_all_state_event(}{{\color{keywordcolour}send\_all\_state\_event}\color{black}(}1
			{sync_send_event(}{{\color{keywordcolour}sync\_send\_event}\color{black}(}1
			{sync_send_all_state\_event(}{{\color{keywordcolour}sync\_send\_all\_state\_event}\color{black}(}1
			{start_timer(}{{\color{keywordcolour}start\_timer}\color{black}(}1
			{cancel_timer(}{{\color{keywordcolour}cancel\_timer}\color{black}(}1
			{state_timeout(}{{\color{keywordcolour}state\_timeout}\color{black}(}1
			{timeout(}{{\color{keywordcolour}timeout}\color{black}(}1
			{hibernate(}{{\color{keywordcolour}hibernate}\color{black}(}1
			{format_status(}{{\color{keywordcolour}format\_status}\color{black}(}1
			{code_change(}{{\color{keywordcolour}code\_change}\color{black}(}1
			{terminate(}{{\color{keywordcolour}terminate}\color{black}(}1
			{s1(}{{\color{keywordcolour}s1}\color{black}(}1
			{s2(}{{\color{keywordcolour}s2}\color{black}(}1
			{s3(}{{\color{keywordcolour}s3}\color{black}(}1
			{s4(}{{\color{keywordcolour}s4}\color{black}(}1
			{s5(}{{\color{keywordcolour}s5}\color{black}(}1
			{s6(}{{\color{keywordcolour}s6}\color{black}(}1
			{s1(}{{\color{keywordcolour}s1}\color{black}(}1
			{s2(}{{\color{keywordcolour}s2}\color{black}(}1
			{s3(}{{\color{keywordcolour}s3}\color{black}(}1
			{s4(}{{\color{keywordcolour}s4}\color{black}(}1
			{s5(}{{\color{keywordcolour}s5}\color{black}(}1
			{s(}{{\color{keywordcolour}s6}\color{black}(}1
			{send_a1(}{{\color{keywordcolour}send\_a1}\color{black}(}1
			{send_a2(}{{\color{keywordcolour}send\_a2}\color{black}(}1
			{send_a3(}{{\color{keywordcolour}send\_a3}\color{black}(}1
			{send_a4(}{{\color{keywordcolour}send\_a4}\color{black}(}1
			{send_a5(}{{\color{keywordcolour}send\_a5}\color{black}(}1
			{send_a6(}{{\color{keywordcolour}send\_a6}\color{black}(}1
			{send_to(}{{\color{keywordcolour}send\_to}\color{black}(}1
}

\usepackage{macros}
\usepackage{amsmath}
\usepackage{thmtools}
\usepackage{bussproofs}
\usepackage{stmaryrd}  
\usepackage{multirow}
\usepackage{array}
\usepackage{tabularx}
\usepackage{booktabs}
\usepackage{svg}
\newtheorem{remark}{Remark}

\newcommand{\CODE}[1]{\texttt{\small#1}}

\makeatletter
\providecommand*\xrightarrowtriangle[2][]{%
  \ext@arrow 0055{\arrowfill@\relbar\relbar\rightarrowtriangle}{#1}{#2}}
\makeatother

\newcommand*{\rightsquigarrowtriangle}{%
    \mathrel{\raisebox{.2mm}{%
        \begin{tikzpicture}
            \draw[
            arrows={-Triangle[scale=0.8]},
            line cap=round, decorate, decoration={
                zigzag,
                segment length=3.7,
                amplitude=.8,
                post length=2pt,
            }] (0,0) -- (.3,0);
        \end{tikzpicture}%
    }}%
}

\begin{document}

\title{Mixed Choice in Asynchronous Multiparty Session Types}

\author{Laura Bocchi}
\email{L.Bocchi@kent.ac.uk}
\orcid{0000-0002-7177-9395}
\affiliation{%
 \institution{University of Kent}
 \city{Canterbury}
 \country{UK}
}

\author{Raymond Hu}
\email{r.hu@qmul.ac.uk}
\orcid{0000-0003-4361-6772}
\affiliation{%
 \institution{Queen Mary University of London}
 \city{London}
 \country{UK}
}

\author{Adriana Laura Voinea}
\email{laura.voinea@glasgow.ac.uk}
\orcid{0000-0003-4482-205X}
\affiliation{%
 \institution{University of Glasgow}
 \city{Glasgow}
 \country{UK}
}

\author{Simon Thompson}
\email{S.J.Thompson@kent.ac.uk}
\orcid{0000-0002-2350-301X}
\affiliation{%
 \institution{University of Kent}
 \city{Canterbury}
 \country{UK}
}

\begin{abstract}
We present a multiparty session type (MST) framework with \emph{asynchronous
mixed choice} (MC).
We propose a core construct for MC that allows transient inconsistencies in
protocol state between distributed participants, but ensures all participants
can always eventually reach a mutually consistent state.
We prove the correctness of our system by establishing a progress property and
an operational correspondence between global types and distributed local type
projections.
Based on our theory, we implement a practical toolchain for specifying and
validating asynchronous MST protocols featuring MC, and programming compliant
\CODE{gen\_statem} processes in Erlang/OTP.
We test our framework by using our toolchain to specify and reimplement part of
the \CODE{amqp\_client} of the RabbitMQ broker for Erlang.


\end{abstract}

\maketitle

%
\section{Introduction}
\label{sec:intro}

\emph{Multiparty session types} (MST)~\cite{HondaYC08} is a typing discipline
for concurrent processes that interact via message passing in communication
sessions.
The main idea is that an MST communication protocol can be statically checked
for \emph{communication safety}, i.e., freedom from fundamental errors such
as reception errors (receiving unexpected messages), deadlocks (wait-for
cycles) and orphan messages (messages that the receiver will never
attempt to consume).

MST is an active area of research due to its potential to offer programmatic
techniques for safe specification and lightweight verification of communication
protocols in concurrent and distributed systems.
The key challenges being tackled include expressiveness of the types,
tractability of the metatheory, and practicality of session-based programming
and verification methods.
These challenges are accentuated in the setting of distributed systems (DS)
where communications are inherently asynchronous and failure is the norm.

This paper tackles a crucial problem that concerns all of the above
challenges: the notion of \emph{mixed choice} in \emph{asynchronous} MST.
In classical
MST~\cite{HondaYC08,DBLP:conf/concur/BettiniCDLDY08long,DBLP:conf/sfm/CoppoDPY15},
the construct for choice in a protocol, called a \emph{directed choice}, looks
as follows.

\centerline{ \(
\p \mathbin{\rightarrowtriangle} \q : \{ a_i . \GT_i \}_{i \mathclose{\in} I}
\)}

\noindent
It specifies that a participant $\p$ makes an \emph{internal} choice to send
one of the $a_i$ messages to $\q$ and continue in protocol $G_i$.
Participant $\q$ receives the $a_i$ as an \emph{external} choice and continues
in $G_i$ correspondingly.
Some systems (e.g.,~\cite{DBLP:conf/cav/LangeY19,DBLP:conf/cav/LiSWZ23})
support generalised choice constructs that allow the $a_i$ to be sent to
different $\q_i$.
However, the key point remains that the choice is directed by $\p$.

By contrast, this paper develops the following construct for \emph{mixed}
choice.

\centerline{\(
\q \mathbin{\rightarrowtriangle} \p : a_1 . \GT_1 ~~\triangleright~~ \p \mathbin{\rightarrowtriangle} \q : a_2 . \GT_2
\)}

\noindent
It specifies that $\p$ and $\q$ each independently face a choice between a
\emph{mix} of input and output actions: $\q$ faces a mixed choice between
sending $a_1$ on the left and receiving $a_2$ on the right, and vice versa for
$\p$.
In an asynchronous setting, this means that $\q$ may opt to send $a_1$ and
continue in protocol $G_1$ \emph{concurrently with} $\p$ opting to send $a_2$
and continue in $G_2$.
In this way, asynchronous mixed choices inherently describe a form of race
condition, which classical MST intentionally prohibits outright because
directed choice syntactically partitions all choices as input or output only.

Yet such race conditions are useful and important in many real applications.
Indeed, much recent research on improving the expressiveness and practicality
of MST has touched on aspects of mixed choice, including work on
\emph{exceptions}~\cite{DBLP:journals/mscs/CapecchiGY16,DBLP:conf/esop/VieringCEHZ18,DBLP:journals/pacmpl/FowlerLMD19},
\emph{interrupts}~\cite{DBLP:journals/fmsd/DemangeonHHNY15,DBLP:conf/forte/ChenVBZE16},
\emph{timing}~\cite{DBLP:journals/pacmpl/IraciCHZ23} and
\emph{timeouts}~\cite{DBLP:conf/coordination/PearsBK23,DBLP:conf/ecoop/HouLY24},
\emph{failure
handling}~\cite{DBLP:conf/forte/AdameitPN17,DBLP:conf/concur/BarwellSY022,DBLP:conf/ecoop/BarwellHY023,DBLP:conf/forte/BrunD24},
and
\emph{fault-tolerance}~\cite{DBLP:journals/pacmpl/VieringHEZ21,DBLP:journals/lmcs/PetersNW23}.
These works (implicitly) involve patterns where a participant has the option on
one hand to asynchronously output a message, while on the other hand it is
simultaneously prepared for some input event, say, catching a concurrent
channel exception, or receiving a criss-crossing interrupt or timeout message,
or handling the failure of some other participant.

The fundamental problem in reasoning about mixed choice in MST is that it
allows participants to diverge in their local views of the protocol during
execution.
The insight of this paper is that communication safety can be achieved despite
\emph{transient} inconsistencies in distributed protocol state depending on how
participants react to race-y messages and interact to resolve conflicts.

\paragraph{Contributions and roadmap}

This paper presents the following contributions.

\begin{itemize}[leftmargin=*]
\item
We introduce \textbf{\mMST{}}, the first theory of global and local MST with an
\emph{explicit} construct for \emph{asynchronous mixed choice}.
To date, mixed choices in asynchronous sessions have only been partially
expressed through specific-purpose constructs for exceptions, failure handling
and so forth.
By contrast, we present a general-purpose core theory of mixed choices that
captures and unifies the fundamental essence of such ad hoc constructs.
We propose an \emph{asymmetric} design for mixed choice that statically ensures
participants can eventually resolve the inherent race conditions through
explicit interactions and agree on how the protocol should proceed.

\item
We establish the correctness of \mMST{} by proving a progress
property for our mixed choice types, and an operational
correspondence between global types and the corresponding system of local
types.
Together these guarantee that every non-terminated participant in an \mMST{}
protocol can always progress and that its behaviour is always
protocol-compliant.
Our results lay general foundations for a more flexible and practical concept
of multiparty sessions that allows participants to deal with discrepancies in
their distributed views of the protocol and safely converge on a consistent
outcome.

\item
We apply our theory by implementing a prototype toolchain for specifying and
programming \mMST{}-based protocols as \CODE{gen\_statem} programs in
Erlang/OTP.
We test the expressiveness and practicality of \mMST{} by using our toolchain
to implement coordination protocols in the Erlang client of the RabbitMQ
message broker, as well as a selection of examples from MST literature augmented
with mixed choices.
\end{itemize}

\noindent
\textbf{\Cref{sec:overview}} gives a high-level overview of our mixed choice
construct.
\textbf{\Cref{sec:globaltypes,sec:localtypes}} present our theory and its
formal properties.
\textbf{\Cref{sec:implementation}} describes our toolchain for Erlang, our
RabbitMQ use case and other examples.
\textbf{\Cref{sec:related}} discusses related work, limitations and future work.
Full details and proofs are available in the appendix. 

\section{Overview}
\label{sec:overview}

Our overall framework comprises two main stages.

\begin{itemize}[leftmargin=*]
\item
Specification and static validation of a source \mMST{} protocol based on our
formal theory.
\item
Implementation of each role in the asynchronous protocol in Erlang using
correct-by-construction modules generated from the source protocol.
\end{itemize}

\noindent
The following two subsections illustrate the key concepts in each stage using a
practical example.

\smallskip
\begin{figure}[t]
\small
\begin{tabular}{l|l}
\begin{minipage}{0.4\textwidth}
{\begin{lstlisting}[numbers=left,language=scribble]
// Exception, Interrupt, etc. similar
global protocol Timeout(
      role A, role B, role C) {
  mixed {  // "Left-hand" side (LHS)
    a1() from A to B;  (*\label{ln:L1}*)
    a2() from A to C;
    a3() from B to C;  (*\label{ln:CLHS}*)
    a4() from B to A;  (*\label{ln:ALHS}*)
    a5() from C to A;  (*\label{ln:L2}*)
  } or {  // "Right-hand" side (RHS)
    TOa() from B to A;  (*\label{ln:ARHS}\label{ln:R1}*)
    TOc() from B to C;  (*\label{ln:CRHS}\label{ln:R2}*)
} }
\end{lstlisting}}
\end{minipage}
&
\begin{tabular}{c}
\begin{minipage}{0.45\textwidth}
{\begin{tikzpicture}
\node (A) {\CODE{A}};
\node[below=3mm of A] (B) {\CODE{B}};
\node[below=3mm of B] (C) {\CODE{C}};
\draw[dotted] (A) --
    node[pos=0.1,above](A1){\GREY{$1, !a_\mathtt{1}$}}
    node[pos=0.25,above](A2){\GREY{$2, !a_\mathtt{2}$}}
    node[pos=0.67,above](A3){}
    node[pos=0.75,above](A4){\GREY{$3,?a_\mathtt{4}$}}
    node[pos=0.9,above](A5){\GREY{$4, ?a_\mathtt{5}$}}
    ++(6cm, 0) node (AA) {};
\draw[dotted] (B) --
    node[pos=0.35,below](B1){\GREY{$1, {}?a_\mathtt{1}$}}
    node[pos=0.45,above](B2){\GREY{$2, {}!a_\mathtt{2}$}}
    node[pos=0.55,below](B3){\GREY{$3, {}!a_\mathtt{4}$}}
    ++(6cm, 0) node (BB) {};
\draw[dotted] (C) --
    node[pos=0.3,below](C1){\GREY{$1, ?a_\mathtt{2}$}}
    node[pos=0.52,below](C2){\GREY{$2, ?a_\mathtt{3}$}}
    node[pos=0.66,below](C3){\GREY{$3, !a_\mathtt{5}$}}
    ++(6cm, 0) node (CC) {};
\node[anchor=center] at (A5.south) {\tiny$\bullet$};
\draw[->] (A1.south) -- node[left]{\BLUE{$a_1$}} (B1.north);
\draw[->] (A2) -- node[right,pos=0.15]{\BLUE{$a_2$}} (C1);
\draw[->] (B2) -- node[left,pos=0.6]{\BLUE{$a_3$}} (C2);
\draw[->] (B3.north) -- node[above,pos=0.2]{\BLUE{$a_4$}} (A4.south);
\draw[->] (C3) -- node[right,pos=0.3]{\BLUE{$a_5$}} (A3);
\draw[dashed,gray,->] (A3.center) -- ++(0,5mm) -- ++(13mm,0) -- ++(0,-1.8mm);
\end{tikzpicture}}
\end{minipage}
\\
\\[-2mm]
\begin{minipage}{0.45\textwidth}
{\begin{tikzpicture}
\node (A) {\CODE{A}};
\node[below=3mm of A] (B) {\CODE{B}};
\node[below=3mm of B] (C) {\CODE{C}};
\draw[dotted] (A) --
    node[pos=0.2,above](A1){\GREY{$1, !a_\mathtt{1}$}}
    node[pos=0.4,above](A2){\GREY{$2, !a_\mathtt{2}$}}
    node[pos=0.65,above](A3){\GREY{$3, ?^*\mathtt{TOa}$}}
    ++(5cm, 0) node (AA) {};
\draw[dotted] (B) --
    node[pos=0.1,below](B1){\GREY{$1,!^*\mathtt{TOa}$}}
    node[pos=0.35,below](B2){\GREY{$5, \mathsf{gc}$}}
    node[pos=0.48,above](B3){\GREY{$5, !\mathtt{TOc}$}}
    ++(5cm, 0) node (BB) {};
\draw[dotted] (C) --
    node[pos=0.65,below](C1){\GREY{$1, ?^*\mathtt{TOc}$}}
    node[pos=0.95,below](C2){\GREY{$4, \mathsf{gc}$}}
    ++(5cm, 0) node (CC) {};
\draw[->,red] (A1) -- node[right,pos=0.3]{\BLUE{$a_1$}} (B2);
\draw[->] (A2) -- node[above]{\BLUE{$a_2$}} (C2);
\draw[->,red] (B1.north) -- node[above,pos=0.1]{\BLUE{$\mathtt{TOa}$}} (A3.south);
\draw[->] (B3) -- node[right]{\BLUE{$\mathtt{TOc}$}} (C1);
\end{tikzpicture}}
\end{minipage}
\end{tabular}
\\
\end{tabular}
\caption{A multiparty timeout pattern as a safe \mMST{} protocol using asynchronous mixed choice.}
\label{fig:timeout}
\end{figure}

\subsection{Asymmetric Mixed Choice in Asynchronous MST}
\label{sec:asymmMC}

Our toolchain takes \mMST{} protocols written in our extension of the Scribble
protocol language~\cite{DBLP:conf/tgc/YoshidaHNN13,DBLP:conf/fase/HuY16}.
Communications are asynchronous: interactions are non-blocking on
the sender side, while the receiver side blocks until a message is available
for reading.
This means a sender moves ahead in the protocol immediately after dispatching a
message without waiting for the message to be received.
Messages are delivered in order of dispatch in each direction between each pair
of roles; receivers read messages from the expected sender in a FIFO manner.
This model reflects our target Erlang programs and wider domains such as
TCP-based Internet applications and Web services.

\Cref{fig:timeout} (left) illustrates a small protocol called \CODE{Timeout}
involving three participants, whose behaviour is abstracted in MST as
\emph{roles}.
The syntax of our \emph{mixed choice} (MC) construct is:

\smallskip
\centerline{\CODE{mixed \{ /* "Left-hand" side (LHS) */ \} or \{ /* "Right-hand" side (RHS) */ \}}}
\smallskip

\noindent

The protocol features an MC between role \CODE{A} on the LHS and
\CODE{B} on the RHS.
It expresses a typical timeout pattern where \CODE{A} has the option to send
message \CODE{a1} to \CODE{B} (and proceed asynchronously) on the LHS, but it
must also be prepared to handle the potentially concurrent timeout message
\CODE{TOa} from \CODE{B} on the RHS.
Conversely, \CODE{B} has the option to wait for the \CODE{a1} message on the
LHS, or give up waiting and asynchronously send \CODE{TOa} on the RHS.
Depending on how \CODE{A} and \CODE{B} proceed, role \CODE{C} must be prepared
to handle one or both of the \CODE{a2} from \CODE{A} and the \CODE{TOc} from
\CODE{B}.

We summarise the key concepts in the design of our MC and how we ensure MST
safety.

\begin{itemize}[leftmargin=*]
\item
Our MC is an \emph{asymmetric} construct.
In each MC, we designate the sender on the RHS as a special role that we refer
to as the \emph{observer} of the MC.
The LHS can be considered a default or speculative branch, which can be
asynchronously superceded by the RHS on the instigation of the
observer.

\item
We identify a notion of \emph{commitment} of roles to an MC branch (LHS or
RHS).
Commitment means that the inherent race condition of an MC has been resolved
from the perspective of that role and it knows the protocol will henceforth
proceed only in that branch.

An MC starts with no roles committed to either branch.
Regardless of (speculative) interactions between other roles in the LHS, the
first action by the observer in an MC (input on the LHS or output on the RHS)
commits the observer to that branch.
Any subsequent action by another role $\role{r}$, where the action causally
depends on the observer, commits $\role{r}$ to the same branch.

\item
Protocol validation due to our formal theory ensures that however a protocol
(speculatively) proceeds, all eventual commitments are monotonic and
always consistent with the observer.
\end{itemize}

To illustrate, \Cref{fig:timeout} (right) depicts two possible executions of
the protocol.
For now, the reader can focus on the arrows and the blue labels; the grey
annotations will be explained in Sec.~\ref{sec:moti}.

\begin{itemize}[leftmargin=*]
\item
The upper chart is a run where observer \CODE{B} opts to receive \CODE{a1} and
follow \CODE{A} on the LHS, i.e., it does not raise the timeout.
While \CODE{A} and \CODE{C} may proceed asynchronously on the LHS, the observer
\CODE{B} is the first role to commit when it consumes the \CODE{a1}.
In turn \CODE{C} commits to the LHS when it consumes the \CODE{a3} from
\CODE{B}, and \CODE{A} commits when it consumes the \CODE{a4}.
Although the \CODE{a5} can arrive on the LHS at \CODE{A} before the \CODE{a4},
\CODE{A} must consume the \CODE{a4} (and commit to the LHS) first following the
protocol.

\item
The lower chart is a run where observer \CODE{B} opts to overrule \CODE{A} and
raise the timeout by sending \CODE{TOa}, which commits \CODE{B} to the RHS.
Although \CODE{A} and \CODE{C} may be concurrently engaged in interactions on
the LHS, role \CODE{A} eventually receives the \CODE{TOa}, causing it to switch
from the LHS and commit to the RHS.
Similarly, \CODE{C} switches and commits to the RHS when it receives the
\CODE{TOc}.
\end{itemize}

\paragraph{Stale message purging.}
The latter of the above cases demonstrates that supporting mixed choices safely
in asynchronous MST requires one further key concept.
Due to asynchrony, actions performed by one role may concurrently render other
messages that are buffered or being sent obsolete.
For instance, in the lower chart, the \CODE{a1} message `criss-crosses' with
\CODE{TOa} and arrives at \CODE{B} after \CODE{B} has already committed to the
RHS.
From \CODE{B}'s perspective, sending the \CODE{TOa} renders the (already
in-transit) \CODE{a1} obsolete; in such cases, we refer to messages like
\CODE{a1} as \emph{stale} messages.
Similarly, the receipt of \CODE{TOc} by \CODE{C} renders the \CODE{a2} stale.

Our theory formalises a runtime mechanism that allows stale messages to be
automatically purged \emph{transparently} to the user program; e.g., in the
case mentioned above, the runtime at \CODE{B} (resp.\ at \CODE{C}) will
transparently purge the \CODE{a1} that arrives after sending \CODE{TOa} (resp.\
the \CODE{a2} after receiving \CODE{TOc}).
Crucially, stale message purging can be safely performed by the \emph{local}
runtime of each distributed participant using only knowledge available to that
participant.
The notion of stale messages and purging never arises in classical MST.
Our transparent purging mechanism could be considered a message passing
analogy to garbage collection of stale objects in programming languages with
automated memory management (such as Erlang).

\paragraph{MC protocol validation.}
For reference, the \CODE{Timeout} protocol written in our formal notation for
\emph{global types} (\Cref{sec:globaltypes}) is as follows.
The $\triangleright$ separates the LHS and RHS of the MC.
The additional \BLUE{blue} colouring indicates the points at which the roles become
committed.

\smallskip
\centerline{$
\mathtt{A} \mathbin{\rightarrowtriangle} \mathtt{B} : a_1 \,.\,
\mathtt{A} \mathbin{\rightarrowtriangle} \mathtt{C} : a_2 \,.\,
\mathtt{\color{blue}B} \mathbin{\rightarrowtriangle} \mathtt{\color{blue}C} : a_3 \,.\,
\mathtt{B} \mathbin{\rightarrowtriangle} \mathtt{\color{blue}A} : a_4 \,.\,
\mathtt{C} \mathbin{\rightarrowtriangle} \mathtt{A} : a_5 \,.\, \tend
\;\;\triangleright\; \;
\mathtt{\color{blue}B} \mathbin{\rightarrowtriangle} \mathtt{\color{blue}A} : \mathtt{TOa} \,.\,
\mathtt{B} \mathbin{\rightarrowtriangle} \mathtt{\color{blue}C} : \mathtt{TOc} \,.\, \tend
$~~\textcolor{green}{\ding{51}}}
\smallskip

\noindent
Our static protocol validation ensures that role commitments are always
consistent with the observer and that role termination is safe (as illustrated
in the earlier example runs).
By contrast, the below (renaming \CODE{TOc}) is rejected because commitment is
ambiguous for \CODE{C}.

\smallskip
\centerline{\(
\mathtt{A} \mathbin{\rightarrowtriangle} \mathtt{B} : a_1 \,.\,
\mathtt{A} \mathbin{\rightarrowtriangle} \mathtt{C} : a_2 \,.\,
\mathtt{\color{blue}B} \mathbin{\rightarrowtriangle} \mathtt{\color{blue}C} : \RED{a_3} \,.\,
\mathtt{B} \mathbin{\rightarrowtriangle} \mathtt{\color{blue}A} : a_4 \,.\,
\mathtt{C} \mathbin{\rightarrowtriangle} \mathtt{A} : a_5 \,.\, \tend
\;\;\triangleright\; \;
\mathtt{\color{blue}B} \mathbin{\rightarrowtriangle} \mathtt{\color{blue}A} : \mathtt{TOa} \,.\,
\mathtt{B} \mathbin{\rightarrowtriangle} \mathtt{\color{blue}C} : \RED{a_3} \,.\, \tend
\)~~\textcolor{red}{\ding{55}}}
\smallskip

\noindent
The below (dropping the $a_3$) is also rejected because \CODE{C} reaches \CODE{end}
without committing on the LHS.

\smallskip
\centerline{$
\mathtt{A} \mathbin{\rightarrowtriangle} \mathtt{B} : a_1 \,.\,
\mathtt{A} \mathbin{\rightarrowtriangle} \mathtt{C} : a_2 \,.\,
\mathtt{\color{blue}B} \mathbin{\rightarrowtriangle} \mathtt{\color{blue}A} : a_4 \,.\,
\mathtt{C} \mathbin{\rightarrowtriangle} \mathtt{A} : a_5 \,.\, {\tend}
\;\;\triangleright\; \;
\mathtt{\color{blue}B} \mathbin{\rightarrowtriangle} \mathtt{\color{blue}A} : \mathtt{TOa} \,.\,
\mathtt{B} \mathbin{\rightarrowtriangle} \mathtt{\color{blue}C} : \mathtt{TOc} \,.\, \tend
$~~\textcolor{red}{\ding{55}}}
\smallskip

\noindent
However, the below (dropping the $a_4$) is accepted; \CODE{A} can safely commit
on the LHS because it has a transitive causal dependency with observer \CODE{B}
via \CODE{C}.

\smallskip
\centerline{$
\mathtt{A} \mathbin{\rightarrowtriangle} \mathtt{B} : a_1 \,.\,
\mathtt{A} \mathbin{\rightarrowtriangle} \mathtt{C} : a_2 \,.\,
\mathtt{\color{blue}B} \mathbin{\rightarrowtriangle} \mathtt{\color{blue}C} : a_3 \,.\,
\mathtt{C} \mathbin{\rightarrowtriangle} \mathtt{\color{blue}A} : a_5 \,.\, \tend
\;\;\triangleright\; \;
\mathtt{\color{blue}B} \mathbin{\rightarrowtriangle} \mathtt{\color{blue}A} : \mathtt{TOa} \,.\,
\mathtt{B} \mathbin{\rightarrowtriangle} \mathtt{\color{blue}C} : \mathtt{TOc} \,.\, \tend
$~~\textcolor{green}{\ding{51}}}
\smallskip

Altogether, the design of our asynchronous MC, the concepts of observer and
commitment and the protocol validation guarantee that every role that has not
safely terminated will continue progressing through the protocol as expected.
As an advance pointer, the main conditions checked by the protocol validation
(which we refer to as \emph{awareness}) and the \emph{progress} property for
global types are formally defined and established in \Cref{sec:globalprogress}.

On the expressiveness of \mMST{}, consider again the criss-crossing pattern
between \CODE{A} and \CODE{B} in the lower chart of \Cref{fig:timeout} (right).
By contrast, the conservative syntactic structure of classical
MST~\cite{DBLP:journals/mscs/CoppoDYP16,DBLP:journals/jacm/HondaYC16}
implicitly precludes all such patterns where asynchronous messages criss-cross
between a pair of roles in opposing directions.
The complications introduced by these patterns -- such as the inherent race
conditions, stale messages, and the mechanisms required to resolve these -- are
why reasoning about safety of mixed choices in asynchronous MST (and DS in
general!) is a difficult challenge, and has remained an open problem for MST.

We have used the small \CODE{Timeout} example to introduce our MC and the key
concepts.
Nevertheless, it demonstrates how \mMST{} supports the fundamental
communication pattern underlying various important DS constructs, including
exceptions, interrupts and failure handling, as found in many real-world
applications.
Our MC provides a \emph{core} building block for expressing these constructs in
MST.
\Cref{sec:globaltypes,sec:evaluation} include and summarise a range of further
such examples that involve combining MC with the standard directed choice of
MST, recursion and nested MCs.
%

\subsection{Mixed Choice MST protocols in Erlang}
\label{sec:moti}

Erlang is a concurrent, dynamically typed functional programming language
designed for implementing applications out of message passing processes.
It is used in major applications and platforms such as WhatsApp and RabbitMQ.
Its emphasis on building fault-tolerant distributed systems makes it a good
target for applying and testing \mMST{}.

Erlang/OTP has built-in support for a set of core design patterns, known as
\emph{behaviours}, for implementing processes.
In this paper, we target the \CODE{gen\_statem} behaviour that provides a
generic framework for implementing state machines.
A process is formed by combining the provided generic state machine behaviour
module with a user-written module of \emph{callback} functions for handling
events and state transitions according to the required application-specific
logic.

Based on our theory, we have implemented a toolchain for specifying and
implementing \mMST{} protocols as \CODE{gen\_statem} processes.
Our toolchain involves these steps.

\begin{itemize}[leftmargin=*]
\item
The user writes the source \emph{global} protocol (e.g., \CODE{Timeout} in
\Cref{fig:timeout}) using our \mMST{} extension of Scribble and uses the tool
to validate it based on our formal conditions.
\item
The tool internally \emph{projects} valid global protocols to a \emph{local}
protocol for each role.
The tool represents a local protocol as an \emph{Event-Driven Finite State
Machine} (EFSM) that matches the programming abstractions of \CODE{gen\_statem}
(see below).
\item
From each EFSM, the tool generates (i) a correct-by-construction
\emph{protocol- and role-specific} \CODE{gen\_statem} behaviour module, and
(ii) a corresponding template callback module for the programmer to use and
adapt as required to complete the process definition.
\end{itemize}

\paragraph{Global-to-local projection}
A local protocol is the view of the protocol from a specific role.
Projecting a global protocol to a set of local protocols, one for each role,
yields a distributed model.
In our formal notation, the projected \emph{local types} for \CODE{Timeout}
are:

\centerline{$
\begin{array}{lllll}
\mathtt{Timeout} \proj \mathtt{A} = &
\mathtt{B}\oplus a_1 \,.\,
\mathtt{C}\oplus a_2 \,.\,
\phantom{\mathtt{B}\,\&\, a_3 \,.\,}
\mathtt{B}\, \& \, a_4 \,.\,
\mathtt{C}\, \& \, a_5 \,.\,
\tend
&\triangleright&
\mathtt{B} \,\& \,\mathtt{TOa} \,.\,
\phantom{\mathtt{B} \oplus \mathtt{TOc} \,.\,}
\tend
\\
\mathtt{Timeout} \proj \mathtt{B} = &
\mathtt{A}\, \& \, a_1 \,.\,
\phantom{\mathtt{A}\, \& \, a_2 \,.\,}
\mathtt{C}\oplus a_3 \,.\,
\mathtt{A}\oplus a_4 \,.\,
\phantom{\mathtt{C}\, \& \, a_5 \,.\,}
\tend
&\triangleright&
\mathtt{A} \oplus \mathtt{TOa} \,.\,
\mathtt{C} \oplus \mathtt{TOc} \,.\,
\tend\\
\mathtt{Timeout} \proj \mathtt{C} = &
\phantom{\mathtt{A}\, \& \, a_1 \,.\,}
\mathtt{A}\, \& \, a_2 \,.\,
\mathtt{B}\, \& \, a_3 \,.\,
\phantom{\mathtt{B}\, \& \, a_4 \,.\,}
\mathtt{A}\oplus \, a_5 \,.\,
\tend
&\triangleright&
\phantom{\mathtt{B} \oplus  \mathtt{TOa} \,.\,}
\mathtt{B} \, \& \,  \mathtt{TOc} \,.\,
\tend
\\
\end{array}
$}
\smallskip

\noindent The $\triangleright$ separates the LHS and RHS of the projected mixed
choices.
The $\oplus$ denotes an output-only choice (internal select), and $\&$ denotes the
dual input-only choice (external branch); in this simple example, the
input/output-only choices are all unary.
\Cref{sec:correspondence} proves that the behaviour of a projected local
system corresponds to that of the global protocol.
Projection thus entails that a protocol is realisable as a
distributed system.

\paragraph{EFSM representation}

The core abstraction of \CODE{gen\_statem} is an \emph{Event-Driven Finite
State Machine} (EFSM).
Transitions are described
by\footnote{\url{https://www.erlang.org/doc/system/statem.html}}

\smallskip
\centerline{State (S) $\times$ Event (E) $\rightarrow$ Action (A) $\times$ State (S')}
\smallskip

\noindent
which can be read: if the \emph{current state} is S and \emph{event} E occurs,
then perform \emph{action} A and transition to \emph{successor state} S' -- i.e.,
the transition is triggered by E and A is a consequent effect.
Events may have external sources, such as the arrival of a message ($?$), or
internal ($\tau$), such as a local computation.
Actions include sending ($!$) messages and the empty action ($\epsilon$).
In our \mMST{} setting, a transition can also have the effect of
switching from the LHS of an MC to commit to the RHS.

\begin{figure}[t]
\small
\begin{tabular}{cc}
\begin{minipage}{0.67\textwidth}
{\begin{tikzpicture}
[node/.style={draw,circle,inner sep=0pt,minimum width=3mm,font=\scriptsize}]
\node[node,label=above:{\CODE{A}}] (S1) {$1$};
\node[node,below=6mm of S1] (S2) {$2$};
\node[node,below=6mm of S2] (S3) {$3$};
\node[node,below=6mm of S3] (S4) {$4$};
\node[node,below=6mm of S4] (S5) {$5$};
\draw[->] (S1) -- node[left,pos=0.3]{$\tau/\mathtt{B}!a_\mathtt{1}$} (S2);
\draw[->] (S2) -- node[right,pos=0.3]{$\tau/\mathtt{C}!a_\mathtt{2}$} (S3);
\draw[->] (S3) -- node[right,pos=0.3]{$\mathtt{B}?a_\mathtt{4}$} (S4);
\draw[->] (S4) -- node[pos=0.3,right]{$\mathtt{C}?a_\mathtt{5}$} (S5);
\draw[->] (S1) edge[bend left=80] node[pos=0.05,right, outer sep=5pt]{$\mathtt{B?^*TOa}$} (S5);
\draw[->] (S2) edge[bend right=80] node[pos=0.05,left,yshift=2pt]{$\mathtt{B?^*TOa}$} (S5);
\draw[->] (S3) edge[bend right=80] node[pos=0.05,left,yshift=2pt]{$\mathtt{B?^*TOa}$} (S5);
\node[node,right=28mm of S1,label=above:{\CODE{B}}] (B1) {$1$};
\node[node,below=7mm of B1,xshift=-6mm] (B2) {$2$};
\node[node,below=7mm of B2] (B3) {$3$};
\node[node,below=7mm of B3,xshift=6mm] (B4) {$4$};
\draw[->] (B1) -- node[pos=0.3,left]{$\mathtt{A?}a_\mathtt{1}$} (B2);
\draw[->] (B2) -- node[pos=0.3,left]{$\tau/\mathtt{C!}a_\mathtt{3}$} (B3);
\draw[->] (B3) -- node[left]{$\tau/\mathtt{A!}a_\mathtt{4}$} (B4);
\path (B2) -- coordinate[midway](C) (B3);
\node[node] at ($(C)+(18mm,0)$) (B5) {$5$};
\draw[->] (B1) edge[bend left=38] node[pos=0.05,right,outer sep=6pt]{$\mathtt{A?}a_\mathtt{1}/\mathtt{A!^*TOa}$} (B5);
\draw[->] (B1) -- node[pos=0.8,left,xshift=2pt]{$\mathtt{\tau/A!^*TOa}$} (B5);
\draw[->] (B5) -- node[right]{$\tau/\mathtt{C!TOc}$} (B4);
\node[node,right=30mm of B1,label=above:{\CODE{C}}] (C1) {$1$};
\node[node,below=7mm of C1,xshift=-2mm] (C2) {$2$};
\node[node,below=7mm of C2] (C3) {$3$};
\node[node,below=7mm of C3,xshift=2mm] (C4) {$4$};
\draw[->] (C1) -- node[left,pos=0.3]{$\mathtt{B?}a_\mathtt{2}$} (C2);
\draw[->] (C2) -- node[left,pos=0.3]{$\mathtt{A?}a_\mathtt{3}$} (C3);
\draw[->] (C3) -- node[left,pos=0.3]{$\tau/\mathtt{A!}a_\mathtt{5}$} (C4);
\draw[->] (C1) edge[bend left=85] node[right,pos=0.05,outer sep=6pt]{$\mathtt{B?^*TOc}$} (C4);
\draw[->] (C2) edge[bend left=30] node[right,pos=0.05]{$\mathtt{B?^*TOc}$} (C4);
\node[below=7mm of C3] {\phantom{$4$}};
\end{tikzpicture}}
\end{minipage}
&
\begin{minipage}{0.35\textwidth}
\begin{tabular}{ll}
Transition labels:\\\quad Event `/' Action
\\
Event:\\
\quad`$\tau$' internal\\
\quad`$?$' external msg.\ arrival
\\
Action:\\
\quad`$!$' send $\;|\;$ `$\epsilon$' empty
\\
\\
$?$-only labels are short for $?$ / $\epsilon$
\\
\\
${}^*$ means switch to RHS of MC
\\
\end{tabular}
\end{minipage}
\\
\end{tabular}
\caption{Local \mMST{} protocols for each role of \CODE{Timeout} as EFSMs in
Erlang \CODE{gen\_statem}.}
\label{fig:efsms}
\end{figure}

\Cref{fig:efsms} depicts the EFSM generated by our tool for each role
of \CODE{Timeout}.
The notation $e/\alpha$ on the transitions means event $e$ triggers the
transition and action $\alpha$ is performed; an $e$ on its own is shorthand
for $e/\epsilon$.
For example, the $\tau/\mathtt{B}!a_\mathtt{1}$ from state $1$ of $\mathtt{A}$
means that $\mathtt{A}$ can internally decide to send $a_\mathtt{1}$ to
$\mathtt{B}$ (and transition to $2$).
In $\mathtt{B}$, the two transitions from state $1$ with event
$\mathtt{A}?a_\mathtt{1}$ mean that if message $a_\mathtt{1}$ arrives as an
external event, $\mathtt{B}$ can either do $\mathtt{A}!^*\mathtt{TOa}$ and
transition to $4$, or transition to $2$ (with the empty action).
Note, the ${}^*$ annotation (as also occurs in the text below) indicates the transitions that switch from the LHS to the
RHS of the MC.

Classical MST permits only three kinds of local states: input-only (branch),
output-only (select), and terminal; e.g., state $4$ in $\mathtt{A}$ is a
branch, and $3$ in $\mathtt{C}$ is a select (again, both unary in this simple
example).
By contrast, the more expressive local behaviours enabled by our asynchronous
MC can also be seen in the above EFSMs.
As mentioned, observer \CODE{B} faces an \emph{internal MC} in state~$1$
indicated by the \emph{mix} of $\tau/!^*$ and $?/\epsilon$ transitions induced
by its RHS and LHS, respectively, yielding the timeout behaviour described at
the start of \Cref{sec:asymmMC}.
On the other hand, \CODE{A} faces an \emph{external MC} in state $1$ as
indicated by its mix of $?^*/\epsilon$ and $\tau/!$ transitions: \CODE{A} can
either (wait to) receive the \CODE{TOa} (and switch to the RHS), or internally
decide (by some internal event $\tau$) to send $a_\mathtt{1}$ and proceed to
$2$.
Our formal semantics allows both options, but an implementation may (e.g.) give
priority to the former case if $\mathtt{TOa}$ has already arrived.
Note how \CODE{A} is forced to eventually commit to either the LHS or RHS (in
accordance with \CODE{B}) by state $3$ at the latest by receiving
$a_\mathtt{4}$ or $\mathtt{TOa}$.

We give a couple of further comments on the richer behaviours expressed by our
MC.
First, note that external MCs may take various forms: e.g., both states $1$ in
\CODE{A} and $1$ in \CODE{C} are external MCs featuring $?^*$, but the former
is mixed with a $\tau$ whereas the latter is mixed with an $?$-event.
Second, note that the top-level roles of the MC are \CODE{A} and \CODE{B}, but
their localised (projected) behaviours in the MC are not direct duals; e.g.,
state $3$ in \CODE{A} is an external MC between $a_4$ and $\mathtt{TOa}$ that
has no \emph{direct} counterpart in \CODE{B} (i.e., there is no internal MC
featuring $a_4$ in \CODE{B}).
Additionally, $?^*\mathtt{TOa}$ events are spread over states $1$, $2$ and $3$
in \CODE{A}; whereas in \CODE{B}, the $!^*\mathtt{TOa}$ occurs only in $1$,
although in two different transitions since \CODE{B} can choose to raise the
timeout either before ($\tau$) or after ($?$) the $a_\mathtt{1}$ has arrived.
These behaviours are considerably more exotic than those supported in classical
MST and speak to the challenge of reasoning about asynchronous MC theoretically
and practically.

\paragraph{Programming \mMST{}-based processes in Erlang.}

Our toolchain uses the EFSMs to generate two Erlang modules per role for the
programmer(s) to work with.
The \emph{role module} (RM) provides a customised \CODE{gen\_statem} behaviour
that is specialised to the source protocol and role.
The programmer can consider it as part of the runtime and does not need to use
it directly.
The other, called the \emph{callback module} (CM), is a template module of
callback functions that the RM delegates to for handling event
occurrences and performing state transition actions.
The CM is generated with minimal placeholder code that the
programmer can modify and extend with the required application-specific logic.
The RM and CM together form a protocol- and role-specific \mMST{} process.
A set of processes comprising one for each role forms a complete application.

\begin{figure}
\begin{tabular}{ll}
\begin{minipage}{0.5\textwidth}
{\begin{lstlisting}[numbers=left,language=Erlang, numbersep=3pt, basicstyle=\ttfamily\scriptsize]
% Extract from role A
s1(internal, a1, Data) ->
  gen_a:send_s1_a1(BPid, Data), {next_state, s2,
    Data, [{next_event, internal, a2}]};
s1(cast, {BPid, 'TOa'}, Data) -> {stop, normal, Data}.

s2(internal, a2, Data) ->
  gen_a:send_s2_a2(CPid, Data), {next_state, s3, Data};
s2(cast, {BPid, 'TOa'}, Data) -> {stop, normal, Data}.

s3(cast, {BPid, a4}, Data) -> {next_state, s4, Data};
s3(cast, {BPid, 'TOa'}, Data) -> {stop, normal, Data}.

s4(cast, {CPid, a5}, Data) -> {stop, normal, Data}.


% Extract from role B
s1(internal,'TOa', Data) ->
case make_choice_TOa(Data) of
  1 -> {keep_state, Data};
  2 -> gen_b:send_s1_TOa(APid, Data), {next_state, s3,
         Data, [{next_event, internal,'TOc'}]}
end;  % Continued on the right
\end{lstlisting}}
\end{minipage}
&
\begin{minipage}{0.5\textwidth}
{\begin{lstlisting}[firstnumber=24,numbers=left,language=Erlang, numbersep=3pt, basicstyle=\ttfamily\scriptsize]
% Continued from the left
s1(cast,{APid, a1},Data) ->
  case make_choice_a1(Data) of
    1 -> {next_state, s2, Data,
           [{next_event, internal, a3}]};
    2 -> gen_b:send_s1_TOa(APid, Data),
           {next_state, s5, Data,
             [{next_event, internal,'TOc'}]}
  end.

s5(internal, 'TOc', Data) ->
  gen_b:send_s5_TOc(CPid, Data),
    {stop, normal, Data}.

s2(internal, a3, Data) ->
  gen_b:send_s2_a3(CPid, Data), {next_state, s3,
    Data, [{next_event, internal, a4}]}.

s3(internal, a4, Data) ->
  gen_b:send_s3_a4(APid, Data),
    {stop, normal, Data}.
$~$
$~$
\end{lstlisting}}
\end{minipage}
\\
\end{tabular}
\caption{Extracts from the user-facing callback modules generated for \CODE{A} and \CODE{B}
in the \CODE{Timeout} protocol.}
\label{fig:exceptionB}
\end{figure}

\Cref{fig:exceptionB} shows code from the CMs as generated for roles \CODE{A} and \CODE{B}.
As per their respective EFSMs (\cref{fig:efsms}), they feature callback \emph{state functions} (e.g., \CODE{s1}, \CODE{s2}, etc.)
for handling state transitions according to event occurrences.
The callbacks are fired (from the RM) based on the current state and
by pattern matching against the event, which may come from an internal
(locally triggered transitions, such as timeout decisions) or external
(incoming messages) source.

For instance, in \CODE{A}, state function \CODE{s1} represents an external MC corresponding to state 1 (for \CODE{A} in \cref{fig:efsms}):
\CODE{A} may either send \CODE{a1} to \CODE{B} (the \CODE{internal} clause), or receive \CODE{`TOa'} from \CODE{B} (the \CODE{cast} clause), corresponding to the $\mathtt{\tau/B!a_1}$
and $\mathtt{B?^*TOa}$ transitions, respectively. 
In \CODE{B}, the function \CODE{s1} represents an internal MC.
In the first clause (\CODE{internal}), \CODE{B} handles the internal event from
the generated template function \CODE{make\_choice\_TOa} by which it decides
whether to remain in the current state (\CODE{keep\_state}) and wait to receive
\CODE{a1}, or transition to \CODE{s2} by sending \CODE{`TOa'} to \CODE{A} (cf.\
$\mathtt{\tau/A!^*To}$).
%
Generally, the programmer will modify/replace such template decision functions
according to the required application logic.
In the second clause (\CODE{cast}), \CODE{B} receives the message \CODE{a1}
from \CODE{A}, and based on template function \CODE{make\_choice\_a1(Data)}, either
transitions (cf.\ $\mathtt{A}?a_\mathtt{1}\mathtt{/A!^*TOa}$) to
\CODE{s3} and schedules an internal event \CODE{a2}, or sends \CODE{`TOa'} to
\CODE{A} and transitions (cf.\ $\mathtt{A}?a_\mathtt{1}$) to
\CODE{s2}.

\Cref{sec:implementation} demonstrates the internals of the RM that
fires the callbacks in the above CM.
In short, our tool generates the RM by instantiating the default
\CODE{gen\_statem} behaviour with the required EFSM structure (states, events,
actions, and transitions); this is correct-by-construction and the programmer
should not modify the RM.
Following our formal semantics, the RM is also generated to
encapsulate the (again, correct-by-construction) runtime mechanisms for
handling MC commitment, LHS-to-RHS switching, and stale message purging
specifically for the source protocol and role.
The programmer can assume these mechanisms are provided and correct when
working on the CM.

Recall the example executions in \Cref{fig:timeout} (right).
The grey annotations denote the current state and relevant event/action at each
role according to their EFSMs in \Cref{fig:efsms}; e.g., \GREY{$1,
!a_\mathtt{1}$} means in state $1$, send message $a_\mathtt{1}$.
%
%
Stale message purging, denoted by \GREY{$\mathsf{gc}$} in \Cref{fig:timeout},
is handled internally by the RM; we exclude purge actions from our
depictions of EFSMs as they are transparent to the programmer.
User processes do not consume stale messages, so they must be purged from the
input buffer to prevent interference with future receives (given the FIFO
nature of inputs).

\subsection{Properties of \mMST{}}

We end this overview by summarising the properties of our framework with
advance pointers (a roadmap) to the relevant parts of our formal theory.

\Cref{sec:globaltypes} formalises the syntax and metatheoretical LTS semantics
of our \emph{global types} with mixed choice and the notion of
\emph{committing} messages in MCs.
It defines the conditions checked by protocol validation on protocol structure
(\emph{well-formedness} of committing messages, \emph{balance} of roles across
choice branches) and on inter-role dependencies in MCs (\emph{awareness}).
\Cref{sec:globalprogress} proves that valid protocols enjoy a per-role
\emph{progress} property and that their roles are always consistent in their
commitments (\emph{coherence}).

\Cref{sec:localtypes} formalises the syntax and LTS semantics of distributed
\emph{local types} and asynchronous message queues, the \emph{projection} from
global to local types, and the runtime mechanism for \emph{stale message
purging}.
\Cref{sec:correspondence} proves an operational correspondence in both
directions between valid projectable global types and local systems that
\PURPLE{\emph{preserves projection}}.
The operational correspondence (i.e., preservation of projection) and progress
properties together safely entail that local roles never get stuck in a
deadlock nor (noting the FIFO nature of communications in the local type LTS)
due to receiving an unexpected message.
\Cref{sec:localproperties} further proves orphan message freedom.

As discussed, our practical toolchain is implemented to perform protocol
validation (conservatively, see \Cref{lab:toolchainimpl}) and projection
following our formal theory.
It is implemented to perform a correct-by-construction translation from local
projections to EFSMs and generation of \CODE{gen\_statem} modules to transfer
the above correctness properties to Erlang processes.
The usage contract is that the programmer should not modify the generated RM
and must use the CM according to the generated structures.
The RM is generated with some internal runtime checks against the
programmer supplying an invalid CM.

\section{Global Types}
\label{sec:globaltypes}

The syntax of global types $\GT$ is defined by the grammar below:
\[\begin{array}{lll}
  \GT  &::= & \GactShort{p}{q}{\St{}}
     \mid  \GmsgShort{p}{q}{k}{\St{}}
     \mid  \mu \rv . \GT
     \mid  \rv
    \mid \tend
    \mid \GtoDef{\GactShort{q}{p}{\St{1}}}{c}{\lset}{\rset}{\GactShort{p}{q}{\St{2}}}
     \mid  \GtoP{\GT_1}{c}{\lset}{\rset}{\GT_2}
    \\
    \St{} & ::= &  \{a_i.\GT_i\}_{i \in I}
\end{array}
\]
The first five terms are standard~\cite{DenielouY13}. We recall that $\GactShort{p}{q}{\St{}}$ is an interaction where $\role p$ is the sender and $\role q$ is the receiver. Term $\St{}$ specifies a set of choices: $\p$ can send $\q$ one of the labels $a_i$ (for $i\in I$) and the protocol continues as $\GT_i$.
The message in-transit type $\GmsgShort{p}{q}{k}{\St{}}$ stands for a state where $\role p$ has sent label $a_k$ but $\role q$ has not yet received it (asynchronous communication). Term $\mu \rv.\GT$ is a recursive definition\footnote{We adopt iso-recursive types for a lower level view of nested MC instantiation and straightforward implementation.} and $\rv$ a recursive variable.

The last two terms are new and model \emph{mixed choices} (MC).
The term $\GtoDef{\GactShort{q}{p}{\St{1}}}{c}{\lset}{\rset}{ \GactShort{p}{q}{\St{2}}}$ is an \emph{MC definition}: we call $\GactShort{q}{p}{\St{1}}$ the left-hand side block (LHS) of the MC and $\GactShort{p}{q}{\St{2}}$ the right-hand side block (RHS).
Role $\role p$ is set to receive a message from $\role q$ in the LHS, but may decide to execute the RHS by sending a message to $\role q$ instead. We say that $\role p$ is the \emph{observer} of the MC.
We annotate MC definitions with a unique name $c$, which we may omit when it is not important.

The term $\GtoP{\GT_1}{c}{\lset}{\rset}{\GT_2}$ is an \emph{active MC}. To ensure that roles eventually agree on which block (LHS or RHS) to execute, we use a notion of \emph{commitment}.
Some actions are designated as committing  (defined in \Cref{sec:committing}).
When a role executes a committing action within a block, it commits to that block, meaning that it can no longer perform any action (send or receive) in the other block. The semantics of global types use sets $\lset$ to keep track of the roles that are committed to the LHS, and $\rset$ to keep track those committed to the RHS.
For convenience we may use the notation $\GtoDef{ \GT_1}{c:\p}{\lset}{\rset}{ \GT_2}$
and
$\GtoP{ \GT_1}{c:\p}{\lset}{\rset}{ \GT_2}$
assuming that $\GT_1$ and $\GT_2$ are of the form given by the grammar and that $\role p$ is the observer. For readability, we omit annotations $c$,
$\lset$ or $\rset$ when not needed.
A global type is \emph{initial} if it has no messages in transit and no active MC.

The set of roles of $\GT$, denoted $\RolesG{\GT}$, is defined as usual except for the two new cases for MC:
\[
\RolesG{\GtoDef{  \GT_1}{}{\lset}{\rset}{  \GT_2}} 	=
\RolesG{  \GT_1}~ \cup ~\RolesG{  \GT_2}
\qquad
\RolesG{\GtoP{\GT_1}{}{\lset}{\rset}{\GT_2}
} 	=
\RolesG{  \GT_1}\setminus\rset~ \cup ~\RolesG{  \GT_2}\setminus\lset
\]
The roles of a MC definition are the roles in either of its sides. The case of active MC excludes the roles that have committed to the opposite side. This will be critical when defining progress, to characterise the roles that should continue on each side.
We define a context environment $\Ctname$:
\[
\Ctname :: =  \hole
 \mid  \GactShort{p}{q}{\St{}}\cup\{a.\Ctname\}
\mid   \GmsgShort{p}{q}{k}{\St{}}\cup\{a_k.\Ctname\}
 \mid  \Ctname \lmix \GT
 \mid  \GT \lmix \Ctname
 \mid  \GtoDef{\Ctname}{}{}{}{\GT}
 \mid  \GtoDef{\GT}{}{}{}{\Ctname }
 \mid  \mu \rv . \Ctname
\]
 We say that $\GT'$ is a \emph{subterm} of $\GT$ (or \emph{is in} $\GT$) if there exists $\Ctname$ such that $\GT = \Ctname[\GT']$. We say that two subterms of $\GT_1$ and $\GT_2$ of $\GT$ are \emph{distinguished} if $\GT = \Ctname_1[\GT_1] = \Ctname_2[\GT_2]$
implies $\Ctname_1\neq\Ctname_2$. We say that $\GT$ \emph{has an active MC}
$\GT_1\blacktriangleright\GT_2$
if $\GT_1\blacktriangleright\GT_2$ is a subterm of $\GT$.
Similarly for MC definitions.

\subsection{Committing Set and Well-Formedness}
\label{sec:committing}

In this section we formally introduce the notion of commitment, by defining the
\emph{committing set} of a MC named $c$, that is the set of actions by which
some role commits to either the LHS or to the RHS of $c$.
The committing set is defined on initial global types.
%
%
Intuitively, the committing actions of a MC are: (1) on the LHS, the first receive action by the observer, and all receive actions of a message from a committed sender, (2) on the RHS, all the first actions of a role on that side.
For simplicity, we identify committing actions (e.g., $\p\q?a$) using only
their labels (e.g., $a$). The committing set is therefore, a set of
communication labels.

In the following we denote with $\unfold{\GT}$ the standard \emph{unfold-all-once} operation that unfolds once all recursive types in $\GT$. We say that $\GT_1 \mixi{c} \GT_2$ is the \emph{outermost occurrence} of $c$ in $\unfold{G}$ if it is not a subterm of other MC definitions $c$ in $\unfold{G}$.
Given a global type $\GactShort{q}{p}{\St{}}$ with $\St{}= \{a_i.\GT_i\}_{i\in I}$, we use the notation $\labels{\St{}}=\{a_i\}_{i\in I}$ and $\types{\St{}}=\{\GT_i\}_{i\in I}$.

\begin{definition}[Committing set]\label{def:committ}
Let $\GT$ be an initial global type. Let $c$ be an MC definition in $\GT$, and ${\GactShort{q}{p}{\, \St{1}}\mixi{c}\,\GactShort{p}{q}{\, \St{2}}}$ be the outermost occurrence of $c$ in $\unfold{\GT}$. The committing set of $c$ in $\GT$, denoted $\comSetGc{\GT}{c}$, is defined as follows:
\[\begin{array}{l}\comSetGc{\GT}{c}~\coloneq~ \labG{\St{1} \,\cup\, \St{2}} ~ ~ \cup ~~\bigcup_{G\in \tyG{\St{1}}}\comSet{G}{ \{\role p\}}{c}
~ ~\cup ~~\bigcup_{G\in \tyG{\St{2}}}
 \comSet{G}{\{\role p, \role q\}}{c} \end{array}\]
where $\comSet{\GT}{C}{C}$ is an auxiliary function parameterized by the set $C$ of committed roles, used to track dependency with other committing actions.
In brief, the case for interactions is:

$\comSet{\GactShort{p}{q}{\St{}}}{C}{c} = \begin{cases}
\labG{\St{}} ~\cup  \bigcup_{\GT \in \tyG{\St{}}}  \comSet{\GT}{C\cup\{\role q\}}{c}
& (\p \in C   ~\land~\q\not \in C)
\\
\bigcup_{\GT\in \tyG{\St{}}}\comSet{\GT}{C}{c}
& \text{otherwise}
\end{cases}$

\noindent
The case $\comSet{\GtoDef{\GT_1}{c'}{\lset}{\rset}{\GT_2}}{C}{c}$ 
returns $\comSet{\GT_1}{C}{c}~ \cup  \comSet{\GT_2}{C}{c}$ if $c'\neq c$, and the empty set otherwise. Finally, $\comSet{\mu \mathtt t .\GT}{C}{c}\mathop{=}\comSet{\GT}{C}{c}$ and $\comSet{\mathtt t}{C}{c} \mathop{=} \comSet{\tend}{C}{c} \mathop{=} \emptyset$.
%
 %
     %
\end{definition}

Intuitively, $\comSetGc{\GT}{c}$ identifies the labels of $\St{1}\cup\St{2}$ as the set of initial committing labels in $c$: the labels in $\St{1}$ are committing for $\p$ and labels in $\St{2}$ are committing for both $\q$ and $\p$.
To these labels, we add  those that depend on actions by $\p$ on the RHS of $c$ and $\{\p,\q\}$ on the LHS, by using the auxiliary function.
The function $\comSet{\GT}{C}{c}$ traverses the syntactic structure of $\GT$ until it reaches $\mathtt t$, $\tend$, or a nested occurrence of the same $c$ (introduced by unfolding).
In case $\comSet{\GactShort{p}{q}{\St{}}}{C}{c}$ all labels that are sent from a committed role $\p\in C$ to a non-committed role $\q\not\in C$ are regarded as committing.
Unfolding $\unfold{\GT}$ is used to account for committing labels that are captured into a MC block only after recursive unfolding, as we illustrate in \Cref{ex:unf}.

\begin{example}[Label capture]\label{ex:unf}
Consider the recursive type $\GT_{c} = \mu \mathtt t.\, \GactShort{p}{{ r}}{{\color{magenta}a}. ( \GactShort{q}{p}{b. \, \mathtt t}  \mix \GT_2)}$ where label $a\mathop{\not\in} \GT_2$ and hence does not appear directly in the MC. If we apply the auxiliary function in \Cref{def:committ} to $\GT_c$ itself (rather than its unfolding $\unfold{\GT_c}$) then $a$ would not be included in the committing set of $\GT_c$. However, after some steps that involve recursive unfolding, shown in (\ref{unf}), an occurrence of $a$ is introduced on the LHS of the outer MC. This second occurrence of $a$ \emph{is} committing (it makes $\role r$ commit to the LHS of the outer MC). By using unfolding, \Cref{def:committ} correctly identifies $a$ as committing.
\begin{equation}\label{unf}
\GT_{c}~~~\myreaches{} ~~~
\GactShort{p}{{ r}}{{\color{magenta}a}. ( \GactShort{q}{p}{b. {\color{blue}(\mu \mathtt t.\, \GactShort{p}{{ r}}{{\color{magenta}a}. ( \GactShort{q}{p}{b. \, \mathtt t}  \mix\GT_2)})}} \lmix  \GT_2)}
\end{equation}
\end{example}

\paragraph{Well-formedness}
Global types where a label occurs multiple times, and where some occurrences are committing and some others are not, are problematic. A simple example is (\ref{unf1}), where $b$ appears twice, and is non-committing on the LHS of $c$ while it is committing on the RHS.
\begin{equation}\label{unf1}
   \GT_\textrm{\ding{55}} = \GactShort{q}{p}{a}.\,  \GactShort{q}{r}{{\color{magenta}b}.\, \GT_1}  \mixi{c}  \GactShort{p}{q}{c}. \, \GactShort{q}{r}{{\color{magenta}b}.\, \GT_2}
\end{equation}
It is critical that we avoid ambiguities such as those in (\ref{unf1}), as they may lead to incorrect semantics for mixed choices. To address this, we introduce a well-formedness condition on initial global types, requiring that all occurrences of a label are exclusively either committing or non-committing.

Well-formedness relies on the notion of committing set (\Cref{def:committ}) and on the dual notion  of \emph{non-committing set} (\Cref{def:ncommitt}  below). Well-formedness is formally defined in \Cref{well-formed} by requiring the committing and non-committing sets to be disjoint.

\begin{definition}[Non-committing set]\label{def:ncommitt}
Let $\GT$ be an initial global type, and let $c$ be a MC definition in $\GT$, and ${\GactShort{q}{p}{\, \St{1}}\mixi{c}\GactShort{p}{q}{\, \St{2}}}$ be the outermost occurrence of $c$ in $\unfold{\GT}$. The non-committing set of $c$ in $\GT$, denoted $\ncomSetG{\GT}$, is defined as follows, relying on auxiliary function $\newcomSet{\GT}{C}{C}$ :
\[ \begin{array}{l}\ncomSetG{\GT}~\coloneq ~\bigcup_{G\in \tyG{\St{1}}} \newcomSet{G}{ \{\role p\}}{c}
~ \cup ~\bigcup_{G\in \tyG{\St{2}}}
 \newcomSet{G}{\{\role p, \role q\}}{c} \end{array}\]
 In contrast to $\comSetG{\GT}$ in \Cref{def:committ}, $\ncomSetG{\GT}$ does not include $\labG{\St{1}}\cup\labG{\St{2}}$ in the non-committing set of $c$.
 
 Function $ \newcomSet{G}{C}{c}$ is defined like $ \comSet{G}{C}{c}$ from \Cref{def:committ} except the case for interaction types that is dual: specifically, $\newcomSet{\GactShort{p}{q}{\St{}}}{C}{c}$
returns
 $\newcomSet{\GT}{C\cup\{\role q\}}{c}$ if $(\p \in C   ~\land~\q\not \in C)$, and returns $\labG{\St{}} ~\cup ~  \bigcup_{\GT\in \tyG{\St{}}}\newcomSet{\GT}{C}{c}$ otherwise.
%
\end{definition}




\begin{definition}[Well-formedness]\label{well-formed}
Initial $\GT$ is \emph{well-formed} if for all $c$ in $\GT$, $\ncomSetG{\GT}{}{} \,\cap \comSetG{\GT}{}{} = \emptyset$.
\end{definition}

Since the syntactic structure of $\unfold{\GT}$ is finite, well-formedness yields a decidable algorithm based on the definitions of committing and non-committing sets. Hereafter we assume all initial global types to be well-formed.

\begin{example}[Well-formedness]\label{ex:unf1bad}
One can verify that the global type $\GT_\textrm{\ding{55}}$ in (\ref{unf1}) is not well-formed by observing that $b \in ~ \ncomSetG{\GT_\textrm{\ding{55}}}{}{}$ and  $b\in ~\comSetG{\GT_\textrm{\ding{55}}}{}{}$.
One can verify that $\GT_\textrm{\ding{51}}$ below is well-formed by observing that $\comSetG{\GT_\textrm{\ding{51}}}{}{}=\{a,d,e\}$, $\ncomSetG{\GT_\textrm{\ding{51}}}{}{}=\{b\}$, and $\{a,d,e\}\cap\{b\}=\emptyset$.
\begin{equation*}
\begin{array}{llll}
\GT_\textrm{\ding{51}} & = & \mu \mathtt t.\,
 ( \GactShort{q}{p}{{\color{black}a}. \,
\GactShort{q}{r}{{\color{black}b}}.
\, \mathtt t}  \mixi{c} \,\GactShort{p}{q}{{\color{black}d}}. \GactShort{p}{r}{{\color{black}e}}.\, \mathtt t )\\
\unfold{\GT_\textrm{\ding{51}}} & = &
\GactShort{q}{p}{{\color{black}a}}. \,
\,
\GactShort{q}{r}{{\color{red}b}}.\, {\GT_\textrm{\ding{51}}}
\mixi{c}  \GactShort{p}{q}{{\color{black}d}}.\, \GactShort{p}{r}{{\color{black}e}}.\,
{\color{blue}\mu \mathtt t.\,
 ( \GactShort{q}{p}{a}. \,
\GactShort{q}{r}{{\color{red}b}}.
\, \mathtt t }\mixi{c}
\color{blue}\ldots
)
\end{array}
\end{equation*}
This example is noteworthy: 
label $b$ occurs multiple times in $\unfold{\GT_\textrm{\ding{51}}}$ but all occurrences are consistent as they are all non-committing.
The occurrence of $b$ on the RHS of the outer $c$ is non-committing because the receiver $\role r$ is already committed to that block, having previously received committing label $e$.
\end{example}

Hereafter we assume all initial global types to be well-formed. We further require the \emph{knowledge of choice} assumption~\cite{lmcs:773},
ensuring every participant affected by a branching decision can
determine which branch was taken.

\subsection{Semantics of Global Types}
\label{sec:global_semantics}

The semantics of global types is defined as a Labelled Transition System over terms $\GT$ with labels
\[ \ell :=  \lsnd{p}{q}{a} \mid \lrcv{p}{q}{a} \mid \lnewlab\]
 $\lsnd{p}{q}{a}$ (resp. $\lrcv{p}{q}{a}$) denotes the sending (resp. receiving) action of  message $a$ from $\role p$ to $\role q$. Label $\lnewlab$ is for MC instantiation.
We define the subject of a label as the singleton set containing the role performing the action described by that label:
\(\sbj{\lsnd{p}{q}{a}} = \sbj{\lrcv{q}{p}{a}} = \{ \p \}\). In what follows, we may write $\p$ as a shorthand for the singleton $\{\p\}$.
The subject $\nlab$ is defined to be the empty set.
In the LTS we assume knowledge of the original initial (or \emph{base}) global type from which a given state is reached. We denote it with $\underline{\GT}$. This is akin to assuming knowledge of the original static protocol specification of an ongoing session.

The rules for the LTS are given in \cref{fig:gsem}.
The first set of rules for interactions is given in \cref{fig:gsem} (top), and is a minor adaptation of the semantics in \cite{DenielouY13}.

The remaining rules are new for MC.
 \glab{Inst} instantiates a MC. \glab{Ctx1} and \glab{Ctx2} handle nested instantiations. The conditions on the roles ensure that once a side is resolved (i.e., no role can act in it anymore), it cannot perform degenerate $\nlab$ transitions that would artificially break our correspondence results (\Cref{sec:correspondence}).
 In \glab{Ctx1}, $\GT_l$ is resolved when all roles are committed to the right. In \glab{Ctx2}, $\GT_r$ is  resolved when some roles commit to the left. The asymmetry in the conditions of \glab{Ctx1} and \glab{Ctx2} reflects the fact that commitment to the RHS does not immediately preclude actions on the LHS, until all participants have committed. Maintaining the MC contexts, including the resolved sides in the global types (which will be incorporated into the projected local types), is important for correctly characterising stale messages in the projected systems, and thus for ensuring our correspondence results.
 \glab{LSnd} allows a send action by the LHS if the subject $\role p$ is not committed to the RHS.  \glab{RSnd} is symmetric, except $\role p$ is added to $\rset$ (any action on the RHS is committing).
Rules \glab{LRcv1} and \glab{LRcv2} are for committing and non-committing receive actions on the LHS, respectively. In both cases, the subject $\role q$ must not be committed on the RHS. \glab{RRcv} is similar for receive actions in the RHS.

\begin{figure}
\small
\[
\begin{array}{cc}
\begin{array}{cc}
	%
    %
	 \gspairf{\Gact} \trans{\lsnd{p}{q}{a_k}} \Gmsg~~~  (k\in I)~~
    \mathtt{[Snd]}\\[0.3cm]
	\gspairf   \Gmsg \trans{\lrcv{p}{q}{a_k}} \GT_k  ~~~\mathtt{[Rcv]}~~~ \qquad
	   \inferrule{\forall i \in I \qquad \gspairf\GT_i \trans{\tlab}\GT'_i \qquad \role p, \role q \not \in \sbj{\tlab} 
        }{\gspairf\Gact \trans{\tlab} \GactP{p}{q}{a_i}{\GT'_i}{i\in I}}
    ~~\mathtt{[Cont1]}\\[0.7cm]
     \inferrule{\gspairf{\GT[\mu \rv. \GT / \rv]}\trans{\tlab}{\GT'} }{\gspairf\mu \rv . \GT \trans{\tlab} {\GT' }}~~\mathtt{[Rec]}
     \qquad~~
    \inferrule{ \gspairf\GT_k \trans{\tlab}{\GT'_k} \qquad \role q \not \in \sbj{\tlab} \qquad \forall i\in I\setminus k.\, \GT_i=\GT_i'
    }{\gspairf\Gmsg \trans{\tlab} \GmsgP{p}{q}{k}{a_i}{\GT'_i}{i\in I}}
    ~~
      \mathtt{[Cont2]}\\[0.7cm]
    \hline \\
\end{array}\\
%
%
%
%
%
\begin{array}{rr}
	{ \gspairf{\GtoDef{  \GT_1}{c}{\lset}{\rset}{  \GT_2}}
\trans{\nlab } {\GtoP{  \GT_1}{c}{\emptyset}{\emptyset}{  \GT_2}}} ~~\mathtt{[Inst]}
 &
 \inferrule{\gspairf \GT_l \trans{\nlab } {\GT_l'} \quad \RolesG{\underline{\GT}} \neq \rset
 }
	{ \gspairf{\GT = \GtoP{\GT_l}{}{\lset}{\rset}{\GT_r} }  \trans{\nlab } {\GtoP{\GT'_l}{}{\lset}{\rset}{\GT_r}}}
 ~~ \mathtt{[Ctx1]}\\[0.7cm]
 \inferrule{\gspairf\GT_r  \trans{\nlab } {\GT_r'}\quad \lset = \emptyset
 }
	{ \gspairf{\GT = \GtoP{\GT_l}{~}{\lset}{\rset}{\GT_r} }  \trans{\nlab } {\GtoP{\GT_l}{~}{\lset}{\rset}{\GT_r'}}}
 ~~ \mathtt{[Ctx2]}
&
	\inferrule{\gspairf \GT_l \trans{\lsnd{p}{q}{a}} {\GT_l'}
 \quad
 \color{olive}\role p \not \in \rset }
	{ \gspairf{\GtoP{\GT_l}{}{\lset}{\rset}{\GT_r}}\trans{\lsnd{p}{q}{a}} {\GtoP{\GT'_l}{}{\lset}{\rset}{\GT_r}}
    }
    ~~ \mathtt{[LSnd]}\\[0.8cm]
 \inferrule{\gspairf\GT_l \trans{\lrcv{p}{q}{a}} {\GT_l'} \qquad
 \color{olive}\role q \not \in \rset \qquad a\in \comSetG{
 \underline{\GT}}
 }
	{ \gspairf{\GtoP{\GT_l}{c}{\lset}{\rset}{\GT_r}} \trans{\lrcv{p}{q}{a}} {\GtoP{\GT'_l}{c}{\lset\cup\{\role q\}}{\rset}{\GT_r}}}~~ \mathtt{[LRcv1]}
 &
	\inferrule{\gspairf\GT_l \trans{\lrcv{p}{q}{a}} {\GT_l'}\qquad
 \color{olive}\role q \not \in \rset \qquad  a\not \in \comSetG{
\underline{\GT}}}
	{ \gspairf{\GtoP{\GT_l}{c}{\lset}{\rset}{\GT_r}} \trans{\lrcv{p}{q}{a}} {\GtoP{\GT'_l}{c}{\lset}{\rset}{\GT_r}}}~~\mathtt{[LRcv2]}\\[0.9cm]
	\inferrule{\gspairf\GT_r \trans{\lsnd{p}{q}{a}} {\GT_r'} \qquad
 \color{olive}\role p \not \in \lset }
	{ \gspairf\GtoP{\GT_l}{}{\lset}{\rset}{\GT_r} \trans{\lsnd{p}{q}{a}} {\GtoP{\GT_l}{}{\lset}{\rset\cup\{\role p\}}{\GT'_r}}}~~ \mathtt{[RSnd]}
&
\inferrule{\gspairf\GT_r \trans{\lrcv{p}{q}{a}} {\GT_r'} \qquad \color{olive}\role q \not \in \lset  }
  { \gspairf\GtoP{\GT_l}{}{\lset}{\rset}{\GT_r} \trans{\lrcv{p}{q}{a}} {\GtoP{\GT_l}{}{\lset}{\rset\cup\{\role q\}}{\GT'_r}}}~~  \mathtt{[RRcv]}
\end{array}
\end{array}
\]
\caption{Global semantics: standard rules (top) and new rules for MC (bottom)}
\label{fig:gsem}
\end{figure}

Let $\vv \ell=\ell_0, \ldots, \ell_n$ be a (possibly empty) vector. We write $\GT_0\trans{\vv \ell} \GT_{n+1}$ if $\forall i \in {0, \ldots, n}.\, \gspairl{\GT_i}{\Theta_i} \trans{\ell_{i+1}} \gspairl{\GT_{i+1}}{\Theta_{i+1}}$, (or simply $\GT \myreaches{~}{} \GT'$ when labels are irrelevant).
We say $\GT'$ is reachable from $\GT$ if $\GT \myreaches{~}{} \GT'$. We just say $\GT'$ \emph{is reachable} if it is reachable from an initial global type. Labels or reached states may be omitted.

The LTS for global types has
a monotonicity property: the sets $\lset$ and $\rset$ of a MC are monotonically non-decreasing with respect to transition.

\begin{example}[Global MC steps]\label{ex:wigglymc}

One way a well-formed MC may proceed is shown below.
In the second line, assume $\GT_1 \myreaches{~}{} \GT_1'$ on the LHS where
observer $\p$ is not the subject in any of those steps.

\smallskip
\centerline{$
\begin{array}{lll}

\GtoDef{\q \mathbin{\rightarrowtriangle} \p : a_1 . \GT_1}{c}{\lset}{\rset}{\p \mathbin{\rightarrowtriangle} \q : a_2 . \GT_2}
& \trans{\nlab } &
\q \mathbin{\rightarrowtriangle} \p : a_1 . \GT_1 \blacktriangleright_{\emptyset, \emptyset}^c \p \mathbin{\rightarrowtriangle} \q : a_2 . \GT_2
\\
& ~\myreaches{~}{} &
\q \mathbin{\rightsquigarrowtriangle} \p : a_1 \{ a_1 . \GT_1'\} \blacktriangleright_{\emptyset, \emptyset}^c \p \mathbin{\rightarrowtriangle} \q : a_2 . \GT_2 \quad(*)
\\
\end{array}
$}
\smallskip

\noindent
The marked state may then proceed one of two ways.
In the upper of the paths below, observer $\p$ also commits to the LHS,
followed by $\q$ and other roles after some more steps.

\smallskip
\centerline{
\begin{tikzpicture}
\node (A) {$(*)$};
\node[right=10mm of A,yshift=2.8mm] (B) {
$\GT_1'
\blacktriangleright_{\{\p\}, \emptyset}^c
\GREY{\p \mathbin{\rightarrowtriangle} \q : a_2 . \GT_2}$};
\node[right=27mm of B] (D) {
$\myreaches{~}{}~ \GT_1''
\blacktriangleright_{\{\p, \q, ...\}, \emptyset}^c
\GREY{\p \mathbin{\rightarrowtriangle} \q : a_2 . \GT_2}$};
\node[right=10mm of A,yshift=-2.8mm] (C) {
$\q \mathbin{\rightsquigarrowtriangle} \p : a_1\{a_1 . \GT_1'\}
\blacktriangleright_{\emptyset, \{\p\}}^c
\p \mathbin{\rightsquigarrowtriangle} \q : a_2\{a_2 . \GT_2'\}
~~~\myreaches{~}{}~
\GREY{\q \mathbin{\rightsquigarrowtriangle} \p : a_1\{a_1 . \GT_1'''\}}
\blacktriangleright_{\emptyset, \{\p, \q, \ldots\}}^c
\GT_2''$
};
\path (A) -- node[sloped,yshift=2.5mm]{$\trans{\q\p?a_1}$} (B.west);
\path (A) -- node[sloped,yshift=1mm]{$\trans{\p\q!a_2}$} (C.west);
\end{tikzpicture}
}
\smallskip

\noindent
The lower path corresponds to the race condition situation where $\q$ instead
commits on the RHS, leading to a (transient) period with active interactions
(e.g., $\rightsquigarrowtriangle$) by roles on both sides of the MC.
However, assuming that all roles have appropriate inputs on the RHS by which
they can learn of $\q$'s decision (such as the $a_1$ input by $\p$), they will
eventually follow the observer into committing on the RHS and no further
actions will occur on the LHS.
The conditions for ensuring safe eventual commitment by all roles are
formalised in the next subsection.
\end{example}

\subsection{Progress of Global Types}
\label{sec:globalprogress}

Progress ensures that any role $\p\in \RolesG{\GT}$ either (1) can make a move (immediately or after actions by other roles) or (2) is in a final state. Intuitively, this means that no role is permanently stuck. The notion of \emph{final state} is captured syntactically by $\RolesG{\GT}$: if $\role{r} \in \RolesG{\GT}$ then $\role r$ is \emph{not} in a final state (i.e., it is still active in $\GT$).
Recall that $\RolesG{\GtoP{\GT_1}{}{\lset}{\rset}{\GT_2} } 	=
\RolesG{  \GT_1}\setminus\rset~ \cup ~\RolesG{  \GT_2}\setminus\lset$ and hence, once a role is committed to one block, it is not active in the other block. For example, both $\p$ and $\q$ are in a final state in $\GtoP{\GmsgShort{p}{q}{\{a.\tend\}}}{}{\emptyset}{\{\role p,\role q\}}{\tend}$.

\begin{definition}[Progress]\label{def:pro}
$\GT$ enjoys \emph{progress} if for all $\GT'$ reachable from $\GT$ the following holds:
 \[\role r\in \RolesG{\GT'} \quad \Rightarrow\quad \GT' \myreaches{}\hspace{-0.2cm}\trans{\ell} ~\text{ with }~\sbj{\ell}=\role r\]
\end{definition}

As standard in MST~\cite{HondaYC08,DBLP:journals/mscs/CoppoDYP16},
not all global types enjoy progress.  
%
We give two sufficient\footnote{Complete decidable conditions are, unfortunately, not possible
~\cite{DBLP:journals/iandc/GoudaMY84}.} conditions for progress: \emph{awareness} and \emph{balance}.


\emph{Awareness} guarantees that each role eventually commits when the RHS is taken, and that each role with a terminating execution commits when the LHS is taken. Awareness builds on two relations over roles (\Cref{def:dep}).

\begin{definition}[Role dependencies $\dep{}{}{\GT}$ and $\wdep{}{}{\GT}$]\label{def:dep}\label{edep}
Let $\role p,\role q\in\RolesG{\GT}$.
We define two kinds of dependencies between roles:
\emph{strict dependence}, written $\dep{\role p}{\role q}{\GT}$, and
\emph{eventual dependence}, written $\wdep{\role p}{\role q}{\GT}$.
\begin{itemize}
%
\vspace{-0.2cm}
    \item
    $\dep{\role p}{\role q}{\GT}$, if $\GT\trans{\vec{\ell}\, \ell_\role{q}}~\land ~\sbj{\ell_\role{q}} = \q$
     implies $\exists {\ell}\in \vec{\ell}.~\sbj{\ell} =\role{p}$
%
%
\item
$\wdep{\role p}{\role q}{\GT}$, if $\GT\trans{\vec{\ell}}\GT'$ implies $\exists a. ~\p\q!a \in \vec{\ell}~\lor ~ \GT' \myreaches{}\trans{\p\q!a}$
\end{itemize}
 \end{definition}
 Namely: $\dep{\role p}{\role q}{\GT}$ if $\role q$ will only take action  after $\role p$ does, and $\wdep{\role p}{\role q}{\GT}$ if it is always possible for $\p$ to send a message to $\q$.
We say that $\role p$ \emph{diverges} in $\GT$ if $\GT\myreaches{}\GT'$ implies $\GT'\myreaches{}\trans{\ell}$ with  $\sbj{\ell}=\p$.

\begin{definition}[Awareness]\label{def:aware}
An active MC $\GT = \GtoP{\GT_1}{\role p}{\lset}{\rset}{\GT_2}$ is \emph{aware} if for all $\role r\in \RolesG{\GT}\setminus\role p $:\\
(1) $~\rset = \emptyset
~ \Longrightarrow~\dep{\role p}{\role r}{\GT_2}$ (single-decision), \\
(2) $~ \lset = \emptyset ~\Longrightarrow~\wdep{\role p}{\role r}{\GT_1}$
or $\role r$ diverges in $\GT_1$
(clear-termination).

Similarly, an MC definition $\GT = \GtoDef{\GT_1}{\role p}{}{}{\GT_2}$ is aware if for all $\role r\in \RolesG{\GT}$, $\dep{\role p}{\role r}{\GT_2}$ and $\wdep{\role p}{\role r}{\GT_1}$. A global type $\GT$ is aware if all MC definitions and active MC in $\unfold{\GT}$ are aware.
\end{definition}
Single-decision requires that, on the RHS, all roles depend on the observer until committed. 
Clear termination allows the roles to start communicating on the LHS without waiting for a committing message. However, if their execution terminates, they must receive a committing message; this is key to progress, as illustrated below in Example~\ref{ex:cleartermination}.
The use of $\unfold{\GT}$ is necessary to deal with label capture (as shown in Example~\ref{ex:unf}).

\begin{example}[Clear termination]\label{ex:cleartermination}
In the type below (left): if observer $\p$ receives $a$ and terminates, then $\role q$ is unable to locally determine whether to terminate or continue waiting for a $b$ that will never arrive (i.e., termination is not \emph{clear}); one fix is to add $\GactShort{p}{q}{c}$ on the LHS. As a generalisation, however we can safely lift eventual dependency for infinite executions, as shown in a \emph{stream exception} pattern below (right) where $\q$ establishes a communication (label $c$) and immediately starts the stream, and potentially handle later a connection failure message ($b$ by $\p$):

\smallskip
\centerline{$
\GactShort{q}{p}{a.~\tend}\mix \GactShort{p}{q}{b.~\tend}~~\textcolor{red}{\textrm{\ding{55}}}
 \qquad
\GactShort{q}{p}{c.\, } \mu\rv. (\GactShort{q}{p}{a.~\rv})\mix \GactShort{p}{q}{b.~\tend}~~\textcolor{green}{\textrm{\ding{51}}}
$}
\smallskip

\noindent
Consider now a \emph{third-party exception} below, where a `third-party' observer $\role r$ may decide, after some interactions between $\p$ and $\q$, to either commit on the LHS or raise the RHS exception:

\smallskip
\centerline{$
\GactShort{q}{r}{c.\,} \mu\rv. (\GactShort{q}{p}{a.~\rv})\mix \GactShort{r}{q}{b.\,} \GactShort{r}{p}{b.~\tend}~~\textcolor{green}{\textrm{\ding{51}}}
$}
\smallskip

\noindent
Clear-termination holds because: (i) $\role r$ is able to commit to the LHS by receiving $c$, (ii) all other roles ($\role p$, $\role q$) do not need to commit because they diverge within the LHS.
\end{example}

\begin{example}[Interrupt pattern]\label{ex:interrupt}
Interrupts can be modelled with a minor extension to the theory (implemented in \Cref{sec:evaluation}), if we allow the observer to interact on the LHS before committing on that side.
To this aim, allocate a set $\GC$ of choice labels used by observers to \emph{commit on the LHS}.
The \emph{interrupt pattern} can then be expressed as the global type below, letting ${\color{red}d}\in \GC$ and $a,b,c\not\in \GC$:

\smallskip
\centerline{$
\GactShort{q}{p}{c.\, }\mu\rv. (\GactShort{q}{p}{\{a.~\rv,\, {\color{red}d}.\, \GactShort{p}{q}{e.\, } \tend\}})\mix \GactShort{p}{q}{b.~\tend}
$}
\smallskip

\noindent
Unlike in the stream exception (Example~\ref{ex:cleartermination}), now $\p$ can repeatedly \emph{receive} messages $a$ from $\q$ before, possibly, throwing the interrupt $b$ (or committing/terminating on the LHS).
Our theory can be easily adapted to support interrupt patterns using $\GC$ by: (1) updating the definition of committing set so that the committing chains on the LHS start with actions in $\GC$ (see \Cref{asec:committing2}), (2) adjusting the existential quantifier in \Cref{edep} (eventual dependence $\wdep{}{}{\GT}$) by requiring label $a$ to be in $\GC$ (i.e., $\exists {a}\in\GC$). 
\end{example}

Awareness ensures convergence of the roles in a block by requiring dependencies with the observer. However, if a role does not appear in \emph{some} execution branches, it may not be able to converge with the  observer's decision. This is a known problem also with branching choices:
\begin{equation}\label{eq:particip1}
\GactShort{p}{q}{\{a.\, \GactShort{q}{r}{\{c.\tend\}}, ~
b .\,  \GactShort{q}{p}{\{c.\tend\}}}\}
\end{equation}

In (\ref{eq:particip1}), role $\role r$ does not know whether to terminate or wait for a message, and thus does not terminate. \emph{Balance} (\Cref{wellset}) generalizes this idea to MC, operating on the unfolding $\unfold{\GT}$ to account for label capture. 
\Cref{wellset} also uses a truncation operator $\trunc{\GT}$ that replaces each recursive subterm in $\GT$ with $\tend$. Assuming $\GT$ has no free recursion variables: $\trunc{\mu \rv . \GT'} = \trunc{\tend} = \tend$. All other cases are defined inductively.
Truncation prevents expressiveness loss:
\Cref{wellset} universally quantifies over subterms of $\unfold{\GT}$ which would consider recursive subterms out of their intended context. For example, without truncation any instantiation of example (\ref{unf}) in \Cref{sec:committing} would (unnecessarily) be excluded as unbalanced.

\begin{definition}[\Bname]\label{wellset}
A global type $\GT$ is \emph{\Bpp} if for all subterms $\GT'$ of $\trunc{\unfold{\GT}}$:
\begin{enumerate}
    \item $\GT'=\GactShort{p}{q}{\St{}}~~\Rightarrow~~
    \forall \GT_1,\GT_2\in \tyG{\St{}}, ~~\RolesG{\GT_1}\setminus\{\p,\q\}=\RolesG{\GT_2}\setminus\{\p,\q\}$
    \item
    $\GT'=\GtoDef{  \GT_1}{~}{\lset}{\rset}{  \GT_2} ~~\Rightarrow~~
    \RolesG{\GT_1}=\RolesG{\GT_2}$
   \qquad (3)~
    $\GT'=\GtoP{\GT_1}{~}{\lset}{\rset}{\GT_2}
    ~~\Rightarrow~~
    \RolesG{\GT_1}\cup \lset=\RolesG{\GT_2}\cup\rset $
    \end{enumerate}
\end{definition}
Case (3) allows committed roles to `disappear' from a block. Case (1) is normally entailed by
projection. We give it here for a simpler presentation and separation of concerns.

\begin{restatable}
[Progress]{theorem}{cor:gprog}\label{cor:gprog}
If $\GT$ is initial, aware, and balanced then it enjoys \textbf{progress}.
\end{restatable}
The proof relies on a \emph{coherence} property:
all the committed roles are committed to the same side. Formally, if
$\GT = \Ct{C}{
\GtoP{\GT_1}{}{\lset}{\rset}{\GT_2}}$ then $\lset=\emptyset\, \lor\, \rset=\emptyset$.
Essentially, we prove that (1) awareness, balance, and coherence are preserved by transition, and (2) balanced, aware, and coherent global types enjoy progress. \Cref{cor:gprog} follows since initial global types are always coherent. 
See \Cref{asec:progress} for details of the proofs.

\section{Local Types}
\label{sec:localtypes}

The syntax of \emph{local types} is given by the following grammar:
\[ \ST  ::=   \tbrai{\p}{a}\ST_i
  ~ \mid ~ \tseli{\p}{a}\ST_i
  ~   \mid ~ \mu \mathtt{t} . \ST
  ~ \mid ~ \mathtt{t}
  ~\mid ~ \tend
  ~ \mid ~ \ST_1 \mixi{c} \ST_2                       %
  ~ \mid ~ \ST_1 \lmixi{c} \ST_2
    ~\mid ~ \ST \lmixi{c}\bullet
     ~\mid~  \bullet \lmixi{c}  \ST\]
The first five terms are standard; the notation for roles and labels are as in global types. We recall $\tbrai{\p}{a}\ST_i$ is a branching type, waiting for one of the labels $\{a_i\}_{i\in I}$ and continuing as the corresponding $\ST_i$. Term $\tseli{\p}{a}\ST_i $ is the corresponding send/selection type.  We often omit $\tend$.

Terms $\ST_1 \mixi{c} \ST_2$ and $\ST_1 \lmixi{c} \ST_2 $
are for \emph{MC definition} and \emph{active MC}, respectively.
\PURPLE{The terms $\ST \lmixi{c} \bullet$} and $\bullet \lmixi{c} \ST$ model run-time states in which the role has \emph{committed} on the LHS or the RHS, respectively.
We omit the annotation
$c$ when it is clear from context or not relevant.
In-transit/buffered (i.e., yet to be consumed) messages are modelled using FIFO queues. The localised view of a global type for a given role is thus a \emph{configuration} $(\p, \ST, \sigma)$
where $\role p$ is the role, $\ST$ is the local type specifying the behaviour of $\p$, and $ \sigma$ is a local FIFO queue of messages that have been sent to $\p$ (and possibly arrived) but not consumed.
A global type therefore corresponds overall to a collection of local configurations. We call this collection a \emph{system}, ranged over by $Y,Y'$.
\[\begin{array}{llll}
Y ::= (\p_i, \ST_i, \sigma_i)_{i \in I}~~
\qquad
& \sigma : \q \mapsto \vv{m}~~
\qquad
& m::=\msg{p}{a}{\pth}
\qquad
&
\pth ::= \epath \; | \; \pleft. \pth \; | \; \pright. \pth
\end{array}
\]
We assume the configurations of a system have pairwise-distinct roles.

Messages carry a \emph{path} $\pth$, which fully qualifies the active MC context under which the message was sent.  A path specifies left-to-right the top-most MC to the inner MC of the sending local type. We often omit the trailing $\epath$. 
As an example,
$(\ST_1 \blacktriangleright {\color{magenta} \p \oplus a}) \blacktriangleright{\color{magenta} \p \oplus b}$ may send two messages: $(a,\, \pleft.\pright)$ and $(b, \pright)$.
In global types, the equivalent path information is implicit in the global type context (nested active MCs) in which the corresponding in-transit message term ($\rightsquigarrowtriangle$) occurs.

\subsection{Operational Semantics}
\label{sec:localsemantics}

The local operational semantics uses the labels below and two main judgments:
\newcommand{\llab}{a}
\[
\begin{array}{c}
\ell ::= \p\q!\llab \;|\; \p\q?\llab \;|\; \nlab   \;|\; \gcl\qquad
Y \trans{\ell} Y'
\qquad
\tstate{\ltheta}{\lenv}{ Y }\trans{\ell} Y'
\end{array}
\]
\noindent
Labels $\ell$ now include $\gcl$ for stale message purging from local queues.
The first judgment $Y \trans{\ell} Y'$ is for top-level concurrent execution of systems.
It has just three rules.
{\small
\[
\begin{array}{c}
\inferrule*[Right={[Par]}]{
{ Y }\trans{\ell} Y''
}{
Y, Y' \trans{\ell} Y'', Y'
}
\qquad
\qquad
\quad
\inferrule*[Right={[Discard]}]{
\mathsf{gc}(\ST,\sigma) = \sigma'
}{
(\p, \ST, \sigma)
\trans{\gcl}
(\p, \ST, \sigma')
}
\qquad
\qquad\qquad
\inferrule*[Right={[Low]}]{
\tstate{\ltheta}{\lenvp{0}{0}{\epsilon}}{ Y }\trans{\ell} Y'
}{
Y \trans{\ell} Y'
}
\end{array}
\]
}

\noindent
Rule~\textsc{[Par]} is a structural rule that models concurrent execution.
Rule~\textsc{[Low]} uses lower-level judgments of the form $\tstate{\ltheta}{\lenv}{ Y }\trans{\ell} Y'$ (see below).
Rule~\textsc{[Discard]} performs stale message purging on the local queue
$\sigma$ of a configuration. This is necessary to maintain correspondence
between the executions of global and local types. Consider the example in (\ref{gc}) where the right-hand side term is reached after the following sequence of actions: (1) $\q$ sends $c$ to $\p$ and commits to the RHS, (2) $\p$ sends $a$ to $\q$ (on the LHS), (3) $\p$ receives $c$ and commits to the RHS, (3) $\p$ sends $d$ to $\q$ on the RHS.
\begin{equation}\label{gc}
Y, (\q,\, \p \& a . \p \oplus b.\tend \lmix
\p \oplus c . \p \& d.\tend , \sigma)
\myreaches{}
Y', (\q,\, \bullet \lmix
\p \& d.\tend , \sigma: \p \mapsto (a,\pleft),(d,\pright))
\end{equation}
In the reached term -- on the right-hand side of (\ref{gc}) -- the enqueued message $(a,\pleft)$ pertains to the LHS of the MC and is essentially garbage. The presence of $(a,\pleft)$ in $\q$'s queue would naively break correspondence (\Cref{sec:correspondence}) with the projection of the corresponding metatheoretical global type
$\p\rightsquigarrowtriangle \q:a. \GactShort{q}{p}{b.\tend} \blacktriangleright_{\emptyset,\{\p,\q\} }
\GmsgShort{p}{q}{}{d.\tend}$.

Fortunately, local types allow us to define a \emph{purge} function $\gc$ that operates solely on information that is \emph{local} to the configuration. 
In the general case of nested MC, we identify stale messages by traversing the MC structures of the receiver's local type $\ST$ using the sender's path $\pth$ included in the message.
This check is defined inductively as a predicate $\mathtt{stale}(\pth, \ST)$, which returns true if following $\pth$ in $\ST$ hits a stale ($\bullet$) side. Below, the otherwise cases may include non-initialized MC which are never stale as still uncommitted, and  $\mathtt{stale}(\epsilon, \ST)= \mathtt{false}$ for all $\ST$.
{\small
\begin{equation*}
\label{eq:gc}
\begin{array}{llcll}
\begin{array}{llcll}
\mathtt{stale}(\pleft.\pi, \ST) =
\begin{cases}
 \mathtt{true}
        & \ST = \bullet  \blacktriangleright  \ST_2\\
 \mathtt{stale}(\pi, \ST_1)
        & \ST = \ST_1 \blacktriangleright \ST_2\textit{ or }\\
        &~\ST = \ST_1 \blacktriangleright \bullet  \\
 \mathtt{false}
        & \textit{otherwise}
\end{cases}
&  &
\mathtt{stale}(\pright.\pi, \ST) =
\begin{cases}
 \mathtt{true}
    & \ST = \ST_1 \blacktriangleright \bullet \\
 \mathtt{stale}(\pi, \ST_2)
    & \ST = \ST_1 \blacktriangleright \ST_2\textit{ or } \\
    & ~\ST = \bullet \blacktriangleright \ST_2 \\
 \mathtt{false}
    & \textit{otherwise}
\end{cases}
\end{array}
\end{array}
\end{equation*}
}
Purge (a.k.a.\ garbage collection) $\gc(\ST, \sigma)$ looks into $\sigma$ (for all $\p$ in its domain) and removes all messages $(a, \pi)$ for which $\mathtt{stale}(\pi, \ST)$ is true.

Lower-level judgements $\lenv \vdash Y \trans{\ell} Y'$ model the main behaviour of configurations.
The environment $\pth$ records the MC context for messages being sent.
The first three rules in Figure~\ref{fig:lsem2} are standard for communications and recursion. Rule~$\lrule{Snd}$ appends an $a_k$ message to the receiver's queue and continues as $\ST_k$. The message is annotated with $\pth$ from the context.
Dually, rule~$\lrule{Rcv}$ consumes the first such
message $a_k$ annotated by $\pth$ and continues as $\ST_k$.
$\lrule{Rec}$ is standard. 

The remaining rules are new and define the semantics for MC. \textsc{[New]} instantiates a MC definition.
\textsc{[LSnd]} executes an output action in the LHS of an active but not committed MC. The
context $\pth$ is updated to $\pth . \pleft$ to indicate the action is happening on the LHS of the current MC.
Analogously to global types, sending on the LHS is never committing. For LHS receive actions we have two rules: $\lrule{LRcv1}$ (committing) and~$\lrule{LRcv2}$ (non-committing).
As in the global semantics, we infer committing labels from the base $\underline{\GT}$.\footnote{Observe that this is purely static information and can be derived (pre-processed) from the source global type.  Equivalently, we could embed this information into local types by adding annnotations.
}
$\lrule{RSnd}$ is for sending on the RHS which is always committing;
likewise $\lrule{RRcv}$ for receiving on the RHS (symmetric to $\lrule{LRcv1}$).
%
See \Cref{app:locallts} for the full rules, including a  
structural rule  $\lrule{RCtxt}$ (which is symmetric to $\lrule{LCtxt}$), and
structural rules $\lrule{NLCtxt}$ and $\lrule{NRCtxt}$ for $\lnewlab$ actions on the LHS and RHS of $\ST_1\blacktriangleright \ST_2$, respectively.

\begin{figure}[t]
\small
$ \arraycolsep=1pt
\begin{array}{cr}
\inferrule{%
k \in I
\qquad
m = \msg{p}{a_k}{\pth}
}{%
\tstate{\ltheta}{\lenv}{
(\p, \tseli{\q}{a} \ST_i, \sigma), (\q, \ST, \sigma'[\p \mapsto \vv{m}])}
\trans{\p\q!\llab_k}
(\p, \ST_k, \sigma), (\q, \ST,\sigma'[\p \mapsto \vv{m}\cdot m])
}
&
{\scriptstyle \lrule{Snd}}
\\[0.8cm]
\inferrule{%
k \in I
\qquad
\vv{m} = \vv{m}_1 \cdot \msg{p}{a_k}{\pth} \cdot {\vv{m}}_2
\qquad
(a, \pth) \not\in \vv{m}_1
}{%
\tstate{\ltheta}{\lenv}{
(\q, \tbrai{\p}{a}\ST_i, \sigma[\p \mapsto \vv{m}])}
\trans{\p\q?\llab_k}
(\q, \ST_k, \sigma[\p \mapsto \vv{m}_1 \cdot {\vv{m}}_2])
}
&
{\scriptstyle \lrule{Rcv}}
\\[0.6cm]
\inferrule{
\tstate{\ltheta}{\lenv}{
(\p, \ST[\mu \mathtt t . \ST / \mathtt t], \sigma), Y}
\trans{\ell}
Y'
}{
\tstate{\ltheta}{\lenv}{
(\p, \mu \mathtt t . \ST, \sigma), Y}
\trans{\ell}
Y'
}
\qquad
\tstate{\ltheta}{\lenv}{
(\p, \ST \mixi{c} \ST', \sigma)}
\trans{\lnewlab}
(\p, \ST\lmixi{c} \ST', \sigma)
&
{\scriptstyle \lrule{Rec/New}}
\\[0.6cm]
\inferrule{
{\begin{array}{l}
\tstate{\ltheta}{\lenv . \pleft }{
(\p, \ST_1 , \sigma), Y}
\trans{\p\q!\llab}
(\p, \ST'_1, \sigma), Y'
\end{array}}
}{
\tstate{\ltheta}{\lenv }{
(\p, \ST_1 \lmix \ST_2, \sigma), Y}
\trans{\p\q!\llab}
(\p, \ST'_1 \lmix \ST_2, \sigma), Y'
}
~
~
~
\inferrule{
{\begin{array}{l}
\tstate{\ltheta}{\lenv . \pright}{
(\p, \ST_2 , \sigma), Y}
\trans{\p\q!\llab}
(\p, \ST'_2, \sigma), Y'
\end{array}}
}{
\tstate{\ltheta}{\lenv }{
(\p, \ST_1 \lmix \ST_2, \sigma), Y}
\trans{\p\q!\llab}
(\p, \bullet \lmix \ST'_2, \sigma),Y'
}
&
{\scriptstyle \lrule{LSnd/RSnd}}
\\\\
\inferrule{
{\begin{array}{l}
\tstate{\ltheta}{\lenv . \pleft }{
(\p, \ST_1 , \sigma)}
\trans{\p\q?\llab}
(\p, \ST'_1, \sigma')
\end{array}}
\quad
{\color{olive} a\in \comSetG{
 \underline{\GT}}}
}{
\tstate{\ltheta}{\lenv}{
(\p, \ST_1 \lmix \ST_2, \sigma)}
\trans{\p\q?\llab}
(\p, \ST'_1 \lmix \bullet, \sigma')
}
\quad
\inferrule{
{\begin{array}{l}
\tstate{\ltheta}{\lenv . \pleft }{
(\p, \ST_1 , \sigma)}
\trans{ \p\q?\llab}
(\p, \ST'_1, \sigma')
\end{array}}
\quad
{\color{olive} a\not \in \comSetG{
 \underline{\GT}}
}}{
\tstate{\ltheta}{\lenv }{
(\p, \ST_1 \lmix \ST_2, \sigma)}
\trans{ \p\q?\llab}
(\p, \ST'_1 \lmix\ST_2, \sigma')
}
&
{\scriptstyle \lrule{LRcv1/2}}
\\
\\
\inferrule{
{\begin{array}{l}
\tstate{\ltheta}{\lenv . \pright }{
(\q, \ST_2 , \sigma)}
\trans{\p\q?\llab}
(\q, \ST'_2, \sigma')
\end{array}}
}{
\tstate{\ltheta}{\lenv }{
(\q, \ST_1 \lmix \ST_2, \sigma)}
\trans{\p\q?\llab}
(\q, \bullet \lmix \ST'_2, \sigma')
}
\quad
\inferrule{
\tstate{\ltheta}{ \pth . \pleft  }{(\p, \ST, \sigma), Y}
\trans{\ell}
(\p, \ST', \sigma'), Y'
}{
\tstate{\ltheta}{\lenv }{
(\p, \ST \lmix \bullet, \sigma), Y }
\trans{\ell}
(\p, \ST' \lmix \bullet, \sigma'), Y'
}
&
{\scriptstyle \lrule{RRcv/LCtxt}}
\end{array}$
\caption{Local semantics, selected rules.
See \Cref{app:locallts} for omitted structural rules.
and the two context rules, $\lrule{NLCtxt}$ and $\lrule{NRCtxt}$, for $\lnewlab$ actions on the LHS and RHS of $\ST_1\blacktriangleright \ST_2$, respectively.
}
\label{fig:lsem2}
\end{figure}

\subsection{Projection}
\label{sec:projection}

The distributed counterpart of a global type is a system of configurations derived by \emph{projection}.
The projection of a global type $\GT$ onto a role $\role r\mathop{\in}\RolesG{\GT}$,
written $\GT \proj \role{r}$, returns a pair $(\ST, \sigma)$ of a local type and its queue.
Recall, the queue $\sigma$ is a mapping from roles $\role r'\in\Roles(\GT)$ to a vector of messages received, but not consumed, by $\role r$ from $\role r'$.

The main rules for projection,
given in \Cref{fig:projection},
take a context path $\pi$, which is needed to define messages $m$ being enqueued.
Top-level projection $\GT \mathbin{\proj} \role{r}$ is bootstrapped as $\projrl{ \GT}{ \role{r}}{\epsilon}$.

\begin{figure}[t]
\small
\centering
\[\begin{array}{l}
\projrl{\GactShort{\p}{\q}{\{ a_i : \GT_i\}_{i\in I}}}{\role{r}}{\pth} =
\begin{cases}
\q \oplus_{i \in I} a_i . \ST_i, \sigma
& \text{if }{\role{r}= \p } \quad \quad\text{where in all cases } \forall i\in I,\, {\color{black}\ST_i,\sigma} =
\projrl{ \GT_i }{\role{r}}{\pth}
\\
\p \&_{i \in I} a_i . \ST_i, \sigma
& \text{if }{\role{r}= \q } \quad
\\
\sqcap_{i \in I} \ST_i , \sigma
& \text{if } \role r\not\in\{\role p, \role q\}
\end{cases}
\\[0.7cm]
\projrl{\GmsgShort{p}{q}{k}{\{a_i.\GT_i\}_{i\in I}} }{ \role{r}}{\pth} =
\begin{cases}
\ST_k, \sigma_k
& \text{if }\role{r}= \p
\\
\p \&_{i \in I} a_i . \ST_i, \, \sigma_k[\p \mapsto \msg{p}{a_k}{\pth} \mathop{\cdot} \sigma_k(\p)]
& \text{if }\role{r}= \q
\\
\ST_k, \sigma_k
& \text{if } \role r\not\in\{\role p, \role q\}
\end{cases}
\\
\hspace{4.3cm}\text{where in all cases }\forall i\in I,\, {\color{black}\ST_i, \sigma_i} = \projrl{ \GT_i }{\role{r} }{\pth} \text{ and }\forall i\in I\setminus\{k\},\, \sigma_i = \sigma_0
\\[0.1cm]
\projrl{ \GT_1 \mixc{c} \GT_2 }{\role{r}}{\pth} ~~=~~
\ST_1 \mixc{c} \ST_2, \sigma_0
\qquad \text{where }\forall i\in\{1,2\},\, {\color{black}\ST_i, \sigma_0} =
\projrl{ \GT_i }{\role{r}  }{\pth}
\\[0.1cm]
\projrl{\GtoP{\GT_1}{}{\lset}{\rset}{ \GT_2} }{\role{r}}{\pth} =
\begin{cases}
\ST_1 \lmix \bullet, \sigma_1
& \text{if }\role{r}\mathop{\in} \lset ~~\text{ and }~~
{\color{black}\ST_1, \sigma_1} = \projrl{ \GT_1} {\role{r}}{\pth . \pleft }
\\
\bullet \lmix \ST_2, \sigma_2
& \text{if }\role{r}\mathop{\in} \rset ~~\text{ and }~~
{\color{black}\ST_2, \sigma_2} = \projrl{ \GT_2 }{\role{r}}{\pth . \pright }
\\
\ST_1 \lmix  \ST_2, \sigma_1 \circ \sigma_2
&
\text{if }\role{r}\mathop{\not\in} \lset \mathop{\cup}\rset ~
 \text{ and }~\forall i\in\{1,2\},\, \ST_i, \sigma_i =
\projrl{ \GT_i }{\role{r}}{\pth_i}\\
& \quad \text{ with }~\pth_1 = \pth . \pleft~\text{ and }~ \pth_2 = \pth . \pright
\end{cases}
\\[0.8cm]
\begin{array}{c}
\projrl{\mu \mathtt{t} . \GT }{ \role{r}}{\pth} =
\begin{cases}
\mathtt{end}, \sigma_0 & ~\text{ if } \GT = \mathtt t
\hspace{4cm} \projrl{  \mathtt t}{\role{r}}{\pth} = \mathtt t, \sigma_0
\\
\mathtt{t}', \sigma_0 & ~\text{ if } \GT = \mathtt t' \land \mathtt t' \neq \mathtt t
\hspace{2.75cm} \projrl{\mathtt{end}}{\role{r}}{\pth} = \mathtt{end}, \sigma_0
\\
\mu \mathtt{t} . \ST, \sigma & ~\text{ otherwise, with } \projrl{ \GT}{\role{r}}{\pth} = \ST, \sigma
\end{cases}
\end{array}
\end{array}\]
\caption{Projection \PURPLE{of global types to} local types and input FIFOs}
\label{fig:projection}
\end{figure}
The first two rules for interaction and in-transit messages are standard, except we also project the queues. Notation $\sigma_0$ stands for an `empty queue' that maps each
$\role{p}$ in its domain to $\epsilon$.
The projection of a message in transit
$\GmsgShort{p}{q}{k}{a_i.\GT_i}$ uses the context $\pi$ to
construct the queued message.
The case of projection for an active MC, of the form $\GT_1 \lmix \GT_2$, uses
the operation $\sigma_1 \mathbin{\circ} \sigma_2$ to concatenate queues: if $\mathsf{dom}(\sigma_1) = \mathsf{dom}(\sigma_2)$ then $
\sigma_1 \mathbin{\circ} \sigma_2 = \{\role{r}\mapsto \sigma_1(r) \mathop{\cdot} \sigma_2(r) \mid \role{r} \mathop{\in} \mathsf{dom}(\sigma_1)\}$.
The behaviour of an MC can be considered as starting from the LHS before possibly switching to the RHS, hence the projection accordingly concatenates the
LHS and RHS queues in this order.

The notation $\sqcap_{i\in I}\ST_i$ stands for the 
\emph{merge} operator for projecting onto third-party roles, inspired from~\cite{DenielouY12}.
For example, the projection of
$\GactShort{p}{q}\{a_1.\,\GactShort{q}{r}\{a_1.\tend\},\, a_2.\,\GactShort{q}{r}\{a_2.\tend\}\}$
onto $\role r$ is defined as
$\role p \& a_1.\tend \sqcap \role p \& a_2.\tend$, yielding $\role p \&_{i\in\{1,2\}} a_i.\tend$.
The definition (see~\Cref{zsec:projection})
implicitly handles MC in the same way as other non-branching constructs, i.e., the projection
of each case must be the same.

We extend projection to systems
and define the \emph{system derived} from $\GT$ as $\derived{\GT} = \{ (\role r, \ST_{\role r}, \sigma_{\role r}) \}_{\role{r} \in \RolesG{\GT}}$  where $\forall \role r\in\RolesG{\GT} \;.\; \ST_{\role r}, \sigma_{\role r} = \GT\mathbin{\proj} \role{r}$.

\subsection{Operational Correspondence}
\label{sec:correspondence}

We establish an operational correspondence, called \emph{fidelity}, between global types and the systems obtained by projection. Fidelity is defined as a weak correspondence relation over pairs of global types and systems, treating $\nlab$ and $\gcl$ actions as the silent action $\tau$. This correspondence is precise, comprising a top-down and a bottom-up property, as depicted below.

\centerline{
\small
\begin{tabular}{cccc}
{\begin{tikzpicture}
\node (G1) {$\GT$};
\node[right=41mm of G1] (G2) {$\GT'$};
\node[below=5mm of G1] (Y1) {$Y_{\GT}$};
\node[right=16mm of Y1] (Y2) {$Y$};
\node[right=32mm of Y1] (Y3) {$Y'\color{black}\mathbin{<:} Y_{\GT'}$};
\draw[->, dotted]
(G1) -- node[left]{$\proj$} (Y1);
\draw[->, dashed] (G1) -- node[above]
{$\text{\PURPLE{(bottom-up)}}\qquad\trans{\gcc{\ell}}$} (G2);
\draw[->] (Y1) -- node[above]{$\trans{\ell}$} (Y2);
\draw[->, dashed] (Y2) -- node[above]{$\myreaches{\gcl}$} (Y3);
\draw[->, dotted] (G2) -- node[left]{$\proj$} ($(Y3.north)+ (3.5mm, 0)$);
\end{tikzpicture}}
&\quad  &
{\begin{tikzpicture}
\node (G1) {$\GT$};
\node[below=5mm of G1] (Y1) {$Y_{\GT}$};
\node[right=41mm of G1] (G2) {$\GT'$};
\node[right=16mm of Y1] (Y2) {Y};
\node[right=32mm of Y1] (Y3) {$Y'\color{black}<: Y_{\GT'}$};
\draw[->, dashed] (Y1) -- node[above]{ 
$\trans{{\ell}}$} (Y2);
\draw[->, dashed, color=black] (Y2) -- node[above]{$ \myreaches{\gcl} $} (Y3);
\draw[->, dotted] (G1) -- node[right] (D) {$\proj$} (Y1);
\draw[->] (G1) -- node[above]{$\text{\PURPLE{(top-down)}}\qquad\trans{\gcc{\ell}}$} (G2);
\draw[->, dotted] (G2) -- node[left]{$\proj$} ($(Y3.north)+ (3.5mm, 0)$);
\end{tikzpicture}}
\end{tabular}
}

\noindent
Bottom-up fidelity states that each action of a system can be matched with the same action by the corresponding global type, preserving the correspondence given by projection modulo garbage collection and a preorder `$<:$' over local types (discussed below). Top-down fidelity, conversely, states that each step of $\GT$ can be matched by a step of $Y$.

The asynchronous nature of mixed choices and the decentralised semantics of systems raise two main challenges: (1) upon commitment of a configuration to a MC block, some messages in its queue may become stale; and (2) localised MC instantiations and recursive unfoldings can cause (a minor form of) misalignment of the configurations.
So far, our definitions have paved the way for tackling (1): we use local states $\ST \lmix \bullet$ and
$\bullet \lmix \ST$ to  keep track of local commitments, and use function $\gc()$ and action $\gcl$ to (locally) identify and remove stale messages. Both top-down and bottom-up fidelity may require \PURPLE{some number of} $\gcl$ steps to preserve correspondence.

We address (2) using a preorder `$<:$' on pairs of systems, as we now discuss. In a global type, a MC is instantiated with a single ``centralised'' action $\nlab$, whereas the local configurations in the derived system
independently perform {separate decentralised} $\nlab$ actions.
Decentralised MC instantiation can cause administrative issues for the correspondence. Consider the global types below where $\GT_1= \GactShort{ q}{ p}{a_1.\tend  }$ and $c$ is instantiated before $\role q$'s first send action:

\smallskip
\centerline{$
\begin{array}{lll}
\GT_P   =  \GactShort{ q}{ r}{b. (\GtoDef{\GT_1}{c}{}{}{\GactShort{ p}{ q}{a_2.\tend})}}~\trans{\nlab }\trans{\role p\role q! a_2}
 \GactShort{q}{r}{b.\,
(\GtoP{\GT_1}{c}{\emptyset}{\{\p\}}{\GmsgShort{ p}{ q}{2\,a_2.\tend}})
} =\GT_P'
\end{array}
$}
\smallskip

Consider now configuration $Y_{\role q}$ obtained by projecting $\GT_P$ on $\role q$:

\smallskip
\centerline{
\small{
$
Y_{\role q}= (\role q,\, \role r \oplus b. \, ( \role p \oplus a_1. \tend \mixi{c} \role q \,\& \, a_2 . \tend),\, \sigma_0)
\qquad
Y'_{\role q}= (\role q,\, \role r \oplus b. \, ( \role p \oplus a_1. \tend \lmixi{c}\role q \,\& \, a_2 . \tend),\, \sigma_0)
$}}
\smallskip

\noindent
In $Y_{\role q}$, $\role q$ cannot instantiate the MC before it sends $b$ to $\role r$ (i.e., $Y_{\role q}$ cannot reduce to $Y'_{\role q}$). Namely, the system derived from $\GT_P$ cannot reach the state derived from $\GT_P'$ where all configurations have instantiated $c$ (yet  $\role q$ has not sent $b$). The preorder `$<:$'  is introduced to regulate this deferral of instantiations which may need to be interleaved with other actions. In the example above $Y_{\role q}<: Y'_{\role q}$. The key rule for MC is given below (left). ${\ST}$ is the local type used to bootstrap the derivation at the configuration level (below right). Essentially, the other cases (except for axioms $\tend <: \tend$ and $\rv<:\rv$) are defined inductively, and the preorder is applied to systems pointwise.

\smallskip
\centerline{$
\begin{array}{c}
\inferrule{
 \ST_1' \mathbin{<:} \ST_1
\qquad \ST_2'~ \mathbin{<:}~ \ST_2
}
{
\ST_1' \,{\mixi{c}}\, \ST_2'
~\mathbin{<:}~
\ST_1 \lmixi{c} \ST_2
}
\qquad\qquad
\inferrule{
\ST' \mathbin{<:} {\ST}
}{
(\p, \ST', \sigma )
\mathbin{<:}
(\p, {\ST}, \sigma)
}
\end{array}
$}

\begin{restatable}[Bottom-up fidelity]{lemma}{fullfidelity}
\label{fullfidelity}\label{thm:complete}
Let $\GT$ be a global type reachable from an initial, aware, and balanced global type, and let $Y_{\GT} = \derived{\GT}$. Then the following holds:
\begin{enumerate}
\item $Y_{\GT} \trans{\ell} Y ~\land ~~\ell\neq \rho~~~~\Longrightarrow ~~~~\GT \trans{\gcc{\ell}} {\GT'} \quad \land \quad \exists Y'.~ Y\gcreach{}Y' ~~\land~~ Y' <: \derived{\GT'}$
\item $Y_{\GT} \trans{\rho} Y ~~~~\Longrightarrow ~~~~Y_{\GT} = Y  = \derived{\GT}$
\end{enumerate}
\end{restatable}
\noindent In case (1), if the action by $\GT$ is committing (say for $\p$), then it may cause some of the messages in $\p$'s queue in $Y$ to become stale, hence a $\rho$ action may be needed to restore correspondence. Case (2) highlights a property of projection: derived systems have no stale messages. 

\begin{restatable}[Top-down fidelity]{lemma}{fullcomplete}\label{wc:top}
Let $\GT$ be a global type reachable from an initial, aware, and balanced global type, and let $Y_{\GT} = \derived{\GT}$. Then the following holds:
\begin{enumerate}
    \item  $\gspair{\GT}{\Theta}\trans{\role p\role q \wild a} {\GT'}  \quad\Longrightarrow\quad\exists  Y'.~~~ Y_{\GT} \trans{\role p\role q \wild \llab }\hspace{-2mm}\gcreach Y'
    ~~~\land~~~Y'<: \derived{\GT'}
    \qquad(\wild\in\{!,?\})$
    \item $\gspair{\GT}{\Theta}\trans{\nlab } {\GT'} 
    \quad\Longrightarrow\quad
    \exists Y'.~~~Y_{\GT} \trans{\overrightarrow{\lnewlab}}  Y' ~~~\land~~~Y'<: \derived{\GT'}$
\end{enumerate}
\end{restatable}
\noindent Case (1) is symmetric to case (1) of bottom-up fidelity. In case (2) action $\nu $ can be immediately matched with several actions $\nu  $ by $Y$ (others may be deferred by `$<:$'). 

The operational correspondence is formulated in terms of a weak relation over pairs of global types and systems, where $\nu$ and $\gcl$ are renamed as $\tau$. We write $\taureach$ for a possibly empty sequence of $\tau$ actions and define $\Taureach{\ell} =
   \taureach \trans{\ell} \taureach $ if $\ell \neq \tau$ and $\Taureach{\ell} =\taureach$ otherwise.

\begin{definition}[Weak Correspondence]\label{def:wc}
A relation $\mathcal{R}$ over $(\GT, Y)$ is a \emph{weak correspondence} if whenever $(\GT, Y)\in\mathcal{R}$ then: (1)
If $\GT\trans{\ell}\GT'$ then
$Y\Taureach{\ell}Y'$ and $(\GT',Y')\in\mathcal{R}$; (2) \PURPLE{
If $Y\trans{\ell}Y'$ then
$\GT\,\Taureach{\gcc\ell}\,\GT'$ and $Y'\Taureach{}Y''$  with $(\GT', Y'')\in\mathcal{R}$.}
Two states $\GT$ and $Y$ are weakly correspondent, written $\GT\approx Y$, if and only if there exists a weak correspondence $\mathcal{R}$ such that $(\GT, Y)\in\mathcal{R}$.
\end{definition}

\begin{theorem}
\label{thm:oc}
Let $\GT$ be reachable from an initial, aware and balanced global type, then $\GT\approx \derived{\GT}$.
\end{theorem}
(\Cref{thm:oc}) follows from
\PURPLE{top-down and bottom-up fidelity}, since $Y<:\derived{\GT}$ implies $\GT \approx Y$, which is mechanical by induction.

\subsection{Further Properties of Systems}
\label{sec:localproperties}

\emph{Local progress} ensures that each configuration in a system can make further actions, unless it is in a final state (or it is final for short). Configuration
$(\p, L, \sigma)$ is \emph{final} if $L\equiv\tend$, where $\equiv$ is the structural equivalence defined by the following rules:
$
\tend \blacktriangleright\bullet\equiv \bullet
\blacktriangleright\tend
\equiv
\tend \blacktriangleright
\tend \equiv \tend
$.
Observe that $L\equiv\tend$ implies that $\p$ has no further actions in $L$ and is committed in all MC in $L$.

\begin{definition}[Local Progress]
$Y$ enjoys progress if for all $Y'$ reachable from $Y$:

\smallskip
\centerline{$
(\p,\ST, \sigma)~\text{ in }~Y'~\text{ is not final } \quad \Rightarrow \quad Y'\myreaches{}\trans{\ell}~~\land~~ \sbj{\ell}=\role p
$}
\end{definition}

The \Cref{lPro} of \Cref{thm:oc} lifts global progress (\Cref{cor:gprog}) to systems.

\begin{corollary}[Local Progress]\label{lPro}
If $\GT$ is initial, aware and balanced then $\derived{\GT}$ enjoys progress.
\end{corollary}

From local progress we further establish \emph{orphan message freedom} (OMF),
which states that every message in a queue can eventually be received (see \Cref{app:om} for details).
%
\begin{restatable}[name=OMF]{theorem}{enabledreceptionlocal}
\label{erl}
Let $Y, (\q, \ST, \sigma)$ be a system reachable from the projection of an initial, aware, and balanced global type.
If $\sigma[\p] = \vec{m_1} \cdot (a_k, \pi) \cdot \vec{m_2}$ with $(a_k, \pi) \not\in \vec{m_1}$, then there exists $Y', (\q, \ST', \sigma')$ reachable from $Y, (\q, \ST, \sigma)$ such that either (1)
$\mathtt{stale}(\pi,\ST')$, or (2)
$Y', (\q, \ST', \sigma')\trans{\p\q?(a_k,\pi)}$.
\end{restatable}

OMF in MST was first formulated~\cite{DenielouY12,DenielouY13} to state that
\emph{final} (terminated) states have empty queues.  Our \Cref{erl} is more
general because it also covers non-terminating systems.  In our setting, we
can state the weaker property of Deni\'{e}lou and Yoshida modulo stale message
purging as: if $Y$ is \emph{final} and derived from the projection of an
initial, aware and balanced $\GT$, then for all $(\p, \ST, \sigma)$ in $Y$,
$\gc(\ST, \sigma) = \emptyset$.  This follows immediately from \Cref{erl}: if
a non-stale message were present in $Y$, by \cref{erl} there would exist $Y'$
reachable from $Y$ where that message is consumed or stale, contradicting the
finality of $Y$.

A stronger version of
OMF~\cite{DBLP:journals/lmcs/ChenDSY17} has been studied in the setting of
(binary) session types with fairness conditions by~\citet{
DBLP:conf/ecoop/PadovaniZ25}.
Due to MC, our \Cref{erl} first differs by considering garbage
collection (case (1) in \Cref{erl}).
Secondly, their \emph{fair termination} condition enforces that every message is
\emph{necessarily} received in every possible execution, whereas our \Cref{erl} ensures that every
enqueued message may always be \emph{possibly} received or garbage collected (e.g., given $\mu \mathtt t .\, \GactShort{p}{q}{a.\, \mathtt t }$, it is always possible for $\q$ to receive 
but the unfair execution where 
$\p$ sends forever while $\q$ never receives is also possible).
We leave to future work an investigation of
fairness conditions \`{a} la~\citet{DBLP:journals/jlap/CicconeDP24,DBLP:conf/ecoop/PadovaniZ25} for asynchronous MC in MST. %
\cref{sec:related} gives further comparisons with \citet{DBLP:journals/jlap/CicconeDP24}.

\section{Implementation and Further Examples}
\label{sec:implementation}

\subsection{Toolchain Implementation}
\label{lab:toolchainimpl}

\paragraph{Global protocol validation}
We have implemented \mMST{} as an extension to the Scribble protocol
language~\cite{DBLP:conf/tgc/YoshidaHNN13,DBLP:conf/fase/HuY16}.
The practical syntax for our MC is as demonstrated in \Cref{sec:overview}.
Following our formal theory, our implementation allows MC to be combined with
the standard MST constructs for regular directed choice (non-mixed
branch/select) and recursion to express a range of patterns and DS constructs
found in practical applications.

Our toolchain internally translates Scribble specifications to a representation
based on our formal definitions and \emph{syntactically} checks the source
protocol for
(a) the well-formedness of committing message labels (\Cref{well-formed}),
(b) awareness (\Cref{def:aware}) and balance (\Cref{wellset}),
and (c) projectability (\Cref{sec:projection}) onto all roles.
In our current implementation, checking (b) syntactically means inferring role
occurrences and dependencies between roles as written in the source protocol
without (semantically) unfolding recursive types.
This is sound: such syntactically inferred dependencies conservatively
imply our formal conditions.

\emph{Code generation.}
The toolchain implements projection of global types to local types following
the formal theory.
Each local type is then translated to an event-driven finite state machine
(EFSM) representation.
As described in \cref{sec:moti}, from the EFSM we generate two Erlang modules
for each role: a \emph{role module} (RM) and a \emph{callback module} (CM).
In \cref{fig:exceptionB}, we showed the CM for role \CODE{B}.
In \cref{fig:apiB}, we show the corresponding RM for \CODE{B}.
The RM builds on \CODE{gen\_statem}, providing an EFSM structure to manage state transitions and message passing in accordance with the protocol. Specifically, it provides state function definitions (\cref{fig:apiB}, lines 27-40), and callback specifications (\cref{fig:apiB}, lines 2-8) that must be implemented by the CM to handle incoming messages and internal events.
It also provides functions for sending messages (\cref{fig:apiB}, lines 17-23). Moreover, it leverages \CODE{gen\_statem}’s asynchronous selective receive to defer out-of-order events by returning a postpone action (e.g. \CODE{{keep\_state, Data, [postpone]})} (lines 27-28), requeuing them until the state machine transitions into a state that can properly handle them.
Following Erlang convention, the RM excludes any application-specific logic, which instead is left to callbacks in the CM.

\begin{figure}[t]
\begin{center}
\begin{tabular}{c}
    \begin{lstlisting}[numbers=left,language=erlang, numbersep=3pt, basicstyle=\fontsize{8}{8}\ttfamily]
% -------- Callback specifications --------
-callback s5(term(), {atom()}, state_data()) -> {stop, normal, state_data()}.
-callback s1(term(), {atom()} | {pid(), {term()}}, state_data()) -> 
  {next_state, s5, state_data(), [{next_event, internal, {'TOc'}}]} | {keep_state, state_data()} | 
  {keep_state, state_data(), [postpone]} | {next_state, s2, state_data(), [{next_event, internal, {a3}}]}.
-callback s2(term(), {atom()}, state_data()) -> 
    {next_state, s3, state_data(), [{next_event, internal, {a4}}]} | {keep_state, state_data()}.
-callback s3(term(), {atom()}, state_data()) -> {stop, normal, state_data()}.
%% -------- Record and type definitions for maintaining state --------
-record(data, { a_pid :: pid(), c_pid :: pid() }).  %% Process identifiers for roles A and C.             
%% -------- Send helpers --------
send_s1_TOa(APid, Data) -> gen_statem:cast(APid, {self(), {'TOa'}, Path}).
send_s5_TOc(CPid, Data) -> gen_statem:cast(CPid, {self(), {'TOc'}, Path}).
send_s3_a4(APid, Data) -> gen_statem:cast(APid, {self(), {a4}, Path}).
send_s2_a3(CPid, Data) -> gen_statem:cast(CPid, {self(), {a3}, Path}).
%% ---------- Callback function definitions delegate the processing to the CM ----------
s1(EventType, {APid, {a1}, Pi}, Data) ->
  case stale(Pi, left) of
      true -> {keep_state, Data};
      false -> Side = CallbackModule:s1(EventType, {APid, {a1}}, Data)
      case Side of {next_state, s6, _} -> commit_entry(?MC1, left), Side;
                   {next_state, s3, _} -> commit_entry(?MC1, right), Side
      end;
  end;
s1(EventType, {'TOa'}, Data) -> Side = CallbackModule:s1(EventType, {'TOa'}, Data)
      case Side of {next_state, s6, _} -> commit_entry(?MC1, left), Side;
                   {next_state, s3, _} -> commit_entry(?MC1, right), Side
      end.
s5(EventType, 'TOc', Data) -> CallbackModule:s5(EventType, 'TOc', Data).
s2(EventType, a3, Data) -> CallbackModule:s2(EventType, a3, Data).
s3(EventType, a4, Data) -> CallbackModule:s3(EventType, a4, Data).
    \end{lstlisting}
    \end{tabular}
\end{center}
	\caption{Extract from the RM generated for \CODE{B} in \CODE{Timeout}, implementing its protocol- and role-specific \CODE{gen\_statem} behaviour and declaring the callback functions to be implented by the CM.}
	\label{fig:apiB}
\end{figure}

In mixed choices, participants may receive messages from non-selected branches until all participants become committed.
To prevent protocol violations stemming from these stale messages, the toolchain implements event queue management and state data tracking within the \CODE{gen\_statem} framework.
The implementation follows the theory, and derives the stale messages to be purged from the global type, as defined in \cref{eq:gc}.
\CODE{gen\_statem} maintains an event queue where incoming events are enqueued and dispatched sequentially based on arrival order and their priority, with internal events being prioritised over external events.
Each incoming MC message carries a path \CODE{Pi} recording the sides taken through active MCs. The RM keeps a local commitment map, and as illustrated in \cref{fig:apiB}, lines 22-26, when an event is received the RM module purges the event if the path \CODE{Pi} is stale, otherwise it forwards the event to the CM module.

\emph{Runtime requirements.}
Our implementation uses Erlang's core features for process spawning and asynchronous message passing, and the built-in \CODE{gen\_statem} library for executing event-driven finite state machines (EFSMs). However, a runtime for our theory can be implemented in any setting that supports event-driven concurrency and asynchronous messaging, over which a framework for session-based EFMSs can be readily developed~\cite{DBLP:journals/pacmpl/VieringHEZ21} if not directly supported as in Erlang.  Our protocol validation and projection is independent of Erlang.


\subsection{Expressiveness by Examples}
\label{sec:evaluation}

\Cref{tab:examples} summarises a range of examples from MST literature that we
have extended with mixed choice (MC) to support features such as timeouts and
exceptions.
Our MC allows use cases like SMTP to be expressed more fully than in prior MST
systems due to supporting timeouts (prior works simply disregarded that aspect).
Our examples test the static elements of protocol validation, projection, EFSM
translation and Erlang module generation.
We also implemented minimal but functional skeleton programs for each role of
the examples to test the runtime I/O and event handling dynamics, and embedded
MC mechanisms such as stale message purging.
However, our minimal implementations generally do not perform the full
application logic of the examples.

\begin{table}[t]
\footnotesize
\caption{
MST examples extended with asynchronous mixed choice (MC).
}
\label{tab:examples}
\vspace{-3mm}
\begin{tabular}{lllllllll}
\multicolumn{9}{l}{M=Multiparty~~B/S=Branch/Select~~Rec=Recursion~~nMC=nested MC~~ndMC=non-directed MC~~GC=purging}
\\[2pt]
Example & M & B/S & Rec & MC & nMC & ndMC & GC & Source
\\
\hline
Calculator & $\checkmark$ & $\checkmark$ & & $\checkmark$ & & & $\checkmark$ & [Hu and Yoshida, 2016]
\\
CircuitBreaker & $\checkmark$ & $\checkmark$ & $\checkmark$& $\checkmark$ & & $\checkmark$ & $\checkmark$ & [Lagaillardie, Neykova and Yoshida 2022]
\\
DistributedLogging & & $\checkmark$ & $\checkmark$ & $\checkmark$ & & & $\checkmark$ & [Lagaillardie, Neykova and Yoshida 2022]
\\
Fibonacci & & $\checkmark$ & $\checkmark$ & $\checkmark$ & & & $\checkmark$ & [Hu and Yoshida, 2016]
\\
SMTP & & $\checkmark$ & $\checkmark$ & $\checkmark$ & $\checkmark$ & & $\checkmark$ & [Hu and Yoshida, 2016]
\\
TwoBuyer & $\checkmark$ & $\checkmark$  & & $\checkmark$ & $\checkmark$ & $\checkmark$ & $\checkmark$ & [Honda, Yoshida, Carbone, 2008]
\\
TravelAgency & $\checkmark$ & $\checkmark$ & $\checkmark$ & $\checkmark$ & & & $\checkmark$ & [Hu, Yoshida and Honda 2008]
\\
OnlineWallet & $\checkmark$ & $\checkmark$ & $\checkmark$ & $\checkmark$ & & $\checkmark$ & $\checkmark$ & [Neykova, Yoshida and Hu, 2013]
\end{tabular}
\end{table}

\emph{Distributed system constructs.}
A key motivation for \mMST{} is to provide a core construct that can express a
range of important constructs -- building blocks of many practical distributed
systems -- that were previously only supported by bespoke and
disparate MST extensions.
These include exceptions~\cite{DBLP:conf/concur/CarboneHY08},
interrupts~\cite{DBLP:journals/fmsd/DemangeonHHNY15},
timeouts~\cite{DBLP:conf/coordination/PearsBK23,DBLP:conf/ecoop/HouLY24}, and
failure handling~\cite{DBLP:conf/esop/VieringCEHZ18}.
A basic timeout pattern was illustrated in \Cref{sec:overview}; exceptions are
similar.
We illustrate interrupts and failure handling below.

\smallskip
\emph{Failure handling}.  We give an example of MC drawn from the major topic of failure
handling and fault-tolerance in
MST~\cite{DBLP:conf/esop/VieringCEHZ18,DBLP:journals/pacmpl/VieringHEZ21,DBLP:conf/ecoop/BarwellHY023,DBLP:journals/lmcs/PetersNW23,DBLP:conf/forte/BrunD24}.
%
%
The existing work typically models application protocols assuming that failure
detection (FD) is implicitly provided by the runtime infrastructure.
With MC as a core construct, we can now explicitly model such mechanisms and
how role behaviours may depend on them.
\Cref{fig:failure} (left) gives a small example, where Worker \CODE{W} is a
failure-prone role and the observer \CODE{FD} represents the FD service, based
on the use cases of
\citet{DBLP:conf/esop/VieringCEHZ18,DBLP:journals/pacmpl/VieringHEZ21} featuring
heartbeat-based FD.
We specify it as a recursive MC where in each iteration (i.e.,
unfolding of the recursive type) \CODE{FD} waits to receive a heartbeat
(\CODE{HB}) from \CODE{W} with the option to send a \CODE{Crash}
notification to \CODE{M} (e.g., upon a timeout, connection error, or other
failure condition).
The 
annotation 
on L\ref{ln:failed} is given by a simple variant of our core theory where
protocol validation prohibits further usage of a role considered failed (i.e.,
\CODE{W} cannot occur again in the continuation of the protocol after
L\ref{ln:failed}).
The \CODE{Timeout} message is useful if the connection is actually still live,
allowing \CODE{W} to handle its own (reported) demise, but can be considered
redundant otherwise.
%
%
Following our theory, protocol validation allows the awareness of \CODE{W}
in the LHS to be transitively relayed from \CODE{FD}
to \CODE{M} via \CODE{OK}, then to \CODE{W} via \CODE{more}.

\begin{figure}[t]
\begin{minipage}[t]{0.5\textwidth}
\begin{lstlisting}[numbers=left, numbersep=3pt, language=scribble, basicstyle=\fontsize{8}{10}\ttfamily]
global protocol FailH(role M, role W, role FD) {  
  init(Data) from M to W;
  rec X {
    mixed {  // LHS
      HB() from W to FD;
      OK() from FD to M;
      result(Data) from W to M;
      more(Data) from M to W;
      continue X;
    } or {  // RHS
      Timeout() from FD to W; @'failed W'  (*\label{ln:failed}*)
      Crash()  from FD to M;
} } }
\end{lstlisting}
\end{minipage}%
\begin{minipage}[t]{0.45\textwidth}
\begin{lstlisting}[numbers=left, numbersep=3pt, language=scribble, basicstyle=\fontsize{8}{10}\ttfamily]
@'explicit-observer-left-commits'
global protocol Interr(role P, role Q) {
  mixed {  // LHS
    Start() from Q to P;
    rec X {
      choice at Q {
        More() from Q to P; continue X;
      } or {
        Stop() from Q to P*; (*\label{ln:explicit}*) // P commits
        Ack() from P to Q;
    } }
  } or {  // RHS
    Interrupt() from P to Q; } }
\end{lstlisting}
\end{minipage}
\caption{Heartbeat-based failure detection (left), and an asynchronous interrupt (right).
}
\label{fig:failure}
\end{figure}

\smallskip
\emph{Interrupts.}
%
%
%
Recall the
interrupt patterns
using MC
in \Cref{ex:interrupt}.
The interrupt pattern can be expressed in our Scribble extension as in
Figure~\ref{fig:failure} (right).

We make an option for the user to explicitly indicate the
committing interactions for observers in the LHS of an MC (as opposed to the
observer implicitly committing on the first such message received as in the
core theory).
The explicit committing interaction is marked by a \CODE{*} on the observer
role occurrence, e.g., \CODE{P*} on Line~\ref{ln:explicit}.

\subsection{Case Study: RabbitMQ}

We apply our toolchain to a non-trivial, real-world case study, RabbitMQ. The toolchain can generate code for each participant in the chosen subset of the AMQP protocol, while also allowing for interoperability: the code generated for one participant (e.g., Consumer) can interoperate with pre-existing RabbitMQ implementations.

The Advanced Message Queuing Protocol\footnote{\url{https://www.amqp.org/}} (AMQP) is an open-standard protocol designed for message-oriented middleware. RabbitMQ, a widely-used open-source message broker adhering to the AMQP standard, leverages the Erlang client library \CODE{amqp\_client}\footnote{\url{https://github.com/rabbitmq/rabbitmq-server/tree/main/deps/amqp_client}} to facilitate interaction between Erlang and Elixir applications and RabbitMQ nodes.
Within this ecosystem, the \CODE{amqp\_selective\_consumer} module, implemented alongside its behaviour \CODE{amqp\_gen\_consumer}, plays a crucial role in managing message consumption with precise control over delivery and cancellation. Notably, \CODE{amqp\_selective\_consumer} is already implemented as a state machine using Erlang's \CODE{gen\_server} behaviour; however, this implementation is relatively ad-hoc. We {replace} the \CODE{amqp\_selective\_consumer}, and its implemented behaviour \CODE{amqp\_gen\_consumer} from the RabbitMQ Erlang client library, \CODE{amqp\_client} with a behaviour and a callback module generated from a session type representation of the protocol.

We {model} the relevant part of the AMQP protocol in Scribble, \cref{fig:amqp}, capturing the interactions between the consumer \CODE{C}, channel \CODE{H}, and server \CODE{S} roles for selective message delivery. The consumer registers with the channel and initiates a \CODE{basic\_consume} request, which is forwarded to the server. The server responds with \CODE{basic\_consume\_ok}, confirming the consumer's subscription. The \CODE{SelectiveMessageDelivery} recursive block models ongoing message exchanges, using an MC construct to represent message processing on one hand and cancellation on the other.  MC is essential for accurately modelling the concurrency and asynchronous behaviour in AMQP. It allows the protocol to express that either the server may deliver a new message -- \CODE{basic\_deliver}, or the consumer may choose to cancel the subscription -- \CODE{basic\_cancel}.

Our toolchain validates this global type, projects it to a local type for each role, and constructs an EFSM representation of each local type. For the consumer role, \CODE{C}, the toolchain generates a correct-by-construction behaviour and callback modules.
The generated modules replace \CODE{amqp\_selective\_consumer} and \CODE{amqp\_gen\_consumer}, implementing the same interfaces expected by other components. All features of \CODE{amqp\_selective\_consumer} are preserved, including message consumption, and cancellation. We use RabbitMQ's existing test suite to validate the new \CODE{amqp\_selective\_consumer} implementation.

\begin{figure}[t]

\begin{tabular}{l|l}
\begin{minipage}{0.5\textwidth}
{\begin{lstlisting}[language=Scribble, basicstyle=\fontsize{7}{10}\ttfamily]
// roles: C =$~$Consumer, H =$~$Channel, S =$~$Server
global protocol AMQP(role C, role H, role S) {
  register_default_consumer() from C to H;
  basic_consume(consumer_tag, nowait) from C to H;
  basic_consume(consumer_tag, nowait) from H to S;
  basic_consume_ok(consumer_tag) from S to H;
  basic_consume_ok(consumer_tag) from H to C;
  rec SelectiveMessageDelivery {
    mixed {
      basic_deliver(consumer_tag, delivery_tag,
          exchange, routing_key) from S to H;
      // Continued on right column...
\end{lstlisting}}
\end{minipage}
&
\begin{minipage}{0.49\textwidth}
{\begin{lstlisting}[language=Scribble, basicstyle=\fontsize{7}{10}\ttfamily]
  process_message() from H to C;
  basic_deliver(consumer_tag, delivery_tag,
      exchange, routing_key) from H to C;
  processing_complete(delivery_tag) from C to H;
  update_delivery_state(delivery_tag) from H to S;
  continue SelectiveMessageDelivery;
} or {
  basic_cancel(consumer_tag, nowait) from C to H;
  basic_cancel(consumer_tag, nowait) from H to S;
  basic_cancel_ok(consumer_tag) from S to H;
  basic_cancel_ok(consumer_tag) from H to C;
} } }
\end{lstlisting}}
\end{minipage}
\\
\end{tabular}
\caption{Subset of AMQP using a recursive MC in our extended Scribble.}
\label{fig:amqp}
\end{figure}

\section{Related Work, Limitations and Future Work}
\label{sec:related}

This paper presents our \emph{asynchronous multiparty} session type system with a core construct for \emph{mixed} choice and safe automated purging of stale messages.
\citet{DBLP:journals/corr/abs-2604-06872} present another recent work on asynchronous mixed choice.
Unlike our approach, their system does not have any notion of identifying stale messages or automated purging.
They instead require processes to eventually and explicitly receive all messages, including criss-crossed and (what we identify as) stale messages.
Their theoretical system bypasses projection by directly typing processes against global types, both modelled using coinduction and omitting syntactic constructors for recursive types, whereas our system projects global types to traditional inductive local types in a toolchain for practical programming in Erlang.

There are two further main lines of related work in the literature: session types with synchronous mixed choices, and MST extensions for bespoke exception-like constructs.

The literature includes several works on \emph{synchronous} mixed choices for binary \cite{DBLP:conf/esop/VasconcelosCAM20, DBLP:journals/corr/abs-2004-01324, DBLP:journals/tcs/CasalMV22} and multiparty~\cite{DBLP:journals/corr/abs-1203-0780, DBLP:conf/esop/JongmansY20,
DBLP:conf/ecoop/JongmansF23, DBLP:journals/corr/abs-2508-13616, DBLP:journals/corr/abs-2405-08104, DBLP:journals/pacmpl/HinrichsenQB26} session types.
In synchronous settings, mixed choices behave very similarly to regular (\emph{non}-mixed) choice: both can be modelled by an
\emph{atomic} reduction step (e.g., $G_1 \mathop{+} G_2 \mathbin{\rightarrow}
G_i'$ for $i \mathop{\in} \{1, 2\}$ and $G_i \mathbin{\rightarrow} G_i'$).
Reasoning about safety of \emph{asynchronous} mixed choices (MC) and
disciplining the inherent race conditions is a substantially different
endeavour, as shown by our developments through the
design of our MC,
static validation,
runtime mechanisms
and metatheory.
We note the work of \citet{DBLP:conf/coordination/PearsBK23} on asynchronous \emph{binary}
session types with mixed choice, where each choice has a \emph{timing} constraint.
Their system allows choices with a mix of input and
output actions, but they statically enforce that actions in
different directions are never viable at the same point in \emph{time}.
Similar approaches based on timed session types~\cite{DBLP:conf/ecoop/HouLY24} share this limitation.
By contrast, we deal with bona fide asynchronous MC where both communication
directions w.r.t.\ all pairs of participants are concurrently viable.

This paper aims to distill the \emph{essence} of mixed choice in
asynchronous (M)ST, various aspects of which were studied in several areas; e.g.,
interaction exceptions~\cite{DBLP:conf/concur/CarboneHY08} (binary),
interaction handlers~\cite{DBLP:journals/mscs/CapecchiGY16} (with synchronous triggers),
interruptible blocks~\cite{DBLP:journals/fmsd/DemangeonHHNY15} (where default/interrupted blocks share common continuations),
and
internal exceptions~\cite{DBLP:journals/pacmpl/FowlerLMD19} (local failure control, no type construct).
We have shown how MC can express multiparty asynchronous timeouts
and interrupts; exceptions are similar.

A crucial area for real-world applicability of session types
is support for failure handling, where aspects of
mixed choice arise intrinsically.
\citet{DBLP:conf/esop/VieringCEHZ18,DBLP:journals/pacmpl/VieringHEZ21}
developed an MST system with try-handle constructs for handling
\emph{partial} failures of a protocol due to (suspected) participant crashes.
In these works a process may be faced with a
choice between outputs in the normal protocol flow mixed with potential inputs
related to failure detection/notification.
The (supposed) failure of a participant rules out all interactions with that
participant thereafter. By contrast, our MC has a
finer-grained notion of commitment that is \emph{per instance} of an MC (not per protocol).
\Cref{sec:evaluation} demonstrated such a failure handling pattern using our MC.
\citet{DBLP:conf/ecoop/BarwellHY023} present an approach to failure handling
in MST that avoids syntactic extensions
(cf.~\cite{DBLP:conf/concur/BettiniCDLDY08long}).
Their system designates a special message label $\mathsf{crash}$ to denote
failure of the ``sending'' participant,  which otherwise behaves as
a regular label in choice constructs.  E.g.,
$\p \mathbin{\rightarrowtriangle} \q : \{ a. \,\GREY{\GT_1}, \mathsf{crash}. \GREY{\GT_2}\}$
specifies that $\q$ awaits either message $a$ from $\p$ or
notification of $\p$'s crash.
As in the related work above, crash events may occur concurrently with the I/O actions of the main protocol flow and their system incorporates some related machinery (e.g., queue cleaning);
however, neither their user-level types nor processes have explicit constructs
for mixed choice.
See \Cref{asec:related} for further notes on the above works.

Communicating systems featuring mixed choice have been studied
outside of session types.
Communicating automata~\cite{10.1145/322374.322380} for instance do not rule out
mixed choices but progress is, in general, undecidable.
\citet{DBLP:conf/icse/LangeNTY18} present a tool that infers
\emph{behavioural types}~\cite{DBLP:journals/csur/HuttelLVCCDMPRT16}
for channel-based, shared-memory concurrency programs in Go.
Their system permits mixed choice for the \CODE{select} construct and assumes
finitely buffered channels, enabling decidable model
checking (e.g., deadlock-freedom) of the inferred types.
Our work instead supports asynchronous communication over unbounded channels.

Recently, \citet{DBLP:conf/cav/LiSWZ23} presented the first sound \emph{and complete}
projection method for global types generalised with \emph{sender-driven} choice
($\p \mathbin{\rightarrow} \{\q_i:m_i.G_i\}_{I}$).
Their automata-theoretic approach separates local machine
synthesis from implementability checking, but, like \citet{DBLP:conf/concur/MajumdarMSZ21},  do not support mixed choice.
By contrast, our work builds on  classical MST with regular
(non-mixed) choices restricted to directed choice and syntactic projection.
This is because our asynchronous MC is influenced by the other works discussed
above on, e.g., interrupts and failure handling.

\citet{DBLP:journals/jlap/CicconeDP24} presented an MST system for \emph{fair
termination} in a synchronous session $\pi$-calculus without mixed
choice.
Their system guarantees processes (with multiple sessions) will fairly
terminate by combining: a validation that session types always potentially
terminate, a restricted fusing of session initiation and process spawning
\cite{DBLP:journals/jfp/Wadler14,DBLP:journals/mscs/CairesPT16}, a
notion of ranking that limits processes to finite behaviours, a
liveness-preserving subtyping relation, and a fairness assumption on executions
(potential termination leads to actual termination).
There may be connections between our notion of clear termination, that is per
MC, with their notion of (whole) protocol termination; at present, our system
differs in that we allow a session to be unbounded provided each MC individually
satisfies awareness.
\citet{DBLP:conf/ecoop/PadovaniZ25} recently developed fair termination for
asynchronous binary sessions
(as mentioned in Sec.~\ref{sec:localproperties}).

\paragraph{Limitations and future work.}
Our system builds on classical MST with regular choices restricted to directed
choice, and syntactic projection and
merge~\cite{DBLP:conf/ecoop/Stutz23,DBLP:journals/pacmpl/ScalasY19}.
This yields a projection that is sound but not complete in the base MST
constructs (directed choice, recursion), let alone with MC.
In future work, we plan to investigate (conservative forms of) mixed
choice in the generalised automata-theoretic setting of
\citet{DBLP:conf/cav/LiSWZ23}.
The challenges include reconciling their language-based approach
with our mechanisms (e.g., stale message purging) and metatheory (e.g.,
operational correspondence and preservation of projection).
Another issue is that their projection yields state machines that are more
general than local types: accepting
(terminating) states may have  outgoing transitions, and choice branches may be
unbalanced across roles.
Such generality must be reconciled with the objective of our
system (and classical MST) that all projected behaviours, including MCs,
be realisable as fully distributed processes; e.g., we typically aim to rule out
protocols where termination could be non-deterministic (cf.\ Def.~\ref{def:aware}).
\PURPLE{Building on the above, \citet{DBLP:journals/corr/abs-2501-16977}
recently extended projection to a larger class of automata-based global
specifications, but show that the projectability of \emph{mixed} choice in
their more general setting is undecidable.
Their negative result motivates pragmatic approaches to supporting mixed
choice such as in this paper.
}
The automata-based system of \citet{DBLP:conf/cav/LangeY19} may also provide
avenues for lifting directed choice and generalising our protocol validation.

One of our present limitations 
is that our system does not incorporate \emph{delegation}~\cite{HondaYC08,DBLP:conf/concur/BettiniCDLDY08long}.
Considering our MC concepts (e.g., stale message purging and path
identifiers) in a setting with delegation is a topic for future work.

An interesting question is how fair multiparty termination
\cite{DBLP:journals/jlap/CicconeDP24} may be extended to asynchronous mixed
choices.
One direction could be to investigate fair termination for our notion of MC
global/local types and formalise a process-level language.
Unlike the `multithreaded' $\pi$-calculi in the mentioned works, however, the
model of concurrency in our practical Erlang setting is
event-driven~\cite{DBLP:conf/ecoop/HuKPYH10,DBLP:journals/pacmpl/VieringHEZ21}
and does not have a specific primitive for fusing session initiation and
process spawning.
Adapting their notion of static typing of process ranks may be a challenge in
languages such as Erlang.

\section*{Data-Availability Statement}

The source code of our toolchain, examples and RabbitMQ case study is
available online.\footnote{\url{https://github.com/rhu1/scribble-gt-scala/tree/artifact}}
It will be submitted for review as an artifact.

\begin{acks}
This work was partially funded by EPSRC project EP/T014512/1, EP/T014628/1 (STARDUST).
\end{acks}
\bibliographystyle{ACM-Reference-Format}
\bibliography{main}

%
 \pagebreak
 \appendix

 \begin{tabular}{ll}
 \Cref{asec:related}
 & Additional notes on related work (cf.\ \Cref{sec:related}).
 \\
 \Cref{asec:multiple}& Additional notes on handling multiple sessions
 \\
 \Cref{app:committing}
 & Full definitions of committing and non-committing label sets (cf.\ \Cref{sec:committing}).
 \\
 \Cref{asec:progress}
 & Progress of global types: Omitted definitions and proofs (cf. \Cref{sec:globalprogress}).
 \\
 \Cref{asec:localtypes}
 & Local types: Full definitions (cf.\ \Cref{sec:localsemantics,sec:projection}).
 \\
 \Cref{app:correspondence}
 & Operational correspondence: Omitted definitions and proofs (cf.\ \Cref{sec:correspondence}).
 \\
 \Cref{app:orphan}
 & Orphan messages: Omitted definitions and proofs (cf.\ \Cref{sec:localproperties}).
 \end{tabular}

 \section{Additional Notes on Related Work}
\label{asec:related}

\paragraph{Mixed choice in \textbf{binary} sessions.}

\citet{DBLP:conf/esop/VasconcelosCAM20} developed a binary session type system
for synchronous sessions with mixed choice.
\citet{DBLP:journals/corr/abs-2004-01324,DBLP:journals/tcs/CasalMV22} further
established a type preservation property for mixed sessions, an embedding of
classical sessions~\cite{Vasconcelos12} into the mixed sessions, and a partial
encoding in the reverse direction.
All of these works are limited to \emph{binary} sessions with
\emph{synchronous} communication, where mixed choice subsumes non-mixed
choice.
By contrast, mixed choice in (safe) asynchronous sessions gives rise to
(transient) \emph{inconsistencies} between the views of the protocol from
different participants, which are precluded by the power of
synchrony.

\citet{DBLP:conf/coordination/PearsBK23} developed a system with binary
session types for mixed choice with \emph{timing} constraints for
\emph{asynchronous} sessions.
Their system syntactically allows two-party choices with a mix of input and
output actions with the (statically enforced) restriction that actions of
different directions are \emph{never} viable at the same point in \emph{time}.
By contrast the entire purpose of our system is to permit and safely deal with
{truly} asynchronous mixed choices where both communication directions
w.r.t.\ all pairs of participants \emph{are} concurrently viable.

In the binary settings of the above works, note that the notion of
\emph{fidelity} is in a sense moot since two-party local types effectively
coincide with their (conceptually) associated global type.
By constrast, multiparty local type projections are, in general,
\emph{partial} w.r.t.\ the global type, necessitating our developments to
establish an operational correspondence between asynchronous global and local
types featuring mixed choice
(cf.\ \Cref{sec:localtypes}).

\paragraph{Mixed choice in \textbf{synchronous} multiparty sessions.}

In MST, existing work on mixed choice has been limited to \emph{synchronous}
sessions.
A consequence is that the constructs for mixed choice in all the following
works are \emph{symmetric} in nature, unlike the \emph{asymmetric} construct
we have developed in this paper for safety in asynchronous sessions.

\citet{DBLP:journals/corr/abs-1203-0780} developed a set-theoretic semantic
framework for global types and multiparty sessions based on a trace semantics.
Although their global types include the syntax for $G \vee G'$, the choice
constructs in their local types are further syntactically constrainted to
$T \oplus T'$ or $T + T$ representing the traditional internal-\emph{only}
and external-\emph{only} choices.
Moreover, their trace semantics does not model the intermediary (and
potentially inconsistent) states that may arise from asynchronous
mixed choice, and each trace element denotes an \emph{atomic} interaction
between the sender and receiver, in contrast to the intermediary states
represented by, e.g., $\rightsquigarrowtriangle$ types (e.g.,
\cref{sec:globaltypes}) that expose the inherent race conditions.

\citet{DBLP:conf/esop/JongmansY20} studied weak bisimulation between global
and local types featuring the `$+$' operator for choice in a synchronous
system.
They informally remark on encoding a form of asynchrony by representing
buffered channels explicitly as roles, e.g., encoding an asynchronous
$p \rightarrowtriangle_a q$
interaction as the pair of synchronous interactions
$p \rightarrowtriangle_s b_{pq} . b_{pq} \rightarrow_s q$
where $b_{pq}$ is a role that represents a message buffer from
$p$ to $q$.
However, this does not amount to asynchrony as we model in this paper because
a mixed choice becomes
$(p \rightarrowtriangle_s b_{pq} . \GREY{G_1'})
+
(q \rightarrowtriangle_s b_{qp} . \GREY{G_2'})$
where by their (synchronous) semantics either output action by $p$ or $q$
instantly commits \emph{all} roles to that branch and precludes the other
action from occurring, thus precluding any intermediary states and the
problematic race conditions.

\citet{DBLP:conf/ecoop/JongmansF23} present rules for extending a similar
simulation approach to processes featuring the same formulation of `$+$' for
(mixed) choice as~\citet{DBLP:conf/esop/JongmansY20}.
Their system is again synchronous with the same operational semantics; all the
preceding comparison points with~\citet{DBLP:conf/esop/JongmansY20} apply
here.
\citet{DBLP:journals/corr/abs-2405-08104} present a typing system between
local types and processes with mixed choice.
Their system is synchronous and their work does not include global types nor
operational correspondence between global and local types as in our system.

We note, however, that our present paper shares some simplifications with the
above works.
For example, our system focuses on modelling the semantics of a \emph{single}
session; this is the same as in all of the MST works mentioned above,
Consequently, none of these works nor our present paper has studied
mixed choice for multiparty sessions in the presence of session
\emph{delegation}~\cite{HondaYC08}, where a session message carries another
session \emph{channel} as a payload.
\citet{DBLP:journals/corr/abs-2508-13616} extend a Simple MultiParty Session~\cite{DBLP:journals/lmcs/DagninoGD23,DBLP:journals/jlap/BarbaneraBD25a} system with mixed choice for synchronous sessions.
For tractability, they rely on a notion of modularisation, akin to subsessions, to restrict the scope in which participants can interact via mixed choice.  This approach may be related to the subsession-based failure handling pattern in~\cite{DBLP:journals/pacmpl/VieringHEZ21}.

\citet{DBLP:journals/pacmpl/HinrichsenQB26} present Mixtris, a higher-order separation logic for multiparty message passing, including mixed choice, between higher-order imperative programs in a synchronous shared memory setting.
They show how to implement mixed choice over channels based on synchronisation cells, with a mechanisation of their system and results using Rocq and Iris~\cite{DBLP:conf/popl/JungSSSTBD15}.

\paragraph{Exception-like communications patterns in session types.}

As discussed in~\Cref{sec:intro}, mixed choice is at the heart of many crucial
communication patterns in real-world distributed systems.
One motivation for our paper is to distil the essence of mixed choice in
asynchronous MST, various aspects of which have been studied in several
prominent areas.

\smallskip
\textbf{Exceptions and interrupts.}
\citet{DBLP:conf/concur/CarboneHY08} developed a \emph{binary} session type
system for processes with \emph{interaction exception} handlers.
Their \emph{try-catch} type $\alpha\{\![\beta]\!\}$ specifies a normal
protocol $\alpha$, and an exceptional protocol $\beta$
triggered by throwing an exception.
This allows situations where a process may be faced with, say, an input in its
normal protocol but an output in the triggered exception handler, and dually
for its binary peer.
Their approach introduces a notion of \emph{meta reduction}
$P \mathop{\Searrow} P'$ to deal with propagating exceptions, nested exceptions
and queue cleaning, which is related to our garbage collection.

\citet{DBLP:journals/mscs/CapecchiGY16} extend MST with
\emph{interaction handlers} that require reasoning about concurrency of
exceptions raised by multiple peers.
Their approach introduces an \emph{exception environment} $\Sigma$ into the reduction
relation $\Sigma \vdash P \rightarrow \Sigma' \vdash P'$,
where $\Sigma$ records the raised exceptions.
It is notable, however, that $\Sigma$ is operated on \emph{synchronously} and
\emph{atomically} by all multiparty participants, i.e., raising an exception
notifies all parties instantly.
%

\citet{DBLP:journals/fmsd/DemangeonHHNY15} proposed types for
\emph{interruptible blocks} in multiparty sessions.
In their type
$\{\!| G |\!\} \langle \ell\text{ by }r\rangle; G'$,
the $G$ can be interrupted
by $r$ at any point.
Whether or not $G$ is interrupted it is followed by $G'$, unlike our present
paper that supports safe mixed choice between \emph{different} continuations.
Their restriction permits a \emph{synchronous} operational semantics for
\emph{global} types (the \emph{derivative} relation) with \emph{atomic} steps
for both interactions and the signalling of interrupts.

\citet{DBLP:journals/pacmpl/FowlerLMD19} develop a functional language with
binary session types and primitives for raising and handling exceptions.
Their system focuses on exceptions as a local control flow feature
\emph{external} to their session types, i.e., there is no \emph{type}
construct for exceptions, unlike our mixed choice in this paper or all the
other works discussed above.

\citet{DBLP:conf/esop/BrunD23,DBLP:conf/forte/BrunD24} recently developed
timeout branches for a low-level MST where messages can be arbitrarily
reordered and lost, which is a radically different communication model than in
our paper and the other works mentioned above.
Consequently, their notion of safety is weaker than in standard MST: they
enforce that timeout branches are always defined, and ensure that \emph{if} a
message is received then it has the expected payload.

\smallskip
\textbf{Failure handling.}
A crucial area towards the application of session types to real-world
distributed systems is support for failure handling, where aspects of
mixed choice arise naturally.

\citet{DBLP:conf/esop/VieringCEHZ18} develop an asynchronous MST system with a
\emph{try-handle} construct for handling \emph{partial} failures of a
protocol due to participant crashes.
They target distributed system with central \emph{coordinators} for
\emph{reliable} failure detection.
By contrast, \citet{DBLP:journals/pacmpl/VieringHEZ21} developed an
asynchronous MST system that extends
\emph{subsessions}~\cite{DBLP:conf/concur/DemangeonH12} with handlers for
\emph{peer}-based \emph{unreliable} failure detection.
These works both treat patterns where a process may be faced with a
choice between outputs in the normal protocol flow mixed with potential inputs
related to failure detection/notification.
Participant failure handling is an instance of mixed choice where
the (supposed) failure of a participant rules out all interactions with that
participant henceforth, unlike our mixed choice in this paper where the
fine-grained notion of commitment is \emph{per instance} of a mixed choice (a
participant failure could be considered a coarse-grained commitment to
\emph{all} active and future instances of a particular mixed choice).

\paragraph{Practical MST frameworks based on API generation from global protocols.}

Our paper shares the motivation of a range of works on developing MST
frameworks for practical applications, e.g.,
\cite{10.1007/978-3-642-40787-1_8,DBLP:journals/pacmpl/CastroHJNY19,DBLP:journals/pacmpl/00020HNY20}.
Like the mentioned works, our overall framework has a two-stage design: (i)
formal metatheory of global and local types (eschewing a process-level
language) for validating protocol specifications, and (ii) a practical
methodology for implementing MST-based sessions.
Also like the mentioned works and others based on
Scribble~\cite{DBLP:conf/fase/HuY16}, our practical methodology provides a
toolchain for formally-grounded protocol validation and projection (based on
the theory), and correct-by-construction code generation and/or runtime
mechanisms for processes.
Some works offload aspects of the protocol validation to supplementary methods
outside the core syntactic type system, including model
checkers~\cite{DBLP:journals/pacmpl/ScalasY19} and SMT
solvers~\cite{DBLP:conf/cc/NeykovaHYA18,DBLP:journals/pacmpl/CastroHJNY19,DBLP:journals/pacmpl/00020HNY20}.

Regarding the code generation and/or runtime mechanisms aspect,
\citet{10.1007/978-3-642-40787-1_8} developed a tool for generating
protocol-specific runtime monitors for sessions in Python;
\citet{DBLP:journals/pacmpl/CastroHJNY19} employed API generation for Go backed
up by Z3 for solving indexing constraints in role-parametric types;
and \citet{DBLP:journals/pacmpl/00020HNY20} developed types generation for static
refinement typing in F$^\star$.
In this paper, we develop code generation for protocol- and role-specific
\CODE{gen\_statem} callbacks and runtime mechanisms for automated stale message
purging.

As discussed, a key motivation for this paper is to develop a core construct
for MC to capture and unify the essence of practical constructs such as
exceptions, timeouts and failure handling.
Ultimately, a goal for the research community is to develop a complete system
that can integrate the full breadth of features needed by many
practical applications in addition to MC patterns, such as parameterisation of
protocols~\cite{DBLP:journals/corr/abs-1208-6483,DBLP:journals/pacmpl/CastroHJNY19}
and dynamic
topologies~\cite{DBLP:conf/fase/HuY17,DBLP:conf/ecoop/Castro-PerezY23}.

Other MST works that focus on a formal theory of global-local types without
a process-level language include works on richer global types and completeness
of projection~\cite{DBLP:conf/concur/MajumdarMSZ21,DBLP:conf/cav/LiSWZ23}, and
the semantics of multiparty sessions~\cite{DBLP:journals/corr/abs-1203-0780}.

\section{Handling Multiple Sessions}
\label{asec:multiple}
The key properties in our metatheory are the progress of valid global types,
and the operational correspondence (fidelity) between a valid global type and
its distributed local type projections; the latter entails a preservation
property for projection, and transfers global progress to the local
level.
These results pertain formally to a single multiparty session.
As such, our toolchain is primarily designed at present to generate a set of
Erlang modules from one source protocol and ensure the aforementioned
(single-session) properties.
Nevertheless, a programmer may (e.g.) use the modules of multiple, separate
protocols to implement a multi-session program.
We can informally outline the pragmatic conditions for (multi-session) progress
in our practical framework.
There are two main facets.

\paragraph{Local computations and inter-session dependencies.}
The first facet concerns the preconditions (i.e., assumptions) of our practical
framework regarding \emph{local} computations and events.
\Cref{sec:overview} described the usage contract of our toolchain regarding the
generated Erlang modules: the user must not modify the generated protocol- and
role-specific module called the Role Module (RM), and can only modify the
generated template callback module called the Callback Module (CM) according to
the pre-generated structures of the EFSM and callback functions.
Two key subconditions related to the latter are that the user must ensure:

\begin{itemize}[leftmargin=*]
\item
Any internal or local computation event required to fire a pending callback 
will eventually occur.
\item
Beyond its specific session I/O actions, every fired callback performs only
non-blocking actions (e.g., local computations), and eventually terminates
successfully to cede control back to the Erlang runtime (in order to fire the
next callback).
\end{itemize}

\noindent
Note, these conditions are independent of whether a program has a single
session or multiple sessions.
In short, it is the user's responsibility to ensure the correctness of local
computations, i.e., all Erlang code beyond session I/O and callback actions,
including arithmetic/logical expressions, general data processing and library
calls.
To implement a multi-session program, the user generates the (separate) RM and
CM modules for the relevant roles of each protocol, and implements the
necessary callbacks of each CM.
Inter-session dependencies that are \emph{internal} to a local program can be
expressed in various ways according to the dependency.
To illustrate, take the two protocols in \cref{fig:mini-protocols}.
Consider implementing one \CODE{Alice} program to participate in two sessions,
one for each protocol.
\Cref{fig:erlang-multi} gives two rudimentary examples of cross-session
dependencies (we mention alternative and safer approaches below under ``Future
work'').

\begin{itemize}[leftmargin=*]
\item
(Left) The start of one session is chained to the completion of the other.
Assume an \CODE{Alice} program that first starts the runtime handling for only
the \CODE{Heartbeat} session.
Eventually, the shown \CODE{s7} callback is fired when \CODE{Alice} is ready to
fire \CODE{pong} (some time after the preceding handling of \CODE{ping} has
finished), leading to the (local, non-blocking) \CODE{start\_link} action by
\CODE{Alice} to start handling the \CODE{Request} session.
\CODE{Bob} may asynchronously send \CODE{request} before or after that point;
even if before, \CODE{Alice} consumes the \CODE{request} message only after
\CODE{ping} and \CODE{s7}.

\item
(Right) The handling of an event is postponed until some local condition is
fulfilled.
Assume an \CODE{Alice} program that starts the runtime handling for both
sessions.
Eventually, the shown \CODE{s5} callback is fired when \CODE{Alice} receives
\CODE{request}.
If the local condition \CODE{LocalCond} is not yet fulfilled, \CODE{Alice}
postpones the handling of the \CODE{request} message by deferring the active
handler (in this example for 50ms).
%
\CODE{LocalCond} may be set by some arbitrary local computation -- or
a local event triggered by (e.g.) the \CODE{ping} handler in the other
session -- either way, our framework does rely on the programmer to ensure the
correctness of internal computations/events, as stated earlier.
\end{itemize}

\begin{figure}[t]
\begin{tabular}{ll}
\begin{minipage}{0.5\textwidth}
{\begin{lstlisting}[language=Scribble, numbers=left, basicstyle=\ttfamily\scriptsize]
global protocol Heartbeat(role Alice, role Carol){
    ping() from Carol to Alice;
    pong() from Alice to Carol;
}
\end{lstlisting}}
\end{minipage}
&
\begin{minipage}{0.46\textwidth}
{\begin{lstlisting}[language=Scribble, numbers=left, basicstyle=\ttfamily\scriptsize]
global protocol Request(role Alice, role Bob) {
    request()  from Bob to Alice;
    response() from Alice to Bob;
}
\end{lstlisting}}
\end{minipage}
\\
\end{tabular}
\caption{Two simple protocols for an example multi-session program.}
\label{fig:mini-protocols}
\end{figure}

\begin{figure}[t]
\begin{tabular}{ll}
\begin{minipage}{0.40\textwidth}
{
\begin{lstlisting}[numbers=left,language=Erlang, numbersep=3pt, basicstyle=\ttfamily\scriptsize]
s7(internal, {pong}, Data) ->
  gen_alice:send_s7_pong(CarolPid, Data),
  _ = gen_alice2:start_link(alice2, []);
  {stop, normal, Data}.
$~$
$~$
$~$
\end{lstlisting}
}
\end{minipage}
&
\begin{minipage}{0.6\textwidth}
{
\begin{lstlisting}[numbers=left,language=Erlang, numbersep=3pt, basicstyle=\ttfamily\scriptsize]
s5(cast, {BobPid, {request}}, Data) ->
  case LocalCond of
    true ->
      {next_state, s7, Data, [{next_event, internal, {response}}]};
    _ ->
       defer({BobPid, {request}}, 50), {keep_state, Data}
  end.
\end{lstlisting}}
\end{minipage}
\\
\end{tabular}
\caption{Basic inter-session dependencies internal to a local Erlang program: trigger vs.\ defer}
\label{fig:erlang-multi}
\end{figure}

\paragraph{Event-driven sessions and progress.}
The second main facet is that our Erlang sessions are \emph{event-driven} (ED),
rather than the typical `multithreaded' model of $\pi$-calculi.
The key points (in addition to those stated earlier) are that:

\begin{itemize}[leftmargin=*]
\item
The Erlang runtime fires a callback and activates the relevant session only when the
expected event has occurred and is ready for consumption.  Otherwise sessions
are passively suspended, meaning that no session ever actively blocks another
session's callback from firing.

\item
Interleaving of multiple sessions is enacted semantically by the Erlang
runtime interleaving the firing of callbacks on different sessions, one by one,
as their events occur.
That is as opposed to the syntactic interleaving of the typical I/O prefixes
within a $\pi$-calculus process, where a prefix on one session can actively
block a prefix on another session from being executed (which could, e.g., cause
a deadlock cycle).
\end{itemize}

\noindent
The above concepts are found in many practical systems (e.g., Erlang, Akka,
Java NIO, etc), and have been previously formalised in the setting of
(linearly-typed) ED session
$\pi$-calculi~\cite{DBLP:conf/ecoop/HuKPYH10,DBLP:journals/mscs/KouzapasYHH16,DBLP:journals/pacmpl/VieringHEZ21}.
A global progress property related to the above points was formally established in
the specific setting of \citet{DBLP:journals/pacmpl/VieringHEZ21}.

Altogether, (multi-session) progress in our Erlang programs can be understood by
considering the combination of these points: (1) our formal result that every
distributed MC protocol individually satisfies progress
(\Cref{cor:gprog,thm:oc}); (2) the characteristics of ED sessions;
(3) our assumptions on local computations in the practical programs; and (4)
inter-session dependencies within a local program expressed by local
computations/events.

Some multi-session programs may involve only (1)-(3), e.g., as in some
client-server applications that spawn a fully independent session per client;
in such cases, the ED framework can effectively transfer global type progress to
the multi-session process level directly.
Some cases of (4) can also be limited to ensure progress, such as when sessions
are strictly chained to spawn and run consecutively, or spawned in a restricted
tree topology of parent-child subsessions (as found in systems based on linear
logic~\cite{DBLP:journals/jfp/Wadler14,DBLP:journals/mscs/CairesPT16}).
As mentioned, however, our practical framework does rely on the user to ensure
the correctness of internal computations/events, including more advanced
inter-session dependencies, in full generality.

\paragraph{Future work.}
In future work, we will consider extending our toolchain to support the
specification of inter-\emph{protocol} dependencies via composition constraints
in the style of \citet{DBLP:journals/programming/BocchiOV23}.
This will allow ad hoc inter-session dependencies at the process level to
instead be expressed at the \emph{protocol level}.
The toolchain can compose such protocols into a single EFSM and generate APIs
that safely embed the dependencies correctly by construction, rather than
relying on the programmer to express them via local computations as illustrated
above.

We also plan to investigate formalising the event-driven (ED) concurrency of
Erlang's \CODE{gen\_statem} and its correctness properties at the process
level; however, it should be noted that Erlang is dynamically typed by default.
By contrast, \citet{DBLP:journals/jlap/CicconeDP24} recently proved a formal
multi-session progress property for a `multithreaded' session $\pi$-calculus
system based on \emph{fair termination} (see \Cref{sec:related}).
There may be connections between their notion of fair termination and our
practical condition that callbacks and local computations should eventually
terminate: their work may give directions for formalising ED callback
termination.
\citet{DBLP:journals/pacmpl/VieringHEZ21} established a progress property for
an ED $\pi$-calculus involving multiple sessions (their principles for ED
progress are similar to those we described in the earlier paragraph).
Earlier
works~\cite{DBLP:conf/coordination/PadovaniVV14,DBLP:journals/mscs/CoppoDYP16}
developed multi-session progress properties using additional
analysis mechanisms on top of session typing.

 \newcommand{\comSetGenD}[2]{\Downarrow^c\hspace{-1mm}(#1, {\color{gray} #2})}

\section{Committing and non-committing sets - full definitions}
\label{app:committing}

The auxiliary function for committing sets is defined as follows:
\[
\begin{array}{lll}
\comSet{\GactShort{p}{q}{\St{}}}{C}{c}  =   \begin{cases}
    \labG{\St{}} ~\cup  \displaystyle \bigcup_{\GT \in \tyG{\St{}}} \hspace{-0.4cm} \comSet{\GT}{C\cup\{\role q\}}{c}  & \quad  \p \in C  { ~\land~\q\not \in C}
    \\[0.6cm]
     \displaystyle \bigcup_{\GT\in \tyG{\St{}}}\hspace{-0.4cm}\comSet{\GT}{C}{c}  & \quad \textit{otherwise}
    \end{cases}
\\[1cm]
 \comSet{\GtoDef{\GT_1}{c'}{\lset}{\rset}{\GT_2}}{C}{c}  ~=    ~ \begin{cases}
\comSet{\GT_1}{C}{c} ~~\cup  \comSet{\GT_2}{C}{c} & c\neq c'
\hspace{1.8cm}  \comSet{\mu \mathtt t .\GT}{C}{c}=\comSet{\GT}{C}{c}
\\
\emptyset & \textit{otherwise}\hspace{1.3cm}    \comSet{\tend}{C}{c}= \comSet{\mathtt t }{C}{c} =\emptyset
 \end{cases}
\end{array}
\]

The dual auxiliary function for non-committing sets is defined as follows:

\[
\begin{array}{lll}
\newcomSet{\GactShort{p}{q}{\St{}}}{C}{c}  =   \begin{cases}
     \displaystyle \bigcup_{\GT \in \tyG{\St{}}} \hspace{-0.4cm} \newcomSet{\GT}{C\cup\{\role q\}}{c}  & \quad  \p \in C  { ~\land~\q\not \in C}
    \\[0.6cm]
     \labG{\St{}} ~\cup \displaystyle \bigcup_{\GT\in \tyG{\St{}}}\hspace{-0.4cm}\newcomSet{\GT}{C}{c}  & \quad \textit{otherwise}
    \end{cases}
\\[1cm]
 \newcomSet{\GtoDef{\GT_1}{c'}{\lset}{\rset}{\GT_2}}{C}{c}  ~=    ~ \begin{cases}
     \newcomSet{\GT_1}{C}{c} ~~\cup  \newcomSet{\GT_2}{C}{c}  & c\neq c' \hspace{1.7cm} \newcomSet{\mu \mathtt t .\GT}{C}{c}=\newcomSet{\GT}{C}{c}\\
     \emptyset & \textit{otherwise} \hspace{1.3cm}   \newcomSet{\tend}{C}{c}= \newcomSet{\mathtt t }{C}{c} = \emptyset
 \end{cases}
\end{array}
\]

\subsection{Extension using observers committing sets}
\label{asec:committing2}

Let $\GT$ be an initial type and $\GactShort{q}{p}{\, \St{1}}\mix^{\color{magenta}c}\,\GactShort{p}{q}{\, \St{2}}$ in $\unfold{\GT}$ a MC definition in $\unfold{\GT}$.
Fix a special set of labels $\GC$ of committing actions for the LHS of MC $c$ by its observer $\p$. We define the $\GC$-committing set of $c$ as
\[\begin{array}{lll}
\labG{\St{1}\, \cap\, \GC} ~\cup~ \labG{\St{2}}
\\ \quad \cup ~
\bigcup_{G\in \tyG{\St{1}\cap \GC}}\comSetGenD{G}{C\cup \{\role p\}}
\\ \quad    \cup ~
\bigcup_{G\in \tyG{\St{1}\setminus \GC}} \comSetGenD{G}{C}
\\ \quad \cup~
\bigcup_{G\in \tyG{\St{2}}}
 \comSetGenD{G}{C\cup\{\role p, \role q\}}
 \end{array}\]

where $\comSetGenD{\GT}{C}$ is defined as:
\[
\begin{array}{lll}
%
\small
 %
 %
  \comSetGenD{\GactShort{p}{q}{\{a_i.\GT_i\}_{i\in I}}}{C}  = \\  \qquad\qquad\begin{cases}
    \bigcup_{i\in I} \{a_i\} ~\cup \comSetGenD{G_i}{C\cup\{\role q\}}  & \p \in C    \land \q\not\in C
    \\
(\bigcup_{i\in I} \{a_i\}\cap \GC ) \cup
\{ \comSetGenD{G_i}{C\cup\{\role q\}}  ~|~ a_i \in \GC\}
\cup
\{ \comSetGenD{G_i}{C} ~|~ a_i \not\in \GC \} & otw
    \end{cases}   \\[0.6cm]
     %
    %
     %
     %
  \comSetGenD{\mu \mathtt t .\GT}{C} = \comSetGenD{\GT}{C} \\[0.4cm]
 \comSetGenD{\mathtt t}{C} = \comSetGenD{\tend}{C}=\emptyset
    %
\end{array}
\]

The definition of eventual dependency based on $\GC$ is as follows: $\role q$ \emph{eventually depends on}
$\p$ in $\GT$ (wrt $\GC$), written $\wdep{\role p}{\role q}{\GT}$, if  $\GT\trans{\vec{\ell}}\GT'$ implies $\exists a\in\GC. ~\p\q!a \in \vec{\ell}~\lor ~ \GT' \myreaches{}\trans{\p\q!a}$.

 \section{Progress of Global Types: extended definitions and proofs}
\label{asec:globalprogress}

Progress is proved via a number of properties. The principal ones have been discussed in the main text. Here, we give full definitions of all properties needed in the proof of progress.
First, in Section~\ref{asec:wellnested} we discuss well-nestedness, an auxiliary property on the structure of MC, preserved by transition.
In Section~\ref{asec:awareness} we prove preservation of awareness and in Section~\ref{asec:balance} we discuss preservation of {balance}. We then elaborate on their role to yield the progress property via an invariant property called \emph{coherence} in Section~\ref{asec:coherence}.
We finally wrap up the progress property in Section~\ref{asec:progress}.

\begin{remark}[Instance annotations]To facilitate reasoning on the instances of MC on a global type we annotate MC instance with an instance identifier $n$.
This annotation will be used only to prove global properties. The properties non correspondence rely on the non-annotated semantics.
\end{remark}

\subsection{Global types with annotated MC}
The syntax of global types with annotated MC is defined by the grammar below:
\[\begin{array}{lll}
  \GT  ::=  \GactShort{p}{q}{\St{}}
     \mid  \GmsgShort{p}{q}{k}{\St{}}
     \mid  \mu \rv . \GT
     \mid  \rv
    \mid \tend
    \mid \GtoDef{\GactShort{q}{p}{\St{1}}}{\tdef{c}{p}}{\lset}{\rset}{\GactShort{p}{q}{\St{2}}}
     \mid  \GtoP{\GT_1}{ \inst:\role p}{\lset}{\rset}{\GT_2}
\end{array}
\]

We annotate MC definitions with a unique MC name $c$.
Active MC are annotated with $\inst = c,n$, which are pairs of MC names and counters $n\in\ndom$. Since a MC may be in the body of a recursive type, many instances of that MC may be generated upon recursive unfolding. We use $\inst$ to identify these instances. For readability, we omit annotations $\role p$, $c,n$, $\inst$, $\lset$, or $\rset$ when not needed.

The semantics of global types is defined as a Labelled Transition System over terms $\GT$ with labels
\[ \ell :=  \lsnd{p}{q}{a} \mid \lrcv{p}{q}{a} \mid \nu c,n\]

The semantics uses a mapping $\Theta$ from global types and MC names to $\ndom$, that gives the greatest counter among all active MC $c$ in $\GT$
\[\Theta(\GT, c) = \mathtt{max}(\{0\} \cup \{n \, |\, \text{$\GT$ has active MC with identifier ${c},{n}$}\}) \]
$\Theta$ is used to ensure that active MC are uniquely identified by  $c,n$.

\begin{figure}
\small
\[
\begin{array}{cc}
\begin{array}{cc}
	%
    %
	 \gspairf{\Gact} \trans{\lsnd{p}{q}{a_k}} \Gmsg~~~  (k\in I)~~
    \mathtt{[Snd]}\\[0.3cm]
	\gspairf   \Gmsg \trans{\lrcv{p}{q}{a_k}} \GT_k  ~~~\mathtt{[Rcv]}~~~ \qquad
	%
	   \inferrule{\forall i \in I \qquad \gspairf\GT_i \trans{\tlab}\GT'_i \qquad \role p, \role q \not \in \sbj{\tlab} 
    }{\gspairf\Gact \trans{\tlab} \GactP{p}{q}{a_i}{\GT'_i}{i\in I}}
    ~~\mathtt{[Cont1]}\\[0.7cm]
     \inferrule{\gspairf{\GT[\mu \rv. \GT / \rv]}\trans{\tlab}{\GT'} }{\gspairf\mu \rv . \GT \trans{\tlab} {\GT' }}~~\mathtt{[Rec]}
     \qquad~~
    \inferrule{ \gspairf\GT_k \trans{\tlab}{\GT'_k} \qquad \role q \not \in \sbj{\tlab} \qquad \forall i\in I\setminus k.\, \GT_i=\GT_i'
    }{\gspairf\Gmsg \trans{\tlab} \GmsgP{p}{q}{k}{a_i}{\GT'_i}{i\in I}}
    ~~
      \mathtt{[Cont2]}\\[0.7cm]
    \hline \\
\end{array}\\
%
%
%
%
%
\begin{array}{rr}
\inferrule{ \Theta(\GtoDef{  \GT_1}{\tdef{c}{p}}{\lset}{\rset}{  \GT_2},c)<  n
}
	{ \gspairf{\GtoDef{  \GT_1}{\tdef{c}{p}}{\lset}{\rset}{  \GT_2}}
\trans{\nu c, n} {\GtoP{  \GT_1}{\tact{c}{n}{p}}{\emptyset}{\emptyset}{  \GT_2}}} ~~\mathtt{[Inst]}
%
 %
 &
 \inferrule{\gspairf \GT_l \trans{\nu c, n} {\GT_l'} \quad \RolesG{\underline{\GT}} \neq \rset \quad \Theta(\GT,c)< n}
	{ \gspairf{\GT = \GtoP{\GT_l}{}{\lset}{\rset}{\GT_r} }  \trans{\nu c, n} {\GtoP{\GT'_l}{}{\lset}{\rset}{\GT_r}}}
 ~~ \mathtt{[Ctx1]}\\[0.9cm]
 \inferrule{\gspairf\GT_r  \trans{\nu c, n} {\GT_r'}\quad \lset = \emptyset \quad \Theta(\GT,c)< n}
	{ \gspairf{\GT = \GtoP{\GT_l}{~}{\lset}{\rset}{\GT_r} }  \trans{\nu c, n} {\GtoP{\GT_l}{~}{\lset}{\rset}{\GT_r'}}}
 ~~ \mathtt{[Ctx2]}
%
&
	\inferrule{\gspairf \GT_l \trans{\lsnd{p}{q}{a}} {\GT_l'}
 \quad
 \color{olive}\role p \not \in \rset }
	{ \gspairf{\GtoP{\GT_l}{}{\lset}{\rset}{\GT_r}}\trans{\lsnd{p}{q}{a}} {\GtoP{\GT'_l}{}{\lset}{\rset}{\GT_r}}
    }
    ~~ \mathtt{[LSnd]}\\[0.9cm]
 \inferrule{\gspairf\GT_l \trans{\lrcv{p}{q}{a}} {\GT_l'} \qquad
 \color{olive}\role q \not \in \rset \qquad a\in \comSetG{
 \underline{\GT}}
 }
	{ \gspairf{\GtoP{\GT_l}{}{\lset}{\rset}{\GT_r}} \trans{\lrcv{p}{q}{a}} {\GtoP{\GT'_l}{}{\lset\cup\{\role q\}}{\rset}{\GT_r}}}~~ \mathtt{[LRcv1]}
 &
	\inferrule{\gspairf\GT_l \trans{\lrcv{p}{q}{a}} {\GT_l'}\qquad
 \color{olive}\role q \not \in \rset \qquad  a\not \in \comSetG{
\underline{\GT}}}
	{ \gspairf{\GtoP{\GT_l}{}{\lset}{\rset}{\GT_r}} \trans{\lrcv{p}{q}{a}} {\GtoP{\GT'_l}{}{\lset}{\rset}{\GT_r}}}~~\mathtt{[LRcv2]}\\[0.9cm]
	\inferrule{\gspairf\GT_r \trans{\lsnd{p}{q}{a}} {\GT_r'} \qquad
 \color{olive}\role p \not \in \lset }
	{ \gspairf\GtoP{\GT_l}{}{\lset}{\rset}{\GT_r} \trans{\lsnd{p}{q}{a}} {\GtoP{\GT_l}{}{\lset}{\rset\cup\{\role p\}}{\GT'_r}}}~~ \mathtt{[RSnd]}
&
\inferrule{\gspairf\GT_r \trans{\lrcv{p}{q}{a}} {\GT_r'}\quad \color{olive}\role q \not \in \lset }
  { \gspairf\GtoP{\GT_l}{}{\lset}{\rset}{\GT_r} \trans{\lrcv{p}{q}{a}} {\GtoP{\GT_l}{}{\lset}{\rset\cup\{\role q\}}{\GT'_r}}}~~  \mathtt{[RRcv]}
\end{array}
\end{array}
\]
\caption{Global semantics: standard rules (top) and new rules for MC (bottom)}
\label{fig:gsemannotated}
\end{figure}

The rules for annotated semantics are given in~\Cref{fig:gsemannotated}. We comment on the differences with the rules given in the paper (in~\Cref{fig:gsem}). \glab{Inst} instantiates a MC uning $\Theta$. \glab{Ctx1} and \glab{Ctx2} handle nested instantiations, using  $\Theta$ to ensure that $n$ is strictly greater than any other counter for $c$, and label $\nu c,n$ to propagate this requirement across the derivation tree (to ensure uniqueness). The other rules are unchanged with respect to~\Cref{fig:gsem}.

\subsection{Unique instances}
\label{asec:ui}

\begin{restatable}[Unique instances]{proposition}{UIinvariant}
\label{UIinvariant}
If $\GT_1 \blacktriangleright^{c,n}{\GT_2}$ and
${\GT_1'}\blacktriangleright^{c,m}{\GT_2'}$ are distinguished subterms of a reachable $\GT$ then $n\neq m$.
\end{restatable}
\begin{proof}
  Fix $\Theta$ appropriate for $\GT$ such that
  $\gspair{\GT}{\Theta}\trans{}{\GT'}$.
  The proof is by induction on the derivation, proceeding by case analysis on the last rule used.

  \paragraph{Base case} There is only one base case by rule $\mathtt{[Inst]}$. $\mathtt{[Inst]}$ that increments $\Theta(c)$ by one unit. Since $\Theta$ is appropriate for $\GT$ (hypothesis) then $\Theta(c)+1$ will be strictly greater than any instance number occurring in $\GT$ for $c$.

  \paragraph{Inductive cases}
  In case of transitions by $\mathtt{[Cont1]}$, $\mathtt{[Cont2]}$,  $\mathtt{[Rec]}$, and $\mathtt{[Ctx1]}$ (and symmetric) the thesis holds directly by induction.
 Remarkably, $\mathtt{[Ctx2]}$ ensures that an instantiation with $n$ happens in only one side of a timeout for any one transition. Observe that rule $\mathtt{[RTAct]}$, allowing the two sides to move in the same transition, cannot be applied to instantiate a timeout because of premise $\ell\not\in\{\nu c,n\,|\, c\in \mathcal{C} \land n\in\mathbb{N}^+ \}$.
  \end{proof}

\subsection{Monotonicity}
\label{asec:monotonicity}

\begin{restatable}[Monotonicity]{proposition}{lemmonoton}
\label{lem:monoton}
For any global type of the form $\Ctname[\GtoP{\GT_1}{\inst}{\lset}{\rset}{\GT_2}]$:
\begin{align}\Ctname[\GtoP{\GT_1}{\inst}{\lset}{\rset}{\GT_2}]\trans{}\gspairl{\Ctname'[\GtoP{\GT_1'}{\inst}{\lset'}{\rset'}{\GT_2'}]}{\Theta'}\quad \Rightarrow
\lset'\supseteq\lset ~~\land ~~ \rset'\supseteq\rset\nonumber\end{align}
\end{restatable}
Monotonicity extends trivially to MC definitions, which have no commitments and are instantiated into active MC with empty commitment sets.

\begin{proof}

We show that, given a MC with identifier $\inst$, a transition can only extend its commitment sets $\lset$ and $\rset$. Let $\GT = \Ctname[\GtoDef{\GT_1}{\tdef{c}{p}}{\lset}{\rset}{\GT_2}]$ or
$\GT = \Ctname[\GtoP{\GT_1}{\tact{c}{n}{p}}{\lset}{\rset}{\GT_2}]$. We proceed by induction on $\Ctname$. In the base case, $\Ctname = \hole$, we proceed by analysis on the last rule used for the transition, which can be: $\mathtt{[LSnd]}$, $\mathtt{[LRcv1]}$, $\mathtt{[LRcv2]}$, $\mathtt{[RSnd]}$, or $\mathtt{[RRcv]}$, or $\mathtt{[RTAct]}$. None of these rules is removing elements from $\rset$ and $\lset$ hence the thesis.

In the inductive case, if $\Ctname$ is an interaction it can move by either (a) $\mathtt{[Snd]}$ or $\mathtt{[Rcv]}$ leaving the continuations unchanged and yielding the thesis, or (b) $\mathtt{[Cont1]}$ or $\mathtt{[Cont2]}$ yielding the thesis by induction.

If $\Ctname = \GtoP{\hole }{\tact{c'}{n'}{p'}}{\lset''}{\rset''}{\GT_2''}$ then one of the following three cases can happen: (a) $\GT_2''$ moves by $\mathtt{[RSnd]}$ or $\mathtt{[RRcv]}$ not affecting the right-hand side of timeout $\tact{c'}{n'}{p'}$ hence yielding the thesis; (b) the left-hand side moves and the thesis follows by induction; (c) both sides move by $\mathtt{[RTAct]}$ with the thesis following by induction.
The case for $\Ctname = \GtoP{\GT_1'' }{\tact{c'}{n'}{p'}}{\lset''}{\rset''}{\hole}$ is symmetric.

The two cases for not initialized timeout context and recursion are vacuous since all timeouts are initially timeout definitions and initialization happens in nesting order.
The case for recursion is vacuous since a recursion only has timeout definitions.
\end{proof}

\subsection{Well-nestedness}
\label{asec:wellnested}

We define a structural invariant on MCs with respect to $\lset$ and $\rset$.

\begin{definition}[Well-nested $\GT$]\label{def:wn}
$\GT$ is well-nested if $\GT = \Ct{}{ \GtoP{\GT_1 }{\role q}{\lset}{\rset}{\GT_2}}$ implies
\begin{enumerate}
\item $\lset \not = \emptyset ~\lor~ \GT_1 = \GactShort{p}{q}{\St{}}~ \lor ~\GT_1 = \GmsgShort{p}{q}{k}{\St{}}$ for some $\role p$, $\St{}$, and
\item  $\rset \not = \emptyset ~\lor ~\GT_2 = \GactShort{p}{q}{\St{}}$ for some $\role p$, $\St{}$.
\end{enumerate}
\end{definition}

\begin{lemma}[Well-nested preservation]\label{lem:wn}
If $\GT$ is well-nested and $\gspair{\GT}{\Theta}  \trans{\ell}{\GT'}$ then $\GT'$ is well-nested.
\end{lemma}
\begin{proof}
Since $\GT$ is well-nested, fix $\Ctname$ to be any context such that $\GT = \Ct{}{ \GtoP{\GT_1 }{\tact{c}{n}{q}}{\lset}{\rset}{\GT_2}}$.
We first consider case (1) of Definition~\ref{def:wn}.
\begin{itemize}
\item Base case $\Ctname=\hole$.
\begin{itemize}
\item if $\GT$ is of the form $\GactShort{p}{q}{\St{}}$ or $\GmsgShort{p}{q}{k}{\St{}}$ we have two cases: (1) the prefix moves, which leaves $\St{}$ unchanged hence done, (2) $\CT$ moves and the thesis is by inner induction.
\item if  $\GT =  \GtoP{ \GT_1}{\tact{c}{n}{q}}{\lset}{\rset}{\GT_2}$ and $\lset\not=\emptyset$ then the thesis follows observing that no rules remove roles from $\lset$ and hence this set will still be non-empty after transition. If $\lset = \emptyset$ then by well-nestedness of $\GT$ we have one of the following cases:
\begin{itemize}
\item $\GT =  \GtoP{ \GactShort{p}{q}{\St{}}}{\tact{c}{n}{q}}{\lset}{\rset}{\GT_2}$. If $\GT_2$ or $\St{}$ moves, the thesis is by inner induction. If the prefix moves then it moves by $[\mathtt{LSnd}]$ to $ \GtoP{ \GmsgShort{s}{r}{k}{\St{}}}{\tact{c}{n}{p}}{\lset}{\rset}{\GT_2}$ for some $k$, which is still well-nested.
\item $\GT =  \GtoP{ \GmsgShort{s}{r}{k}{\St{}}}{q}{\lset}{\rset}{\GT_2}$. Again, if $\GT_2$ or $\CT$ moves, the thesis is by inner induction. If the prefix moves, it is for $[\mathtt{LRcv1}]$ since $\lset=\emptyset$ to
$\GtoP{ \St{}}{\tact{c}{n}{q}}{\lset\cup\{\role{q}\}}{\rset}{\GT_2}$ which is well-nested since $\lset\cup\{\role{q}\}\not=\emptyset$.
\end{itemize}
\item If $\GT =  \GtoP{ \GmsgShort{p}{q}{k}{\St{}}}{q}{\lset}{\rset}{\GT_2}$ the case is similar to the above.
\end{itemize}
\item Inductive cases $\Ctname = \GactShort{p}{q}{\St{}}\cup\{a.\hole\}$ or $\GmsgShort{p}{q}{k}{\St{}}\cup\{a_k.\hole\}$,
either by outer induction of $\St{}$ (if $\St{}$ changes) or inner induction on $\GT$ if the hole changes.
\item Inductive case $\Ctname = \GtoP{\hole}{\tact{c}{n}{q}}{\lset}{\rset}{ \GT}$. If $\lset\not=\emptyset$ then the thesis follows observing that no rules remove roles from $\lset$ and hence this set will still be non-empty after transition. If $\lset = \emptyset$ then $\GT =  \GtoP{\GactShort{s}{r}{\St{}}}{\tact{c}{n}{q}}{\lset}{\rset}{ \GT'}$ or  $\GT =  \GtoP{\GmsgShort{s}{r}{k}{\St{}}}{\tact{c}{n}{q}}{\lset}{\rset}{ \GT'}$ by well-nestedness of $\GT$. The case is similar to the second base case for $\Ctname=\hole$.
\item Inductive case $\Ctname = \GtoP{ \GT'}{\tact{c}{n}{q}}{\lset}{\rset}{\hole }$. Similar to the above.
\item Inductive cases $\GT = \GtoDef{\hole}{\tdef{c}{p}}{\lset}{\rset}{\GT}$ or $\GT = \GtoDef{\GT}{\tdef{c}{p}}{\lset}{\rset}{\hole }$ : the move can only be by $[\mathtt{Inst}]$ which does not change the structure of the process, hence done.
\item Inductive case $\GT  =  \mu \rv . \hole $ : the move can only be by $[\mathtt{Rec}]$, and by its premise the type in the hole preserves well-nestedness. The thesis is by induction.
\end{itemize}
The case for Definition~\ref{def:wn}(2) is similar to the one for Definition~\ref{def:wn}(1): either $\rset=\emptyset$ with some participants on the right-hand side of the MC, or  the next committing role is included in $\rset$ by $\mathtt{[Rsnd]}$.
\end{proof}

\subsection{Preservation of awareness}
\label{asec:nested}\label{asec:awareness}

Awareness is preserved by transition. Before proving this property we give an auxiliary proposition. \Cref{pro:dep} states that if all roles in $\GT$ depend on $\role p$ then the first action of $\GT$ is an action by $\role p$.

\begin{proposition}\label{pro:dep}
If $\dep{\role p}{}{\GT}$ and $\gspairl{\GT}{}\trans{\ell}$ then $\sbj{\ell}=\role p$.
\end{proposition}

\begin{restatable}[Awareness Preservation]{proposition}{AP}
\label{awarepreservation}
\(\gspairl{\GT}{} \text{ is aware } ~\land ~\gspairl{\GT}{}\trans{}\gspairl{\GT'}{\Theta'}~~\Rightarrow~~\gspairl{\GT'}{} \text{ is aware. }\)
\end{restatable}
\begin{proof} The proof is by induction on the derivation, proceeding by case analysis on the last rule used.

\paragraph{Base cases}
If the transition is by $\mathtt{[Snd]}$, any MC in any continuation $\GT_i$, which is aware in $\GT$ by hypothesis,  remains unchanged hence aware. The case for $\mathtt{[Rcv]}$ is similar. The case for $\mathtt{[Inst]}$ only changes the outer MC definition into an active MC with empty L-set and R-set. If L-set and R-set are empty awareness for MC definitions is equivalent to awareness for active MC. The thesis is therefore straightforward by awareness of $\GT$.

\paragraph{Inductive cases}
Cases for $[\mathtt{Cont1}]$ and $[\mathtt{Cont2}]$ are straightforward by induction.

For $[\mathtt{Rec}]$, observe that a recursive type has a transition $\ell$ if and only if its one-time unfolding has a transition $\ell$, and they both reach the same state. Therefore, if $\GT = \mu \mathtt t . \GT''$ is aware then also $\GT''\, [\mu \mathtt t.\GT'' / \mathtt t]$ is aware and after a transition they reach the same state $\GT'$. By inductive hypothesis and awareness of $\GT''\, [\mu \mathtt t.\GT'' / \mathtt t]$ it follows that $\GT'$ is aware hence done.

For $[\mathtt{Ctx1}]$  then $\GT = {\GtoP{\GT_{\mathtt l}}{~}{\lset}{\rset}{\GT_{\mathtt r}}}$ and by induction if $\GT_{\mathtt l}$ makes a move to $\GT_{\mathtt l}'$ then $\GT_l'$ is aware. So $\GT_{\mathtt l}'$ and $\GT_{\mathtt r}$ are aware (the second directly by hypothesis) meaning that all MC in them are aware. It remains to show that the outermost MC $\GtoP{\GT_{\mathtt l}}{~}{\lset}{\rset}{\GT_{\mathtt r}}$ is still aware after the transition.

This follows by the fact that the move with label $\nu c, n$ leaves LR-sets, and participants and dependencies unchanged from $\GT_{\mathtt l}$ to $\GT_{\mathtt l}'$. The symmetric case is similar.

For $[\mathtt{LSnd}]$ assume $\GT =\GtoP{\GT_{{\mathtt l}}}{~}{\lset}{\rset}{\GT_r}$
 and $\GT'=\GtoP{\GT_{\mathtt l}'}{~}{\lset}{\rset}{\GT_r}$ and the observer of $\GT$ is $\role p$.
Single decision and clear termination of $\GT_{\mathtt l}$ follow by induction and awareness of $\GT_r$ follows by hypothesis. Single decision of $\GT$ follows by single decision of $\GT_{\mathtt l}$ and by the fact that $\GT_r$ is unchanged by the transition. We only need to show clear termination of the outermost MC in $\GT'$.
If $\lset \not = \emptyset$ then the outermost MC in $\GT'$ is clear termination since the premise in the implication of case (2) of Definition~\ref{def:aware} is negative.
Assume now $\lset = \emptyset$. By clear-termination of $\GT$ at least one of the following holds: (a) all roles diverge in $\GT_1$, (b) every state reached by $\GT_1$ and any role in $\GT_1$, it is possible to reach a state where that roles makes an action.
(a) is clearly preserved by reduction (an infinite execution remains infinite if we remove a finite prefix). (b) is also preserved since rule $[\mathtt{LSnd}]$ does not change $\lset$ and hence future committing action that makes $\GT_1$ clear-termination is also available in $\GT_1'$.


For $[\mathtt{LRcv1}]$ assume $\GT = \GtoP{\GT_l}{\tact{c}{n}{r}}{\lset}{\rset}{\GT_r}$ and the reached state has type  $\GT'=\GtoP{\GT'_l}{\tact{c}{n}{r}}{\lset\cup\{\role q\}}{\rset}{\GT_r}$.
To show clear-termination, observe that the reached state has L-set $\lset\cup\{\role q\}\not = \emptyset$ and hence clear-termination of $\GT'$ follows immediately (by negative premise).
Single decision is trivial since the right-hand side of the MC remains unchanged (hence still enjoys single decision).

Case $[\mathtt{LRcv2}]$ (single-decision and clear-termination) is similar to $[\mathtt{LSnd}]$ except we now know that the action does not have the observer as subject by premise of the rule $[\mathtt{LRcv2}]$.

Cases $[\mathtt{RSnd}]$ and $[\mathtt{RRcv}]$: clear termination follows by hypothesis since the left-hand side of the MC does not change, single decision follows by negative premise since the L-set in the reached state includes the subject of $\ell$ and hence is not empty.
\end{proof}

\subsection{Preservation of balance}
\label{asec:balance}
Balance is not, in general, preserved by transition. For example, the global type on the left of (\ref{eq:balance}) is balanced and it reduces to a state, on the right of (\ref{eq:balance}), that is not since $\role q$ is nor in the roles of the LHS nor in the left commit set.
\begin{equation}\label{eq:balance}
 \GtoP{\GactShort{q}{p}{a_k. \tend}}{\role{p}}{\emptyset}{\emptyset}{\GactShort{p}{q}{\St{}}} ~~ \trans{\role p\role q!a_k}~~
  \GtoP{\GmsgShort{q}{p}{k\,  a_k.\tend}}{\role p}{\emptyset}{\emptyset}  {\GactShort{p}{q}{\St{}}}
\end{equation}

However, in presence of awareness, balance is preserved by transition.

\begin{proposition}[Balance \& Preservation]\label{Finvariant}
\(\gspairl{\GT}{} \text{ is aware and balanced } ~\land ~\gspairl{\GT}{}\trans{}\gspairl{\GT'}{\Theta'}~~\Rightarrow~~\gspairl{\GT'}{} \text{ is balanced. }\)
\end{proposition}
The proof of Proposition~\ref{Finvariant} is mechanical by induction on the transition. One key point to observe is that clear-termination always ensures that any role $\role r$ not yet committed in a MC still appears in the LHS (either because it has to receive a committing message, or because it diverges in the LHS). In the first case each role will naturally occur until it is committed, in the second case it will occur in the unfolding. This is key to preserve balance.

%

\subsection{Coherence}
\label{asec:coherence}

\begin{definition}[Coherent $\GT$]\label{def:coherence}\label{asec:pp}

$\GT$ is \emph{coherent} if: \( \GT = \Ct{C}{
\GtoP{\GT_1}{}{\lset}{\rset}{\GT_2}}\,  \implies \, \lset=\emptyset\, \lor\, \rset=\emptyset\)
\end{definition}
Informally, coherence of $\GT$ requires that no role is committed to different sides of the same active MC in $\GT$. Coherence is preserved by transition and hence is invariant for all states reachable from an initial global types enjoying balance and awareness.

The proof of coherence is given after a few auxiliaries.

\begin{definition}[Ready roles]\label{def:ready}
We say that a role is \emph{ready} in $\GT$,  written $\role p\in \ready{\GT}$, if there exists $\ell$ such that $\gspairl{\GT}{\Theta}\trans{\ell}$ and $\sbj{\ell}=\role p$.
\end{definition}
Namely, a role $\role p$ is ready in $\GT$ if
$\role p$ can immediately make a send or receive action.

The following proposition can be proved mechanically by induction on the syntax of $\GT_1$, observing that the definition of committing set always adds the observer of a MC first, to the committing set of a mixed choice, before all other causally related actions.
\begin{proposition}\label{pro:committing}
Given a MC $\GT_1\mix\GT_2$ with observer $\p$ and (finite or infinite) execution $\GT_1\trans{\vec{l}}$, the first committing action in vector $\ell$ (if any is committing) has subject $\p$.
\end{proposition}

\begin{restatable}[LR-initiation]{proposition}{APP}
\label{awareproperty}\label{auxx}
Consider an \emph{aware} active MC of the form $\Ctname [\GtoP{\GT_1}{\role p}{\lset}{\rset}{\GT_2}]$
\begin{enumerate}
    \item If $\role r\in \ready{\GT_2}\setminus \{\role p\}$ then $\role p \in \rset$.
    \item If $\role p\not\in \rset $ then $\rset=\emptyset$.
    \item If $\role p\not\in \lset $ then $\lset=\emptyset$.
\end{enumerate}
\end{restatable}
\begin{proof}(sketch)
To prove (1,2) observe that the execution of any \good\ state $\GT$ is in one of the following four meta-states:
\begin{enumerate}
\item (initial): ${\GtoDef{\GT_1}{\tdef{c}{p}}{\lset}{\rset}{\GT_2}}$ is a subterm of $ \GT$ and  $\GT$ has no active MC instances $c,n$;
\item (initialized) ${\GtoP{\GT_1}{\tact{c}{n}{p}}{\lset}{\emptyset}{\GT_2}}$ is a subterm of $ \GT$;
\item (R-committed) ${\GtoP{\GT_1}{\tact{c}{n}{p}}{\lset}{\rset \cup \{\role p\}}{\GT_2}}$.
\end{enumerate}
In state (1) $\GT$  can make:
\begin{itemize}
    \item transitions that are not with label $\nu c,n'$ with $n'\in\mathbb{N}$. This can be by $[\mathtt{Snd}]$, $[\mathtt{Rcv}]$, $[\mathtt{Cont1}]$, $[\mathtt{Cont2}]$, $[\mathtt{Rec}]$ that will preserve state (1),
    \item a transition with label $\nu c,n'$ with $n'<n$ and move to state (1),
    \item a transition with label $\nu c,n$ and move to state (2).
\end{itemize}

In state (2) $\GT$  can make
\begin{itemize}
    \item non-committing actions by $[\mathtt{LSnd}]$, $[\mathtt{LRcv2}]$, $[\mathtt{RTAct}]$ that preserve state (2),
    \item an action committing on the left by $[\mathtt{LRcv1}]$ which also preserves state (2),
    \item an action committing on the right. Since $\rset = \emptyset$ and the first committing action on the RHS is by the observer (by \textbf{awareness - single decision})
    hence reaches a state with subterm  ${\GtoP{\GT_1}{\tact{c}{n}{p}}{\lset}{\{\role p\}}{\GT_2}}$ and which is in state (3).
\end{itemize}

Any transition from (3) lead to state (3) by \textbf{monotonicity} (\Cref{lem:monoton}).

The case (3) that $\role p\not\in \lset $ implies $\lset=\emptyset$ can be proved similarly with meta states
\begin{enumerate}
\item (initial): ${\GtoDef{\GT_1}{\tdef{c}{p}}{\lset}{\rset}{\GT_2}}$ is a subterm of $ \GT$ but $ \GT$ has no MC $c,n$;
\item (initialized) ${\GtoP{\GT_1}{\tact{c}{n}{p}}{\emptyset}{\rset}{\GT_2}}$ is a subterm of $ \GT$;
\item (L-committed) ${\GtoP{\GT_1}{\tact{c}{n}{p}}{\lset\cup \{\role p\}}{\rset }{\GT_2}}$.
\end{enumerate}
In this case the only committing action on the left-hand side (main block) would be for rule $[\mathtt{LRcv1}]$ from meta-state (2) to meta-state (3) which is the only rule allowing for a committing action (i.e., altering $\lset$).
By Proposition~\ref{pro:committing} this action is by the observer.
Preservation of meta-state (3) is by \textbf{monotonicity} (\Cref{lem:monoton}).
\end{proof}

\Cref{awareproperty} shows that the observer of a MC (1) is always the one making the first action on the RHS and (2,3) is always the first one to commit to any side.

\begin{restatable}[Coherence Preservation]{lemma}{CP}
\label{lem:cohpre}
\(\GT \text{ is coherent and aware } \land ~\gspairl{\GT}{} \trans{}\gspairl{\GT'}\qquad \quad\Rightarrow \quad \GT' \text{ is coherent}\)
\end{restatable}
\begin{proof}
Coherence relies on the dependency between actions guaranteed by awareness.

Let $\GT = \Ct{C}{\GT_h}$ for some $\GT_h$. We proceed by induction on the syntax of $\Ctname$ and  inner induction on the syntax of $\GT_h$.

\paragraph{Case $\Ctname = \hole$ (base case outer induction).} In this case $\GT_h=\GT$.
 If $\GT = \tend$ the thesis is immediate.
 If $\GT$ is a communication the thesis follow by inner induction on $\GT$.
 If $\GT$ is a recursion then action $\ell$ is by rule $\mathtt{[Rec]}$ and $\GT' = {\GT'}[\mu \rv. \GT / \rv]$. $ \GT'$ is coherent by inner induction, and  $\mu \rv. \GT$ is coherent by hypothesis. It follows that $\GT'$ is coherent, hence done.
 The interesting case is for  \[\GT = \GtoP{\GT_1}{\tact{c}{n}{p}}{\lset}{\rset}{\GT_2}\] We proceed by case analysis on the last rule used to derive transition.
 The last rule applied is one of the following:
\begin{itemize}
\item {$\mathtt{[Snd]}$ or $\mathtt{[Rcv]}$}  : the thesis follows by the coherence hypothesis and the fact that $\lset$ and $\rset$ are not modified by these rules.
\item {$\mathtt{[Inst]}$},  $\mathtt{[Ctx1]}$ or $\mathtt{[Ctx2]}$ :
These rules do not change $\lset$ and $\rset$ so in case of $\mathtt{[Inst]}$ the thesis is immediate, in case of $\mathtt{[Ctx1]}$ and $\mathtt{[Ctx2]}$ it is directly by induction.
\item {$\mathtt{[LSnd]}$ or $\mathtt{[LRcv2]}$} : $\gspairl{\GT_1}{} \trans{\ell} \gspairl{\GT_1'}{\Theta'}$ with $\GT_1'$ coherent by (inner) induction. The transition does not change $\lset$, $\rset$, and $\GT_2$. Then $\GT '= \GtoP{\GT_1'}{\tact{c}{n}{p}}{\lset}{\rset}{\GT_2}$ is also coherent, hence done.
\item {$\mathtt{[LRcv1]}$} : $\gspairl{\GT_1}{} \trans{\lrcv{s}{r}{a}} \gspairl{\GT_1'}{\Theta'}$ with $\GT_1'$ coherent by (inner) induction and $a$ committing in $\GT$. Since $a$ is a committing receive action, the corresponding send action is also committing, the We have two cases:
\begin{itemize}
    \item if $\role{s} \in \lset$ then $\rset = \emptyset$ by coherence of $\GT$.
    \item if $\role{s} \not \in \lset$  then $\role s = \role p$ (i.e., $\role s$ is the observer). By premise of {$\mathtt{[LRcv1]}$} $\role s\not in \role r$ and by ~\Cref{auxx}-- \textbf{LR-initiation} (2) if the observer is not in $\rset$ then $\rset=\emptyset$.
\end{itemize}
In either of the cases above $\rset = \emptyset$, hence $\GT' = \GtoP{\GT_1'}{\tact{c}{n}{p}}{\lset\cup\{\role{q}\}}{\emptyset}{\GT_2}$ is coherent.
\item {$\mathtt{[RSnd]}$}  : in this case
$\GT' = \GtoP{\GT_1}{\tact{c}{n}{p}}{\lset}{\rset\cup\{\role{r}\}}{\GT_2'}$. Since the rule adds an element to $\rset$ we need to show that $\lset$ is empty.
We have two cases:
\begin{itemize}
    \item if $\role r = \role p$, since $\role p\not \in\lset$ (by premise of $\mathtt{[RSnd]}$) then $\lset=\emptyset$ (by Lemma~\ref{auxx}(3))
    \item if $\role r \not = \role p$ then since $\role r$ is R-acting we have $\role p\in\rset$ (by Lemma~\ref{auxx}(1)). By coherence of $\GT$ and $\role p\in\rset$ it follows $\lset =\emptyset$.
\end{itemize}
Since $\lset$ is unchanged in $\GT'$ then $\GT'$ is coherent.
\item { $\mathtt{[RRcv]}$} : Similar to the case for $\mathtt{[RSnd]}$.
\item {$\mathtt{[RTAct]}$ }: $\gspairf{\GT_1} \trans{\ell} \gspair{\GT_1'}{\Theta'}$ and $\gspairf{\GT_2} \trans{\ell} \gspair{\GT_2'}{\Theta'}$. Since $\GT_1'$ and $\GT_2'$ are coherent by induction and $\lset$ and $\rset$ remain unchanged, then
$\GT' = \GtoP{\GT_1'}{\tact{c}{n}{p}}{\lset}{\rset}{\GT_2'}$ is also coherent.
\item {$\mathtt{[Cont1]}$ or $\mathtt{[Cont2]}$}: immediate by (inner) induction.
\end{itemize}

\paragraph{Case $\Ctname$ is a prefix.} If $\Ctname$ is a prefix, either communication or message in transit, the transition of $\GT$ is by $\mathtt{[Snd]}$ or $\mathtt{[Rcv]}$, this case is similar to the first case for $\Ctname = \hole$. The case for transition by $\mathtt{[Cont1]}$ or $\mathtt{[Cont2]}$ is also similar, by inner induction.

\paragraph{Case $\Ctname = \GtoP{\hole  }{\tact{c}{n}{r}}{\lset}{\rset}{\GT_2} $.} By hypothesis, either $\lset=\emptyset$ or $\rset=\emptyset$. If
By Lemma~\ref{lem:wn} (\textbf{well-nestedness}) the form of $\Ctname$ can be assumed to be $
\GtoP{\GactShort{p}{q}{\St{}}\cup\{a.\hole\} }{\tact{c}{n}{p}}{\lset}{\rset}{\GT_2}$.
Proceed by case analysis on the last rule used in the derivation:
\begin{itemize}
\item $\mathtt{[Lsnd]}$ : in the premise, either $\Ctname$ moves by $\mathtt{[Lsnd]}$ or ${\GT}$ (the global type in the hole) moves. In the first case, only the MCs in the selected branch appear in the context after transition and, in those MCs, the sets do not change. The thesis then follows by coherence of $\GT$. If $\gspair{\GT_h}{\Theta} \trans{\ell} \gspair{\GT_h'}{\Theta'}$ by inner induction $\GT_h'$ is coherent, hence, since $\lset$, $\rset$, and $\GT_2$ are coherent by hypothesis, also $\GT'$ is coherent.
\item $\mathtt{[LRcv1]}$ : if $\Ctname$ moves by $\mathtt{[Snd]}$ this case is similar to the one above for $\mathtt{[LSnd]}$.
\item $\mathtt{[LRcv2]}$ : Similar to $\mathtt{[LSnd]}$.
\item $\mathtt{[RSnd]}$ : since this rules adds $\role{q}$ to $\lset$ by~\Cref{auxx} -- \textbf{LR-initiation} (3) --  it must be $\role p\in \lset$ and hence $\rset=\emptyset$ for coherence of $\GT$.
\item $\mathtt{[RRcv]}$ : This rule adds an element of the right-hand side set.
By~\Cref{auxx} -- \textbf{LR-initiation} (1) if $\rset$ is not empty then $\role q \in \rset$ hence  $\lset=\emptyset$.
\end{itemize}
\paragraph{Case $\Ctname = \GtoP{\GT_1 }{r}{\lset}{\rset}{ \GactShort{p}{q}{\St{}}\cup\{a.\hole\} } $.}
Similar to the symmetric one.
\paragraph{Case $\Ctname = \GtoDef{\hole}{\tdef{c}{p}}{\lset}{\rset}{\GT_2} $ or
$\Ctname = \GtoDef{\GT_1}{\tdef{c}{p}}{\lset}{\rset}{\hole } $}. Immediate by hypothesis since $\GT$ only allows an action by [Inst] which does not change $\lset$ and $\rset$.

\paragraph{Case $ \mu \rv . \hole $} By induction.

\end{proof}

Next, Since initial states are coherent, then all states reachable from an initial one are coherent.

\begin{corollary}[Coherence Invariant]
\label{coherenceinvariantI}
Any $\GT$ reachable from an initial global type that enjoys balance and awareness is coherent.
\end{corollary}

 \subsection{Progress}
\label{asec:progress}

Progress is proved by first looking at single steps (\Cref{thm:progress}). The diagram below gives a highlight of how the auxiliary lemmata and definitions contribute to \Cref{thm:progress}. The diagram outlines the dependencies of progress from the definitions (shown in blue) and lemmata given before. 

The final result \Cref{cor:gprog} follows from \Cref{thm:progress} observing that initial well set global types are aware and balanced, and that these properties are preserved by transition.

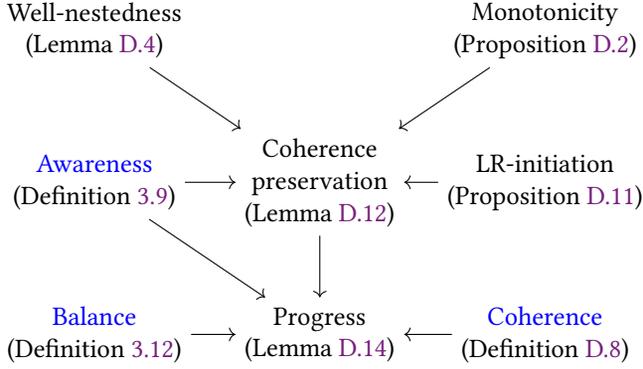
\[\centering
\begin{tikzpicture}[node distance=2cm, auto]
    \node (a) [ align=center]{\color{blue}Awareness\\(\Cref{def:aware})};
    \node (b) [right of=a, node distance=3cm, align=center] {Coherence\\ preservation\\(\Cref{lem:cohpre})};
    \node (c) [right of=b, node distance=3cm, align=center] {LR-initiation\\(\Cref{auxx})};
    \node (d) [below of=b, node distance=2cm, align=center] {Progress\\(\Cref{thm:progress})};
    \node (e) [left of=d, node distance=3cm, align=center] {\color{blue}Balance\\ (\Cref{wellset})};
    \node (f) [above of=a, node distance=2cm, align=center] {Well-nestedness\\ (\Cref{lem:wn})};
    \node (g) [above of=c, node distance=2cm, align=center] {Monotonicity\\ (\Cref{lem:monoton})};
    \node (h) [right of=d, node distance=3cm, align=center] {\color{blue}Coherence\\ (\Cref{def:coherence})};
    
    \draw[->] (a) -- (b);
    \draw[->] (c) -- (b);
    \draw[->] (b) -- (d);
    \draw[->] (e) -- (d);
    \draw[->] (h) -- (d);
    \draw[->] (g) -- (b);
     \draw[->] (f) -- (b);
     \draw[->] (a) -- (d);
\end{tikzpicture}
\]

\begin{restatable}[Progress]{lemma}{Progress}
\label{thm:progress}
If $\gspairl{\GT}{}$ is aware,  balanced, and coherent then $\GT$ enjoys progress.
\end{restatable}
\begin{proof}
We proceed by case analysis on the syntax of $\GT$. 

If $\GT = \GactShort{p}{q}{\{a_i.\GT_i\}_{i\in I}}$ since $\role r\in \RolesG{\GT}$ then either $\role r\in \{\role p,\role q\}$ or  $\role r\in \GT_i$ for all $i$ (by balance of $\GT$).
If $\role r\in \{\role p,\role q\}$ then the thesis holds immediately after either a step $!\role p\role q a_i$ or a step  $!\role p\role q a_i$ followed by $?\role p\role q a_i$.
Assume $\role r \in \GT_i$ for all $i\in I$. Fix $j\in I$, $\GT_j$ is coherent and balanced by inductive definitions of coherence and  balance. By induction,
\begin{equation}\label{eq:pro1}
\GT_j \myreaches{}\hspace{-0.2cm}\trans{\ell} \text{ with }\sbj{\ell}=\role r
\end{equation}
hence, by applying $[\mathtt{Snd}]$ and $[\mathtt{Rcv}]$ to $\GT$ and then (\ref{eq:pro1}) to $\GT_j$:
\[\GT\trans{!\role p\role q a_j} \GmsgShort{p}{q}{j}{\{a_i.\GT_i\}_{i\in I}}
\trans{?\role p\role q a_j} \GT_j \myreaches{}\hspace{-0.2cm}\trans{\ell} \text{ with }\sbj{\ell}=\role r
\]
Which satisfies the thesis. 
The case for $\GT = \GmsgShort{p}{q}{k}{\{a_i.\GT_i\}_{i\in I}}$ is similar.


If $\GT = \GtoP{\GT_1}{\tact{c}{n}{p}}{\lset}{\rset}{\GT_2}$, by \textbf{coherence} of $\GT$ (hypothesis) we can assume  $\lset=\emptyset$ or $\rset=\emptyset$. So we have three cases. 
\begin{enumerate}
\item assume $\lset = \emptyset \land \rset \neq \emptyset$. 
Since $\role r\in \RolesG{\GT} $ then 
\begin{equation}\label{eq:p0}
\role r\in \RolesG{\GT_1}\setminus \rset \cup\role r\in \RolesG{\GT_2}\setminus \emptyset \qquad \text{(by definition of $\RolesG{\GT}$)}
\end{equation}
and 
\begin{equation}\label{eq:p1}
\RolesG{\GT_1}\cup \emptyset = \RolesG{\GT_2}\cup \rset 
\qquad 
\text{(by \textbf{balance})}
\end{equation}
By combining (\ref{eq:p0}) and (\ref{eq:p1}) we know   
$\role r\in \RolesG{\GT_2}$. 

Observe that $\GT_2$ is participating, aware, coherent (by inductive definitions, and universal quantification of awareness over all MCs) and hence, by induction, $\GT_2$ enjoys progress:
\[\GT_2 \myreaches{}\hspace{-0.2cm}\trans{\ell} \text{ with }\sbj{\ell}=\role r
\]
We need to check that all actions of $\GT_2$ that bring to action $\ell$, and $\ell$ itself, can also be executed by the MC. By inspection of the rules, actions on the right-hand side of the MC can be lifted to the MC context
using one of the following rules: $\mathtt{[RSnd]}$,  $\mathtt{[RRcv]}$, or $\mathtt{[Ctx1]}$. 
Rule $\mathtt{[RRcv]}$ can only be applied. Rules $\mathtt{[RSnd]}$ and $\mathtt{[Ctx2]}$ can be applied if the subject of the action is not in $\lset$. This is the case by \textbf{coherence preservation}.


\item assume $\rset = \emptyset$ and $\lset\neq \emptyset$. By definition of $\RolesG{\GT}$  we have that $\role r\in \RolesG{\GT}$ implies $\role r\in \RolesG{\GT_1}\setminus \emptyset$ or $\role r\in \RolesG{\GT_2}\setminus \lset$. 
As in the symmetric case, by \textbf{balance} of $\GT$, we can infer $\role r\in \RolesG{\GT_1}$.
By coherence and balance of $\GT$ also $\GT_1$ is coherent, and enjoys balance. 
By induction, $\GT_1$ enjoys progress:
\[\GT_1 \myreaches{}\hspace{-0.2cm}\trans{\ell} \text{ with }\sbj{\ell}=\role r
\]


First observe that by \textbf{clear-termination} if $\role r\in \Roles{\GT_2}$ then it is also in and

To lift each action of $\GT_1$ to $\GT$ we need to use, for each transition, one of the rules below (as last applied rule) $\mathtt{[LSnd]}$, $\mathtt{[LRcv1]}$, $\mathtt{[LRcv2]}$, or $\mathtt{[Ctx1]}$. 
Rule $\mathtt{[LSnd]}$ requires that the subject of the action is not in $\rset$ which holds by since $\rset=\emptyset$ in $\GT$ and this is preserved by coherence since $\lset\not = \emptyset$. Similarly, $\mathtt{[Ctx1]}$ requires that $\rset$ does not coincide with the set of roles of the base type which is true since $\rset=\emptyset$. 

The only premise that would disallow the same transition to $\GT$ is that the subject of that action must not be in $\rset$ in $\mathtt{[LSnd]}$, $\mathtt{[LRcv1]}$, and $\mathtt{[Lrcv2]}$, which follows since in this case we assume $\rset =\emptyset$. Similarly, the premise of $\mathtt{[Ctx1]}$ that $\RolesG{\underline{\GT}}\not = \rset$ (all roles of the base type have committed on the right), holds since $\rset=\emptyset$.
\item assume $\rset = \lset = \emptyset$. By  \textbf{balance} of $\GT$, $\role r$ is in both $\RolesG{\GT_1}$ and $\RolesG{\GT_2}$.
By coherence and balance of $\GT$ also $\GT_1$ and $\GT_2$ are both coherent, and enjoy balance. 
By induction, for $i\in\{1,2\}$, $\GT_i$ enjoys progress:
\[\GT_i \myreaches{}\hspace{-0.2cm}\trans{\ell} \text{ with }\sbj{\ell}=\role r
\]
We have two cases: (1) all actions in   $\myreaches{}$ are non-committing, or (2) there is at least one committing action in $\myreaches{}$.
In case (1), the $\lset$ and $\rset$ of the intermediate states will be left empty, and hence all actions can be lifted to $\GT$ as the only premises for an action on the RHS (resp. LHS) is that the subject of the action is not in $\lset$ (resp. $\rset$). Similarly for the context rules for MC initialization. 
In case (2), without loss of generality, assume 
$\GT_i \myreaches{}$ is broken into the sequence $\GT_i \myreaches{} \trans{\ell_c}\myreaches{}$ such that all actions before $\ell_c$ are non-committing and 
$\ell_c$ is committing. Let $\GT_c$ be the state reached after $\ell_c$.
All actions before $\ell_c$ can be lifted to $\GT$ (it can be shown using the arguments in case (1) above).  Action $\ell_c$ itself can be lifted to side $i$ as the corresponding commitment set ($\rset$ or $\lset$) is empty leading to a coherent MC where one commitment set is empty and one is the singleton $\{\sbj{\ell_c}\}$. We have reduced this case to one of the cases $\lset=\emptyset \land \rset\neq \emptyset$ or $\lset\neq\emptyset \land \rset= \emptyset$ and the proof can be concluded using the arguments of these cases. 
\end{enumerate}

The other cases are either trivial or straightforward by induction. 
\end{proof}

 \section{Local Types: full definitions}
\label{asec:localtypes}

\subsection{Operational Semantics: complete rules}
\label{app:locallts}

We give the full transition rules in \Cref{afig:lsem} as some rules ($\lrule{RCtxt}$, $\lrule{NLCtxt}$, and $\lrule{NRCtxt}$) are omitted in the paper.

\begin{figure}
\small
$ \arraycolsep=1pt
\begin{array}{cr}
\inferrule{%
k \in I
\qquad
m = \msg{p}{a_k}{\pth}
}{%
\tstate{\ltheta}{\lenv}{
(\p, \tseli{\q}{a} \ST_i, \sigma), (\q, \ST, \sigma'[\p \mapsto \vv{m}])}
\trans{\p\q!\llab_k}
(\p, \ST_k, \sigma), (\q, \ST,\sigma'[\p \mapsto \vv{m}\cdot m])
}
&
{\footnotesize \lrule{Snd}}
\\
\\
\inferrule{%
k \in I
\qquad
\vv{m} = \vv{m}_1 \cdot \msg{p}{a_k}{\pth} \cdot {\vv{m}}_2
\qquad
(a, \pth) \not\in \vv{m}_1
}{%
\tstate{\ltheta}{\lenv}{
(\q, \tbrai{\p}{a}\ST_i, \sigma[\p \mapsto \vv{m}])}
\trans{\p\q?\llab_k}
(\q, \ST_k, \sigma[\p \mapsto \vv{m}_1 \cdot {\vv{m}}_2])
}
&
{\footnotesize \lrule{Rcv}}
\\
\\
\inferrule{
\tstate{\ltheta}{\lenv}{
(\p, \ST[\mu \mathtt t . \ST / \mathtt t], \sigma), Y}
\trans{\ell}
Y'
}{
\tstate{\ltheta}{\lenv}{
(\p, \mu \mathtt t . \ST, \sigma), Y}
\trans{\ell}
Y'
}
\qquad
\tstate{\ltheta}{\lenv}{
(\p, \ST \mixi{c} \ST', \sigma)}
\trans{\lnewlab}
(\p, \ST\lmixi{c} \ST', \sigma)
&
{\footnotesize \lrule{Rec/New}}
\\
\\
\inferrule{
{\begin{array}{l}
\tstate{\ltheta}{\lenv . \pleft }{
(\p, \ST_1 , \sigma), Y}
\trans{\p\q!\llab}
(\p, \ST'_1, \sigma), Y'
\end{array}}
}{
\tstate{\ltheta}{\lenv }{
(\p, \ST_1 \lmix \ST_2, \sigma), Y}
\trans{\p\q!\llab}
(\p, \ST'_1 \lmix \ST_2, \sigma), Y'
}
~
~
~
\inferrule{
{\begin{array}{l}
\tstate{\ltheta}{\lenv . \pright}{
(\p, \ST_2 , \sigma), Y}
\trans{\p\q!\llab}
(\p, \ST'_2, \sigma), Y'
\end{array}}
}{
\tstate{\ltheta}{\lenv }{
(\p, \ST_1 \lmix \ST_2, \sigma), Y}
\trans{\p\q!\llab}
(\p, \bullet \lmix \ST'_2, \sigma),Y'
}
&
{\footnotesize \lrule{LSnd/RSnd}}
\\\\
\inferrule{
{\begin{array}{l}
\tstate{\ltheta}{\lenv . \pleft }{
(\p, \ST_1 , \sigma)}
\trans{\p\q?\llab}
(\p, \ST'_1, \sigma')
\end{array}}
\quad
{\color{olive} a\in \comSetG{
 \underline{\GT}}}
}{
\tstate{\ltheta}{\lenv}{
(\p, \ST_1 \lmix \ST_2, \sigma)}
\trans{\p\q?\llab}
(\p, \ST'_1 \lmix\bullet, \sigma')
}
\quad
\inferrule{
{\begin{array}{l}
\tstate{\ltheta}{\lenv . \pleft }{
(\p, \ST_1 , \sigma)}
\trans{ \p\q?\llab}
(\p, \ST'_1, \sigma')
\end{array}}
\quad
{\color{olive} a\not \in \comSetG{
 \underline{\GT}}
}}{
\tstate{\ltheta}{\lenv }{
(\p, \ST_1 \lmix \ST_2, \sigma)}
\trans{ \p\q?\llab}
(\p, \ST'_1 \lmix \ST_2, \sigma')
}
&
{\footnotesize \lrule{LRcv1/2}}
\\[0.8cm]
\inferrule{
{\begin{array}{l}
\tstate{\ltheta}{\lenv . \pright }{
(\q, \ST_2 , \sigma)}
\trans{\p\q?\llab}
(\q, \ST'_2, \sigma')
\end{array}}
}{
\tstate{\ltheta}{\lenv }{
(\q, \ST_1 \lmix \ST_2, \sigma)}
\trans{\p\q?\llab}
(\q, \bullet \lmix \ST'_2, \sigma')
}
&
{\footnotesize \lrule{RRcv}}
\\[0.8cm]
\inferrule{
\tstate{\ltheta}{ \pth . \pleft  }{(\p, \ST, \sigma), Y}
\trans{\ell}
(\p, \ST', \sigma'), Y'
}{
\tstate{\ltheta}{\lenv }{
(\p, \ST \lmixi{} \bullet, \sigma), Y }
\trans{\ell}
(\p, \ST' \lmixi{} \bullet, \sigma'), Y'
}
\quad
\inferrule{
\tstate{\ltheta}{ \pth . \pright  }{(\p, \ST, \sigma), Y}
\trans{\ell}
(\p, \ST', \sigma'), Y'
}{
\tstate{\ltheta}{\lenv }{
(\p, \bullet \lmixi{} \ST, \sigma), Y }
\trans{\ell}
(\p, \bullet \lmixi{} \ST', \sigma'), Y'
}
&
{\footnotesize \lrule{LCtxt/RCtxt}}
\\[0.8cm]
\inferrule{
\tstate{\ltheta}{ \pth . \pleft  }{(\p, \ST_1, \sigma), Y}
\trans{\lnewlab}
(\p, \ST_1', \sigma), Y
}{
\tstate{\ltheta}{\lenv }{
(\p, \ST_1 \lmixi{} \ST_2, \sigma), Y }
\trans{\lnewlab}
(\p, \ST_1' \lmixi{} \ST_2, \sigma), Y
}
\quad
\inferrule{
\tstate{\ltheta}{ \pth . \pright  }{(\p, \ST_2, \sigma), Y}
\trans{\lnewlab}
(\p, \ST_2', \sigma), Y
}{
\tstate{\ltheta}{\lenv }{
(\p, \ST_1 \lmixi{} \ST_2, \sigma), Y }
\trans{\lnewlab}
(\p, \ST_1 \lmixi{} \ST_2', \sigma), Y
}
&
{\footnotesize \lrule{NLCtxt/NRCtxt}}
\end{array}$
\caption{Local semantics: complete rules.}
\label{afig:lsem}
\end{figure}

\subsection{Projection: definition of merge}
\label{zsec:projection}

The projection of a communication and message in transit relies on a merge operator. We give below  the full definition of merge. It applies to MC by using the first identity case for both MC definitions and active MC.

\[
\begin{array}{lll}\ST_1 \sqcap \ST_2 = \begin{cases}
    \ST_1 &~~ \ST_1 = \ST_2\\[0.5cm]
    \p \& \{a_i.\ST_i^1\}_{i\in I\setminus J} \cup \{b_j.\ST_j^2\}_{j\in J\setminus I} \cup \{a_i. \ST_i^1  \}_{i\in I\cap J}
    & \begin{array}{ll}
    \ST_1 = \p \& \{a_i.\ST_i^1\}_{i\in I}, \text{ and }\\
    \ST_2 = \p \& \{b_j.\ST_j^2\}_{j\in J} \text{ and }\\
    \forall i\in I\cap J.\, a_i.\ST_i^1 = b_i.\ST_i^2\end{array}\\[0.8cm]
    \mu \mathtt t . ( \ST_1' \sqcap \ST_2') &
    \begin{array}{ll}
    \ST_i = \mu \mathtt t . \ST_i',~~i\in\{1,2\}  \end{array}\\[0.5cm]
    \bot & ~\text{ otherwise }
\end{cases}
\end{array}
\]

 \section{Operational Correspondence}
\label{app:correspondence}

\subsection{Full definition of $<:$}
\label{asec:st}
The full definition of preorder `$<:$' on pairs of local types, and its lifting to configurations and systems is defined below.

\[\begin{array}{c}
\inferrule{
\ST_i' \mathbin{<:} \ST_i \quad (\forall i\in I)\quad \wild\in\{\&,\, \oplus\}
}{
\p \, \wild_{i \mathclose{\in} I} \, a_i.\ST_i'
\mathbin{<:}
\p \, \wild_{i \mathclose{\in} I} \, a_i.\ST_i
}
\qquad

\inferrule{
\ST_k' \mathbin{<:} \ST_k \quad k\in I
}{
\p \, \&_{i \mathclose{\in} I} \, a_i.\ST_i'
\mathbin{<:}
\p \, \& \, a_k.\ST_k
}
\\
\\
\inferrule{
\ST_i' \mathbin{<:} \ST_i
\quad
(\forall i \in \{1, 2\} )\quad \wild\in\{\mix, \lmix\}
}{
\ST_1' \, \wild \, \, \ST_2'
<:
\ST_1 \, \wild \, \, \ST_2
}
\\
\\
\inferrule{
\ST_1' \mathbin{<:} \ST_1
\qquad
\ST_2' \mathbin{<:} \ST_2}
{
\ST_1' \mixi{c} \ST_2'
~\mathbin{<:}~
\ST_1 \lmixi{c} \ST_2
}
\quad
\inferrule{
\ST_1' \mathbin{<:} \ST_1
}{
\ST_1' \lmixi{c} \bullet
~ \mathbin{<:}~
\ST_1 \lmixi{c}
\bullet
}
\quad
\inferrule{
\ST_2' \mathbin{<:} \ST_2
}{
\bullet \blacktriangleright^c \ST_2'
~ \mathbin{<:}~
\bullet \blacktriangleright^c \ST_2
}
\\
\\
\inferrule{
\ST' \mathbin{<:} \ST
}{
 \mu \mathtt{t}.\ST' \mathbin{<:} \mu \mathtt{t}.\ST
}
\qquad
\inferrule{
 \ST' [\mu \mathtt{t}. \ST'/\mathtt{t}]\mathbin{<:} \ST
}{
\mu \mathtt{t}. \ST' \mathbin{<:} \ST
}
\qquad
\inferrule{
}{
\mathtt{t} \mathbin{<:} \mathtt{t}
}
\qquad
\inferrule{
}{
\mathsf{end} \mathbin{<:} \mathsf{end}
}\\[0.5cm]
\hline
\\
\inferrule{
\ST' \mathbin{<:} \ST
\quad
}{
(\p, \ST', \sigma )
\mathbin{<:}
(\p, \ST, \sigma)
}~~\mathtt{[PConfig]}
\qquad
\inferrule{
\forall \role{r} \mathop{\in} R ~~ Y_\role{r}'\mathbin{<:} Y_\role{r}
}{
\{Y_{\role{r}}' \}_{\role r \in R} \mathbin{<:} \{Y_{\role{r}}\}_{\role{r} \mathclose{\in} R}
}~~~~\mathtt{[PSystem]}
\end{array}
\]

 \subsection{Fidelity}

\subsubsection{Auxiliaries}

First, observe that projected local configurations do not have garbage. This simplifies our reasoning about the correspondence.

\begin{lemma}
\label{nostale}
Let $\pth_{out} \vdash \GT \mathbin{\proj} \role{r} = \ST , \sigma$ for $\role{r}\in\RolesG{\GT}$. Then for all $\p\in \RolesG{\GT},\, (a,\pth)\in \sigma(\p)$,
$\mathtt{stale}(\pth_{in})=\mathtt{false}$ where $\pth = \pth_{out} . \pth_{in}$.
\end{lemma}
\begin{proof}
By induction on the syntax of $\GT$.
The base cases for $\GT\in\{\tend, \, \mathtt t\}$ follow by observing that forall $\pth$, $\mathtt{stale}(\pth,\tend)=\mathtt{stale}(\pth,\mathtt t)=\mathtt{false}$.
If $\GT =\GactShort{\p}{\q}{ \{a_i : \GT_i\}_{i\in I}}$, in the first two cases of the corresponding projection rule the local type has a send or a receive prefix, hence no path is stale for them. The third case using merge follows by induction. Similarly, if $\GT = \GT_1 \mix \GT_2$ and $\GT=\mu \mathtt t. \GT'$: the projection is not stale for any path.

The interesting case is the one for $\GT = \GT_1 \lmix \GT_2$. Since $\ST$ is a mixed choice then $\pth_{out}'$ is not empty (by projection, one side, left or right is added to the path). So, without loss of generality, assume $\pth = \pth_{out} . \mathtt{s} . \pth_{in}$ with $\mathtt{s}\in\{\pleft, \pright\}$.
If $\role r\in \lset$ then by induction forall $(a, \pth)$ in some of $\role r$'s queues $\sigma$, we have $\mathtt{stale}(\pth_{in}, \ST_1) = \mathtt{false}$ with $\pth = \pth_{out}.\mathtt{s}.\pth_{in}$. Since $\ST = \ST_1 \lmix \bullet$ hence $\mathtt{s} = \pleft$, hence
\[\mathtt{stale}(\pth_{in}, \ST_1) = \mathtt{false}~~\Longrightarrow ~~ \mathtt{stale}(\pleft.\pth_{in}, \ST_1 \lmix \bullet) = \mathtt{false}\]
as desired.

The case for $\role r\in\lset$ is symmetric to the case above for $\role r \in \rset$. The otherwise case for $\role r\not\in \lset\cup\rset$ also follows by observing that, whatever the MC side that originated the message,
\[\mathtt{stale}(\pth_{in}, \ST_i) = \mathtt{false}~~\Longrightarrow ~~ \mathtt{stale}(\mathtt{s}_i.\pth_{in}, \ST_i \lmix \bullet) = \mathtt{false}\]
with $\mathtt{s}_1 = \pleft$ and $\mathtt{s}_2 = \pright$.
\end{proof}

\begin{lemma}[Garbageless Projections]
\label{nogarbage}
Let $\epsilon \vdash \GT \mathbin{\proj} \role{r} = \ST , \sigma$ for a $\role{r}\in\RolesG{\GT}$. Then
\[(\role r, \sigma, \ST) \trans{\gcl}(\role r, \sigma', \ST')~~\Rightarrow~~ \sigma=\sigma' \]
\end{lemma}
\begin{proof}
    This lemma follows from \Cref{nostale} with $\pth = \epsilon$, observing that in absence of stale messages garbage collection on $\sigma$ returns $\sigma$ itself.
\end{proof}

In \Cref{lem:proj} (and in the rest of the paper), by   projectability of $\GT$ on $\p$ we intend the fact that the projection of $\GT$ on $\p$ returns.
We say that $\GT$ is \emph{projectable} if the (partial) projection function $\GT\proj \p$ returns for all $\p \in\RolesG{\GT}$.

\begin{proposition}[Projectability Preservation]\label{lem:proj}
If $\GT$ is projectable and ${\GT}\trans{\ell}\GT'$ then $\GT'$ is projectable.
\end{proposition}


\begin{proposition} \label{emptyqueue}
    If $\pth \vdash \GactShort{p}{q}{S}\proj \q = \ST, \sigma$ then $\sigma(\p)$ is the empty message sequence $\epsilon$.
\end{proposition}

\begin{proposition}[Determinism of $\sigma$]\label{pro:sigmadet}
\[\GT_1,\sigma \trans{\ell}  \GT_1',\sigma_1 ~~\land~~\GT_2,\sigma \trans{\ell}  \GT_2',\sigma_2
~~\Longrightarrow
\sigma_1=\sigma_2\]
\end{proposition}

\begin{proposition}\label{unfoldmerge}
If $\mu \mathtt t. \ST_1 \sqcap\mu \mathtt t. \ST_2$ is defined then $\ST_1[\mu \mathtt t. \ST_1 /\mathtt t] \sqcap \ST_2[\mu \mathtt t. \ST_2 /\mathtt t]$ is defined.
\end{proposition}

\begin{proposition}\label{refoldmerge}
$
\ST_1 [\mu \mathtt t. \ST_1'/\mathtt t] \sqcap \ST_2 [\mu \mathtt t. \ST_2'/\mathtt t] = \ST_1 \sqcap \ST_2 [\mu \mathtt t. \ST_1'\sqcap \ST_2'/\mathtt t]$.
\end{proposition}

\begin{proposition}[Merge and $<:$]\label{merge1}
Consider local types $\ST_1^1$, $\ST_1^2$, $\ST_1^2$, $\ST_2^2$. Assume that for all $i\in \{1,2\}$, $\ST_i^1 <: \ST_i^2$ and $\ST_1^i \sqcap \ST_2^i$ is defined. Then
\[\ST_1^1 \sqcap \ST_2^1 <: \ST_1^2 \sqcap \ST_2^2\]
\end{proposition}
\begin{proof}

By induction on the depth of the derivation syntax of $\ST_1^1 <: \ST_2^1$.

\paragraph{Base cases.} There are two axioms in the judgments in \Cref{asec:st}, for $\mathtt t$ and $\tend$. If $\ST_1^1 = \mathtt t $ it can only be $t <: \ST_1^2 = \mathtt t$.
Since $\ST_1^1 \sqcap \ST_2^1$ is defined, then $\ST_2^1=\mathtt t$ and by definition of `$<:$' we have  $\ST_2^2 = \mathtt t$. Since $\mathtt t \sqcap \mathtt t = \mathtt t$ the thesis is immediate by hypothesis. The case for $\ST_1^1 = \tend$ is similar.

\paragraph{Inductive cases - prefix rule.}
If the \textbf{first rule} is used then $\ST_1^1$ is either a sending or a receiving prefix. If $\ST_1^1 = \p\&_{i\in I} a_i.\ST_i$ then by definition of `$<:$' we have
$\ST_1^2 = \p\&_{i\in I} a_i.\ST_i'$ with
\begin{equation}\label{m11}
\ST_i<:\ST_i'    \qquad \forall i \in I
\end{equation}
Note that by `$<:$' both $\ST_1^1$ and $\ST_1^2$ have choices over the same set $I$.
By definedness of $\ST_1^1\sqcap \ST_2^1$ then $\ST_2^1 = \p\&_{j\in J} a_j.\ST_j''$.
By $\ST_2^1 <: \ST_2^2$ we have $\ST_2^2 =  \p\&_{j\in J}a_j.\ST_j'''$
with
\begin{equation}\label{m12}
\ST_i''<:\ST_i'''    \qquad \forall j \in J
\end{equation}
By definition of merge, $\ST_1^1\sqcap \ST_2^1$ has the following form:
\[ \ST_1^1\sqcap \ST_2^1 = \p\& \{a_i. \ST_i\}_{i\in I\setminus J}~\cup ~ \{a_j. \ST_j''\}_{j\in J\setminus I} ~\cup ~ \{ a_i.\ST_i\sqcap \ST''_i\}_{i\in I\setminus J} \]
and similarly
\[ \ST_1^2\sqcap\ST_2^2 = \p\& \{a_i.\ST_i'\}_{i\in I\setminus J}~\cup~\{a_j.\ST_j'''\}_{j\in J\setminus I} ~\cup~ \{ a_i.\ST_i'\sqcap \ST_i'''\}_{i\in I\cap J}\]
Since merge is defined inductively, if $\ST_1^1\sqcap\ST_2^1$ (resp. $\ST_2^1\sqcap\ST_2^2$) is defined then
all $\ST_i\sqcap \ST''_i$ (resp. $\ST_i'\sqcap \ST'''_i$) are defined for $i\in I\setminus J$. Therefore we can apply induction on the continuations obtaining
\begin{equation}\label{m13}
\ST_i\sqcap \ST_i'' <: \ST_i'\sqcap \ST_i'''     \qquad \forall i \in I\cap J
\end{equation}
Using (\ref{m11}), (\ref{m12}) and (\ref{m13}) as premises for `$<:$' (first rule) we obtain the thesis $\ST_1^1\sqcap\ST_2^1<:\ST_1^2\sqcap\ST_2^2 $ for this case.

If $\ST_1^1 = \p\oplus_{i\in I} a_i.\ST_i$ the thesis is immediate by induction since
$\ST_1^1\sqcap \ST_2^1 = \ST_1^1$.

\paragraph{Inductive cases - MC rules.} The cases for MC, where $\ST_1^1$ and $\ST_2^1$ have the same type of mixed choice (second, fourth and fifth rules of `$<:$' in ~\Cref{asec:st}) is also immediate by induction. The case for the third rule also follows by induction observing $\ST_1^1$ and $\ST_2^1$ have the same MC configurations, and so do $\ST_1^2$ and $\ST_2^2$, and also observing that the definition of merge does not alter such configuration, being the identity (first rule of merge).

\paragraph{Inductive cases - recursion rules.}
If $\ST_1^1 = \mu \mathtt .\ST_1'$ we have two cases. If $\ST_2^1 =\mu \mathtt .\ST_2'$ then the thesis is by induction using the first rule in the third line (\Cref{asec:st}).
If unfolding is needed then the rule applied to derive $\ST_i^1 <: \ST_i^2 $ is
\[\inferrule{
 \ST_i' [\mu \mathtt{t}. \ST_i'/\mathtt{t}]\mathbin{<:} \ST_i^2
}{
\ST_{i}^1 = \mu \mathtt{t}. \ST_i' \mathbin{<:} \ST_i^2
}\]
By hypothesis $\mu \mathtt{t}. \ST_1'\sqcap \mu \mathtt{t}. \ST_2'$ is defined and hence by \Cref{unfoldmerge}
$\ST_1' [\mu \mathtt{t}. \ST_1'/\mathtt{t}] \sqcap \ST_2' [\mu \mathtt{t}. \ST_2'/\mathtt{t}]$ is defined. By induction, therefore,
\begin{equation}\label{merenda}
\ST_1' [\mu \mathtt{t}. \ST_1'/\mathtt{t}] \sqcap \ST_2' [\mu \mathtt{t}. \ST_2'/\mathtt{t}]\mathbin{<:} \ST_i^2\sqcap\ST_2^2\end{equation}
By \Cref{refoldmerge},
$\ST_1' [\mu \mathtt{t}. \ST_1'/\mathtt{t}] \sqcap \ST_2' [\mu \mathtt{t}. \ST_2'/\mathtt{t}] = \ST_1' \sqcap \ST_2' [\mu \mathtt{t}. \ST_1'\sqcap \ST_2'/\mathtt{t}]$ hence (\ref{merenda}) becomes
\begin{equation}\label{merenda1}
\ST_1' \sqcap \ST_2' [\mu \mathtt{t}. \ST_1'\sqcap \ST_2'/\mathtt{t}]\mathbin{<:} \ST_i^2\sqcap\ST_2^2\end{equation}
By applying \ref{merenda1} as a premise for the second rule in the third row (\Cref{asec:st}) we obtain the thesis.
\end{proof}

\begin{lemma}[Send-project] \label{lem:prestale}
Let $\GT\proj = (\role{q}, \ST, \sigma),Y$ and assume that
$(\role q, \ST, \sigma),Y\trans{\p\role q !\llab} (\role {q}, \ST', \sigma'),Y'$ and $\GT \trans{\p\role q !a}\GT'$, then $\ST=\ST'$.
\end{lemma}
\begin{proof}
    We give a proof sketch. By analysis of the global semantics rules there exists $\Ctname$ such that $\GT =\Ctname[\GactShort{p}{q}{\St{}}]$ and  $\GT' =\Ctname[\GmsgShort{p}{q}{k}{\St{}}\cup\{a_k.\Ctname\}]$.
The thesis is by induction on the structure of $\Ctname$, where in the base case for $\Ctname = [\,]$ we observe that $\GactShort{p}{q}{\St{}} \proj \role q$
and $\GmsgShort{p}{q}{k}{\St{}}\cup\{a_k.\Ctname\} \proj \role q$ give the same local type.
\end{proof}
\begin{lemma}[Stale]
\label{lem:stale}
Let $\GT\proj = (\role{q}, \ST, \sigma),Y$ and assume that
$(\role q, \ST, \sigma),Y\trans{\p\role q !\llab} (\role {q}, \ST, \sigma'),Y'$ and $\GT \trans{\p\role q !a}\GT'$. Then the following hold:

\[ \neg \mathtt{stale}(\pth, \ST) \quad \text{iff} \quad \sigma' =\sigma[\p \mapsto \sigma(\p)\cdot (a,\pth)] \]
Namely, projection includes in the queue all and only messages that are not stale.
\end{lemma}
\begin{proof}
Observe that the last message sent is at the tail of the queue $\sigma(\p)$. By induction on the derivation of the local transition, proceeding by case analysis on the last rule applied.

\paragraph{Base case - axiom $\mathtt{[Snd]}$}
In this case we have no MC context. So
$\GT$ has the form of an interaction from $\p$ to $\role q$ (either top level or as continuation of other interactions from different participants). $\GT \proj \p = \role q \oplus\{a_i. \ST_{\p i}\}_{i\in I}$ and $\GT \proj \role q = \p \& \{a_i. \ST_{\role q i}\}_{i\in I}$.
By $\mathtt{[Snd]}$ $\sigma' =\sigma [\p \mapsto \sigma(\p)\cdot (a, \pth)]$. Since $\pth = 0$ we have trivially $\neg \mathtt{stale}(0, \ST)$ hence done.

\paragraph{Inductive cases}
There are four send rules in the local semantics: $\mathtt{[LSnd]}$, $\mathtt{[LCtxt]}$, $\mathtt{[RSnd]}$, and $\mathtt{[RCtxt]}$. The cases for $\mathtt{[LSnd]}$ and $\mathtt{[LCtxt]}$ – send action on the left – proceed similarly, because the local type in the configuration of $\p$ (i.e., whether sender $\p$ is committed in the outmost MC) is not relevant to the proof. Similarly $\mathtt{[RSnd]}$ and $\mathtt{[RCtxt]}$ proceed in the same way for send action on the right. We only show the cases for $\mathtt{[LSnd]}$, and $\mathtt{[RSnd]}$.  Without loss of generality assume that $\GT=\GT_1\blacktriangleright \GT_2$ and $\ST$ is a MC.

If the local action is by $\mathtt{[LSnd]}$ then we can infer $\pth = \pleft . \pth'$ for some $\pth'$, and also that $\p$'s local type is a MC. By inspection of the global rules and hypothesis we can infer $\GT_1 \trans{\p\role q !a}\GT_1'$ for some $\GT_1'$. Now $\ST$ can have one of the following forms:
\begin{itemize}
    \item $\ST = \ST_1 \blacktriangleright \ST_2$. Let $\sigma_L$ and $\sigma_R$ be the queues obtained by projecting $\GT_1$ and $\GT_2$, on $\role q$. Recall that actions on the RHS are always committing, and that $\p$ must not be committed to the RHS when making a LHS step, hence $\sigma_R = \epsilon$.
    By projection (active MC, third case): \begin{equation}
    \label{eqactive}
    \sigma = \sigma_L \circ \sigma_R = \sigma_L \quad \sigma' = \sigma_L'\circ\sigma_R = \sigma_L' \quad \text{(with  $\sigma_L'$ projection of $\GT_1'$ on $\role q$ and  $\sigma_R = \epsilon$)}\end{equation}
    By induction $\neg \mathtt{stale}(\pth', \ST_1)$ iff  $\sigma_L' = \sigma_L[ \p\mapsto  \sigma_L(\p)\cdot (a, \pth)]$ which gives the thesis by observing that $\neg \mathtt{stale}(\pth', \ST_1)$ iff $\neg \mathtt{stale}(\pleft.\pth’, \ST)$ by definition of  $\mathtt{stale}$.
    \item $\ST = \ST_1 \blacktriangleright \bullet$.
    Again, let $\sigma_L$ be the queue obtained by projecting $\GT_1$ on $\role q$.
    By the projection (active MC, second case) $\sigma = \sigma_L$ and $\sigma' = \sigma_L'$ (for $\sigma_L'$ projection of $\GT_1'$ on $\role q$). By induction $\neg \mathtt{stale}(\pth', \ST_1)$ iff  $\sigma_L' = \sigma_L[ \p\mapsto (a, \pth), \sigma_L(\p)]$ which gives the thesis via observing $\neg \mathtt{stale}(\pth', \ST_1)$ iff $\neg \mathtt{stale}(\pleft.\pth’, \ST)$ by definition of  $\mathtt{stale}$.
    \item $\ST = \bullet \blacktriangleright \ST_2$. In this case $\neg \mathtt{stale}(\pleft.\pth', \ST)$ which yields the \emph{if} direction. For the \emph{only-if} case we need to show $\sigma'\neq \sigma[ \p\mapsto (a, \pth)\cdot \sigma(\p)]$ which follows by $\sigma'= \sigma_R = \sigma$. Namely $(a, \pth)$ is not in $\sigma'$.
\end{itemize}

If the action is by $\mathtt{[RSnd]}$ then  $\GT_2\trans{\p\role q!a}\GT_2'$ and $\pi = \pright. \pi'$. We have three cases.
\begin{itemize}
    \item $\ST = \ST_1 \blacktriangleright \ST_2$. We have $\sigma = \sigma_L\circ \sigma_R$ and $\sigma' = \sigma_L\circ\sigma_R'$ for some $\sigma_R'$ resulting from projection of $\GT_2'$ on $\role r$. In this case, however, it can be that both $\sigma_L$ and $\sigma_R'$ are non empty. By induction $\neg \mathtt{stale}(\pth', \ST_2)$ iff $\sigma_R' = \sigma_R[ \p\mapsto (a, \pth)\cdot \sigma_R(p)]$.
    The thesis follows from observing: (a) $\neg \mathtt{stale}(\pth', \ST_2)$ iff $\neg \mathtt{stale}(\pright.\pth', \ST)$ (by definition of $\mathtt{stale}$), (b) $\sigma_L \circ \sigma_R' = \sigma_L \circ \sigma_R[ \p\mapsto (a, \pth)\cdot \sigma_R(p)]$ (by induction) which by definition of $\circ$ (which is not commutative and concatenates the messages of $\sigma_R(\p)$ \emph{after} those of $\sigma_L(\p)$ we have $\sigma_L \circ (\sigma_R[ \p\mapsto (a, \pth)\cdot \sigma_R(p)]) = (\sigma_L \circ \sigma_R) [ \p\mapsto (a, \pth)\cdot \sigma_L(p)\circ \sigma(R)(\p)] = \sigma'$ as required.
    \item $\ST = \ST_1 \blacktriangleright \bullet$. This case is symmetric to the corresponding (third) case for $\mathtt{[LSnd]}$ but observing $\neg \mathtt{stale}(\pright.\pth', \ST)$ for the \emph{if} direction and $\sigma' = \sigma_L =\sigma$ (the new message is not in $\sigma'$).
    \item $\ST = \bullet \blacktriangleright \ST_2$ is  symmetric to case (2) for $\mathtt{[LSnd]}$.
\end{itemize}
\end{proof}

 \subsection{Bottom-up fidelity}
\label{asec:fidelity}

\fullfidelity*

\begin{proof}
Property (2) follows from \Cref{nogarbage} (projected systems have no stale messages). We focus on property (1).

Let $\pth \vdash \GT \mathbin{\proj} \role{r} = \ST_{\role r}, \sigma_{\role r}$ for all $\role{r}\in\RolesG{\GT}$.
We also set $Y_{\role r} = (\role r, \ST_{\role r}, \sigma_{\role r})$ and $Y_{\GT} = \{ Y_{\role r}\}_{\role r\in \RolesG{\GT} }$.
The move by $Y$ is by $\mathtt{[Par]}$ or   $\mathtt{[Low]}$. The case for $\mathtt{[Discard]}$ is not possible since $\ell\neq\rho$. The case for $\rho$ actions is handled in (2).

Without loss of generality we We focus on

In case of $\mathtt{[Par]}$, and then one of the rules for configurations in Figure~\ref{fig:lsem2},
we reason by induction on the syntax of $\GT$.

\paragraph{1. \textbf{Communication}} Assume $\GT = \role p \rightarrow \role q : \{ a_i : \GT_i \}_{i\in I}$ and, without loss of generality, \[Y_{\GT} = Y_{\p} , Y_{\role q}, Y_{R}\] where
\[\begin{array}{ccl}
     Y_{\p} & = & (\role p, ~\tseli{\role q}{a}\ST_{\p i}, ~\sigma_{\p})\\
     Y_{\q} & = &  (\q,  ~{\p}\&_{i\in I} {a_i}.\ST_{\q i}, ~\sigma_{\q}[\p\mapsto \epsilon ])\\
     Y_{R} & =  & \{(\role r, ~\sqcap_{i\in I} \ST_{\role r i}, ~\sigma_{\role r})\}_{\role r\in\RolesG{\GT}\setminus\{\p,\q\}}
\end{array}\]
noting that $\sigma_{\q}(\p)$ is empty (\Cref{emptyqueue}). Without loss of generality, the transition $Y_{\GT} \trans{\ell} $ is by: (1) a send action by $Y_{\p}$ \footnote{Role $\role q$ cannot move and its first action needs to be the reception of the message from $\p$ since local types are mono-threaded by construction)}, or (2) an action by a role $\role r$ in $Y_{R}$, or (3) a MC instantiation by a configuration in $Y_{R}$.

In case (1) the last rule applied (after $\lrule{Par}$) is $\lrule{LSnd}$ of the local semantics:
\[\tstate{\Theta}{\lenv}{
Y_{\GT}}
\trans{\role p\role q!\llab_k}
(\p, \ST_{\p k}, \sigma_{\p}), (\q,  ~{\p}\&_{i\in I} {a_i}.\ST_{\q i}, ~\sigma_{\q}[\p\mapsto (a_k, \pth') ]), Y_{R}
\]
for some path $\pth'$.
By rule $\mathtt{[Snd]}$ of the global semantics
\[\gspair{\GT}{\Theta}\trans{\p\q!a_k}{\p \rightsquigarrow \q : k \{ a_i : \GT_i \}_{i\in I}} = \GT'\]
The projection of $\GT'$ on role $\p$ is
$(\p, \ST_{\p k}, \sigma_{\p})$, the projection on role $\q$ is $(\q, {\p}\&_{i\in I} {a_i}.\ST_{\q i}, \sigma_{\q}[\p\mapsto (a_k,\pth')]$, and the projection on the roles of the configurations in $Y_{R}$ is unchanged. Note that the definition of merge (for each role $\role r$ in $Y_R$ is not affected by the selection of branch $k$ by $\p$).

Consider now case (2) where the action is by a configuration $(\role r,\, \ST_{\role r},\, \sigma_{\role r})$ in $Y_{R}$ with $\mathtt{sbj}(\ell)=\role r$.
Let, for all $i\in I$, $\GT_{i}\proj \role r = \ST_{\role r i}, \sigma_{\role r i}$. Observe $\ST_{\role r}=\sqcap_{i\in I} \ST_{\role r i}$ and by projection
\begin{equation}\label{samesig}
 \forall i,j\in I.~\sigma_{\role r i} =   \sigma_{\role r j}
\end{equation}
We will therefore denote any $\sigma_{\role r i}$ as $\sigma_{\role r}$.
The transition has the following form:
\[Y_{\GT} = (\role r,\, \ST_{\role r},\, \sigma_{\role r}), (\role t,\, \ST_{\role t},\, \sigma_{\role t}),Y'' \trans{\ell}(\role r,\, \ST_{\role r}',\, \sigma_{\role r}'), (\role t, \ST_{\role t}, \sigma_{\role t}'), Y''\]
Note that if $\ell$ is a send action (assume to is from role $\role t$ without loss of generality) then $\sigma_{\role r}=\sigma_{\role r'}$, and if $\ell$ is a receive action (assume it is from role $\role t$ without loss of generality) then
$\sigma_{\role t}=\sigma_{\role t'}$.

By induction, for all $i\in I$
\begin{equation}\label{m1}\GT_i \trans{\ell} \GT_i'
\end{equation}
and
\begin{equation}\label{m2} \ST_{\role r i}',\, \sigma_{\role r i}' \myreaches{\gcl} \ST_{\role r i}',\, \sigma_{\role r i}^{g} <: \GT_i'\proj{\role r}
\end{equation}
By (\ref{samesig}) follows that (\ref{m2}) is equivalent to
\begin{equation}\label{m23} \ST_{\role r i}',\, \sigma_{\role r}' \myreaches{\gcl} \ST_{\role r i} ,\,\sigma_{\role r i}^{g}  <: \GT_i'\proj{\role r}
\end{equation}
and by \Cref{pro:sigmadet} we have that (\ref{m23}) is equivalent to
\begin{equation}\label{m24} \ST_{\role r i}',\, \sigma_{\role r}' \myreaches{\gcl} \ST_{\role r i},\, \sigma_{\role r}^{g}  <: \GT_i'\proj{\role r}
\end{equation}
By using (\ref{m1}) as a premise in rule $\mathtt{[Cont1]}$ of the global semantics we obtain
\[\GT \trans{\ell}\GT'\]
It remains to show the second part of the thesis:  \begin{equation}\label{mthe}
Y' \myreaches{\gcl}Y_{g} <: \GT'\hspace{-1mm}\proj\end{equation}

We decompose the  reasoning, considering each configuration in $Y'$:
\begin{itemize}
\item for role $\role r$ we need to show
\[(\role r, \ST_{\role r}', \sigma_{\role r}') \myreaches{\rho} <: \GT'\proj \role r\]
The first part (garbage collection)
\[(\role r, \ST_{\role r}', \sigma_{\role r}') \myreaches{\rho}(\role r, \ST_{\role r}', \sigma_{\role r}^g) \]
follows by (\ref{m24}) and \Cref{pro:sigmadet}.
The second part (preorder)
\[(\role r, \ST_{\role r}', \sigma_{\role r}^g) <: \GT'\proj \role r\]
follows by using (\ref{m2})  as hypothesis \Cref{merge1} (merge with `$<:$').

\item For $\role t$, the thesis follows again by induction and \Cref{pro:sigmadet}, observing that all $\GT_i'\proj \role t$ are the same for all $i\in I$.
\item For all other roles the thesis holds by hypothesis since the local types are unchanged, and the queues as well.
\end{itemize}
In case (3) if MC is an instantiation action, the thesis is by induction.

\paragraph{2. \textbf{Message in transit}} Assume
$\GT = \GmsgShort{p}{q}{k}{\{a_i.\GT_i\}_{i\in I}}$ and, without loss of generality, \[Y_{\GT} = Y_{\role q}, Y_{R} \qquad  Y_{\q}  =   (\q,  ~{\p}\&_{i\in I} {a_i}.\ST_{\q i}, ~\sigma_{\q}[\p\mapsto (a_k,\pth') ])\]
There are two cases: (1) the transition of $Y_{\GT}$ is a receive action by $Y_{\q}$ or (2) an action by one other configuration in $Y_{R}$.

In case (1) we have by $\lrule{LRcv}$:
\[Y_{\q} \trans{\p\q?a_k} (\q,  ~\ST_{\q k}, ~\sigma_{\q}[\p\mapsto \epsilon ])=Y_{\q}'\]
and by $\lrule{Par}$
\[Y_{\GT} \trans{\p\q?a_k} Y_{\q}',Y_{R}\]
By rule $\mathtt{[Rcv]}$ of the global semantics
\[ \GT = \gspair{\GmsgShort{p}{q}{k}{\{a_i.\GT_i\}_{i\in I}}}{\Theta}\trans{\p\role q ? a_k} \GT_k\]

It remains to show the second part of the thesis:  \begin{equation}
Y' \myreaches{\gcl}Y_{g} <: \GT'\hspace{-1mm}\proj\end{equation}

We decompose the  reasoning, considering each configuration in $Y'$:
\begin{itemize}
\item For $\q$, by induction
\[\ST_{\role q k},\sigma_{\q}'\myreaches{\gcl}\ST_{\role q k},\sigma_{\q}^g <: \GT_k\proj \q\]
which directly yields the thesis.
\item The cases for other configurations are straightforward by hypothesis as $Y_R$ is nto changed.
\end{itemize}

Case (2) where a configuration in $Y_R$ makes a step in the continuations $\GT_i$ proceeds analogously to case (2) for communications.

\paragraph{4. \textbf{MC definition}}

Assume
$\GT=\GtoDef{  \GT_1}{c}{\lset}{\rset}{  \GT_2}$ and instantiation is triggered by configuration of role $\role r_1$ in $Y_{\GT}$.
Observe that by projection $\sigma_1 = \sigma_0$ and
$\ST_1 =\ST^1_1 \mixi{c} \ST^2_1$ with
\begin{equation}\label{fs1}
 \pth\vdash \GT_1 \mathbin{\proj} \role{r_1} = \ST_1^1,\sigma_0
\end{equation}

By the form of $\ST$ the transition can only be by rule $\lrule{New}$ of the local semantics:

\[\tstate{\ltheta}{\lenv}{
(\role {r1}, \ST_1^1 \mixi{c} \ST^2_1, \sigma_0)}, Y_2
\trans{\lnewlab}
(\role {r_1}, \ST_1^1 \lmixi{c} \ST^2_1, \sigma_0), Y_2 \]

By rule $\mathtt{[Inst]}$ of the global semantics
\[\gspair{\GtoDef{  \GT_1}{c}{\lset}{\rset}{  \GT_2}}{\Theta}\trans{\lnewlab}{\GtoP{  \GT_1}{c}{\emptyset}{\emptyset}{  \GT_2}} = \GT'\]
The projection of $\GT'$ is as follows:
\[ \pth \vdash {\GtoP{  \GT_1}{c}{\emptyset}{\emptyset}{  \GT_2}} \proj \role{r_1} =  \ST_1^1 \blacktriangleright^{c}  \ST^2_1, \sigma_0
\]
where (\ref{fs1}) was used to develop the projections of $\GT_1$ and $\GT_2$ on $\p$,  to derive $\sigma_0$ as $\sigma_0 \circ \sigma_0$ as desired for $Y_1$.
To conclude this case we need to consider $Y = Y_1,Y_2$. Observe that $Y_2$ is unchanged by the transition, that $\GT'\proj \role {r_2} = \ST_2^1 \blacktriangleright^{c}  \ST^2_2, \sigma_0$ and that
$ \ST_2^1 \mixi{c}  \ST^2_2 <: \ST_2^1 \lmixi{c} \ST_2^2$ as desired.  This case, when only one configuration instantiate a MC -- not immediately paired with the other ones, unlike the global case, justifies the need for the preorder $<:$ in the statement.

\paragraph{4. \textbf{Active MC}}
In this case $\GT = \GtoP{\GT_1}{c}{\lset}{\rset}{\GT_2}$.
Without loss of generality we consider three cases in which: (1) $\role r_1$ is not in $\rset\cup\lset$, (2) $\role r_1\in\lset$, and (3) $\role r_1\in\rset$.

In (1), $Y_1 = (\role r_1, \ST_{11} \blacktriangleright^c \ST_{12} , \sigma_1)$ and $Y_2 = (\role r_2, \ST_2, \sigma_2)$.
We first consider the case where the last rule applied is on the left. We show the case for $\lrule{LSnd}$. The cases for $\lrule{LRcv1}$
 and $\lrule{LRcv2}$ proceed similarly (but with no need of garbage collection - the only interesting but is observing that commitment in global types reflects directly on local types by projection) whereas $\lrule{LCtxt}$ is straightforward by induction.

$\lrule{LSnd}$ for some path $\pth$:
\[
\begin{array}{lll}
\tstate{\ltheta}{\lenv'}{(\role r_1, \ST_{11} \blacktriangleright^{c} \ST_{12} , \sigma_1),(\role r_2, \ST_2, \sigma_2)}\trans{\role r_1\role r_2!\llab_k}
\\
\qquad\qquad(\role r_1, \ST_{11}' \blacktriangleright^{c} \ST_{12} , \sigma_1),(\role r_2, \ST_2, \sigma_2[\role r_1\mapsto \sigma_2 \cdot (a_k, \pth)])
\end{array}
\]

with premise

\begin{equation}\label{eq:prem11}
\begin{array}{lll}
\tstate{\ltheta}{\pth'.\pleft}{(\role r_1, \ST_{11}, \sigma_1),(\role r_2, \ST_2, \sigma_2)}\trans{\role r_1\role r_2!\llab_k}\\
\qquad\qquad(\role r_1, \ST_{11}', \sigma_1),(\role r_2, \ST_2, \sigma_2[\role r_1\mapsto (a_k, \pth)])
\end{array}
\end{equation}
Observe that $\pth$ may refer to a MC that is nested inside $c$.

By induction on the premise $(\ref{eq:prem11})$ we have that
\[\gspair{\GT_1}{\Theta}\trans{\role r_1\role r_2!a_k}\GT_1'\]
and the projection of $\GT_1'$ on role $\role r_1$ is
$\ST_{11}'', \sigma_1$ with $\ST_{11}'<:\ST_{11}''$.

The assumption for this case (1), that $\role r_1\not\in \rset$, satistfies the premise of rule $\mathtt{[LSnd]}$ of the global semantics, which we can then apply:
\[\gspair{\GtoP{\GT_1}{c}{\lset}{\rset}{\GT_2}}{\Theta}\trans{\role r_1\role r_2!a_k}\GtoP{\GT_1'}{c}{\lset}{\rset}{\GT_2} = \GT'\]

It remains to show the thesis for the projection of $\GT'$ on $\role r_2$.
%

Recall $(\role r_2, \ST_2, \sigma_2)$ moves to $(\role r_2, \ST_2,  \sigma_2[\role r_1\mapsto (a_k, \pth)])$.
By \Cref{lem:prestale} the projected local type from $\GT'$ on $\role r_1$ is $\ST_2$ itself. So, $\GT'\proj \role r_2 = \ST_2,\, \sigma_2'$ for some $\sigma_2'$.
It remains to show that $\sigma' = \gc(\sigma_2[\role r_1\mapsto (a_k, \pth)])$ that is, the new message $(a_k, \pth)$ is in the projected queue if and only if it is not garbage collected (i.e., not stale), which follows by ~\Cref{lem:stale}.

If $Y$ moves on the RHS we have three possible rules: $\lrule{RSnd}$ or $\lrule{RRcv}$ or $\lrule{RCtxt}$. We show $\lrule{RSnd}$ below.

Let $Y_2 = (\role r_2, \ST_2, \sigma_2)$. By  $\lrule{RSnd}$:
\[
\begin{array}{lll}
\tstate{\ltheta}{\lenv'}{(\role r_1, \ST_{11} \lmixi{c} \ST_{12} , \sigma_1),(\role r_2, \ST_2, \sigma_2)}\trans{\role r_1\role r_2!\llab_k}
\\
\qquad\qquad (\role r_1, \bullet \lmixi{c} \ST_{12}' , \sigma_1),(\role r_2, \ST_2, \sigma_2[\role r_1\mapsto \sigma_2 \cdot (a_k, \pth)])
\end{array}
\]

with premise

\begin{equation}\label{eq:prem1}
\begin{array}{ll}
\tstate{\ltheta}{\lenv'.\pright}{(\role r_1, \ST_{12}, \sigma_1),(\role r_2, \ST_2, \sigma_2)}\trans{\role r_1\role r_2!\llab_k} \\
\qquad\qquad (\role r_1, \ST_{12}', \sigma_1),(\role r_2, \ST_2, \sigma_2[\role r_1\mapsto (a_k, \pth)])
\end{array}
\end{equation}

By induction on the premise $(\ref{eq:prem1})$ we have
\[\gspair{\GT_2}{\Theta}\trans{\role r_1\role r_2!a_k}\GT_2'\]
and the projection of $\GT_2'$ on role $\role r_1$ is
$\ST_{12}'', \sigma_1$ with $\ST_{12}'<:\ST_{12}''$.

Since $\role r_1\in\rset$ then $\GT_2$ cannot make action $\role r_1\role r_2!a_k$ hence we can apply $\mathtt{[RSnd]}$ of the global semantics obtaining:
\[\gspair{\GtoP{\GT_1}{c}{\lset}{\rset }{\GT_2}}{\Theta}\trans{\role r_1\role r_2!a_k}\GtoP{\GT_1}{c}{\lset}{\rset\cup\{\role r_1\}}{\GT_2'} = \GT'\]

The projection of $\GT'$ on $\role r_1$ is
$\bullet  \blacktriangleright^{\inst} \ST_{12}'' , \sigma_1$ since $\role r_1$ is committed to the RHS, with the projection of the LHS given by induction. Also by induction $\ST_{12}'<:\ST_{12}''$ hence by definition of `$<:$' we have $\bullet  \blacktriangleright^{\inst} \ST_{12}' <: \bullet  \blacktriangleright^{\inst} \ST_{12}''$ as desired for $\role r_1$. The thesis for $\role r_2$ follows by \Cref{lem:prestale} (the local type is unchanged) and \Cref{lem:stale} (the queue of the reached state is the projected queue modulo garbage collection).

Case (2) is similar to case (1) but simpler, since only LHS moves are possible by $\role r_1$. In case (3) we have $Y_1 = (\role r_1, \bullet \blacktriangleright^{\inst} \ST_{12} , \sigma_1)$ hence the last rule applied is $\lrule{RCtxt}$ and the thesis is straightforward by induction.

\paragraph{5.\textbf{Recursion}}

Directly by induction.
\end{proof}

\subsection{Top-Down fidelity}
\label{asec:completeness}
\fullcomplete*
\begin{proof}
\Cref{thm:complete} follows directly from the more general property for transitions with general $\pth$ environments.
Let $\pth \vdash \GT \mathbin{\proj} \role{r} = \ST_{\role r}, \sigma_{\role r}$ for all $\role{r}\in\RolesG{\GT}$.
We also set $Y_{\role r} = (\role r, \ST_{\role r}, \sigma_{\role r})$ and $Y_{\GT} = \{ Y_{\role r}\}_{\role r\in \RolesG{\GT} }$.
%

We proceed by induction on the transition of $\GT$, reasoning by case analysis on the last rule used.
$\mathtt{[Snd]}$ Assume $\GT = \Gact$. By projection \[Y_{\GT} = (\p , \q \oplus_{i \in I} a_i . \ST_i, \sigma_{\role p}), (\q, \p \&_{i \in I} a_i . \ST_i, \sigma_{\q}), Y_r\]
with $Y_r$ being the configuration of participants in $\RolesG{\GT}\setminus\{\p, \q\}$.

By $\mathtt{[Snd]}$ of the global semantics, for a $k\in I$,
\[\gspairf{\Gact} \trans{\lsnd{p}{q}{a_k}} \Gmsg\]
By one application of rule $\lrule{Par}$ of the local semantics with one application of $\lrule{Snd}$ as a premise:
\[\begin{array}{lll}
\tstate{\ltheta}{\epath}{(\p , \q \oplus_{i \in I} a_i . \ST_i, \sigma_{\role p}), (\q, \p \&_{i \in I} a_i . \ST_i, \sigma_q[\p\mapsto \vec{m}]), Y_r }
 \trans{\p\q!\llab_k} Y_p', Y_q', Y_r
\end{array}
\]
with $Y_p' = (\p , \ST_k, \sigma_{\role p})$ and $Y_q'=(\q, \p \&_{i \in I} a_i . \ST_i, \sigma_q[\p\mapsto \vec{m}\cdot (a_k, \epath)])$.
The projection of $\Gmsg$ on $\p$ and $\q$ are $\ST_k, \sigma_{\role p}$ and $\p \&_{i \in I} a_i . \ST_i, \sigma_q[\p\mapsto \vec{m}\cdot (a_k, \epath)]$, respectively, which correspond to the types and queues in $Y_p'$ and $Y_q'$. The configurations for the other participants are unchanged in the projection (they remain the merge of the projection on that participant of all $\GT_i$ with $i\in I$ -- the third-party case of projection for interaction and message in transit are the same) and after the transition, hence done.

$\mathtt{[Rcv]}$ This case proceeds as $\mathtt{[Snd]}$ except the precise correspondence between $\GT'$ and $Y'$ is lost and `$<:$' is required.
Assume $\GT = \Gmsg$. The projection is of the form

\[Y_{\GT} = (\q, \p \&_{i \in I} a_i . \ST_i, \sigma_q[\p\mapsto (a_k, \epsilon) \cdot \vec{m} ]), (\role r, \sqcap_{i\in I} \ST'_i, \sigma_{\q}), Y_p, Y_r\]
where $Y_p$ is the projection on the sender (not important in this case) and $\role r\not\in\{\p, \q\}$ (without loss of generality we focus on one single $\role r$).

By $\mathtt{[Rcv]}$ of the global semantics,
\[\Gmsg  \trans{\lrcv{p}{q}{a_k}} \GT_k \]
By one application of rule $\lrule{Par}$ of the local semantics with one application of $\lrule{Snd}$ as a premise:
\[\begin{array}{lll}
\tstate{\ltheta}{\epath}{Y_\GT}
 \trans{\p\q?\llab_k}
 (\q, \ST_k, \sigma_q[\p\mapsto  \vec{m} ]), (\role r, \sqcap_{i\in I} \ST'_i, \sigma_{\q}), Y_p, Y_r
\end{array}
\]
Note that $Y_\role r$ is unchanged (as all other roles other than $\q$. The projection of $\GT_k$ on $\p$ is still $Y_p = \proj{\epsilon}{\GT_k}{\p}$. The projection of $\GT_k$ on $\role q$ is indeed the state reached by $\q$, that is $(\q, \ST_k, \sigma_q[\p\mapsto  \vec{m} ])$ with $\sigma_k = \sigma_q[\p\mapsto  \vec{m} ]$. As to all other $\role r$, their projection is now $\proj{\epsilon}{\GT_k}{\role r} = \ST_k, \sigma_k$. But
\begin{equation}\label{boom}
\sqcap_{i\in I} \ST_i  <: \ST_k\end{equation}
hence done.
Observe that (\ref{boom}) holds: (1) by `$<:$' (second rule) if $\role r$ has a receive/branching prefix, (2) otherwise by $\sqcap_{i\in I} \ST_i  = \ST_k$.

$\mathtt{[Cont1]}$, $\mathtt{[Cont2]}$, and $\mathtt{[Rec]}$ are straightforward by induction.

$\mathtt{[Inst]}$ We can assume $\GT = \GtoDef{  \GT_1}{c}{\lset}{\rset}{  \GT_2}$. By projection
\begin{equation}\label{hypn1}
Y_{\GT} = \{(\p , \ST_{1\role{p}} \mixi{c'} \ST_{2\role{p}}, \sigma_0)\}_{\p \in \RolesG{\GT}}
\end{equation}

By $\mathtt{[Inst]}$ of the global semantics
\[\gspairf{\GtoDef{  \GT_1}{c}{\lset}{\rset}{  \GT_2}} \trans{\lnewlab}\GtoP{  \GT_1}{c}{\emptyset}{\emptyset}{  \GT_2}\]
with
\[\GtoP{  \GT_1}{c}{\emptyset}{\emptyset}{  \GT_2} \proj ~= \{(\p , \ST_{1\role{p}} \lmixi{c} \ST_{2\role{p}}, \sigma_0)\}_{{\p\in \RolesG{\GT}}}\]

A configuration in $Y$ that has been projected can either have MC $c$ to instantiate top level or after some communication actions. Let $J \subseteq \RolesG{\GT}$ be the roles whose configurations have the MC $c$ instantiated by $\GT$ at top level, and $m = |J|$. With $m$ applications of rule
$\lrule{Par}$ with $\lrule{New}$ of the local semantics as premise:
\[\begin{array}{lll}
\tstate{\ltheta}{\epath}{Y_{\GT}}\trans{\vv{\nu}}  \{(\p , \ST_{1\role{p}} \lmixi{c}\ST_{2\role{p}}, \sigma_0)\}_{{\p\in \RolesG{\GT}\cap J}} \cup \{(\p , \ST_{1\role{p}} \mixi{c} \ST_{2\role{p}}, \sigma_0)\}_{{\p\in  J}}
\end{array}
\]
The thesis follows by observing that for all $\p\in J$,
$(\p , \ST_{1\role{p}} \mixi{c'} \ST_{2\role{p}}, \sigma_0) <: (\p , \ST_{1\role{p}} \lmixi{c'} \ST_{2\role{p}}, \sigma_0)$.

%
%
%

$\mathtt{[Ctx1]}$ Assume $\GT = \GtoP{  \GT_1}{{c},{n}}{\lset}{\rset}{\GT_2}$. By projection we have two cases: either $\lset$ is empty or not.

First, assume $\lset\neq\emptyset$. Since $\GT$ is reachable from an aware and balanced state (hypothesis), it is coherent (Lemma~\ref{lem:cohpre}) and hence $\rset=\emptyset$. Without loss of generality assume
 \[Y_{\GT} =  \{(\p , \ST_{\role p } \,\blacktriangleright\, \bullet, \sigma_{\role p})\}_{{\p\in \RolesG{\GT}\cap \lset}}~~\cup~~  \{(\p , \ST_{1\role{p}} \,\blacktriangleright\, \ST_{2\role{p}},  \sigma_{\role p})\}_{{\p\in \RolesG{\GT} \setminus \lset}}\]
where $\RolesG{\GT}$ is partitioned in the set of committed roles (in $\lset$), and not committed roles (not in $\lset$).
By $\mathtt{[Ctx1]}$ of the global semantics (omitting $c,n$),
\begin{equation}\label{insteq}\infer{\gspairf{ \GT_1} \trans{\lnewlab}\GT_1'}{\gspairf{\GtoP{  \GT_1}{}{\lset}{\rset}{\GT_2}} \trans{\lnewlab}\GtoP{  \GT_1'}{}{\lset}{\rset}{  \GT_2}}\end{equation}

Focussing on the LHS blocks of the local types in $Y_{\GT}$:
\[Y_{1} = \{(\p , \ST_{\role p}, \sigma_{\role p})\}_{\p \in \RolesG{\GT}\cap\lset}~\cup~ \{(\p , \ST_{1\role p}, \sigma_{1\role p} ) \}_{\p \in \RolesG{\GT}\setminus \lset}\]

By induction, looking at the premise in (\ref{insteq})
\[Y_{1} \trans{\vv{\lnewlab}} \{(\p , \ST_{\p}', \sigma_{\role p})\}_{\p \in \RolesG{\GT}\cap\lset}~\cup~ \{(\p , \ST_{1\role p}', \sigma_{1\role p} ) \}_{\p \in \RolesG{\GT}\setminus \lset}\]
where for all $\p\in \RolesG{\GT}\setminus \lset$, $\sigma_{\role p} = \sigma_{1\role p}\circ\sigma_{2\role p}$ for some $\sigma_{2\role p}$.
Moreover,
\begin{equation}\label{n1}
\forall \role p\in \RolesG{\GT}\cap\lset,~~~\pleft \vdash \GT_1' \mathbin{\proj} \p = \ST_\role{p}', \sigma_{\role p}
\end{equation}
and
\begin{equation}\label{n2}
\forall \p\in \RolesG{\GT}\setminus \lset,~~~\pleft \vdash \GT_1' \mathbin{\proj} \p = \ST_{1\role{p}}', \sigma_{1\role p}
\end{equation}

We need to check that the projection of $\GT'$ is $Y'$.
The thesis for roles committed on the LHS is straightforward by induction, as shown in (\ref{n1}). The case for $\p \in \RolesG{\GT}\setminus\lset$ follows by observing that: (1)
projection of $\GT_1'$ into $\ST_{1\role{p}}, \sigma_{1\role p }$ follows by induction as shown in (\ref{n2}); (2) $\ST_2$ is unchanged hence projection of $\GT_2'$ into $\ST_{2\role{p}},\sigma_{2\role p}$ follows by hypothesis and projection.
The queues are unchanged hence still $\sigma_{\role p} = \sigma_{1\role p }\circ\sigma_{2\role p }$.

The case for $\lset=\emptyset$ is symmetric, with possibly some role committed on the RHS:
\[Y_{\GT} =  \{(\p , \bullet \, \blacktriangleright\, \ST_{\role p } \, , \sigma_{\role p})\}_{\p\in \RolesG{\GT}\cap\rset}~~\cup~~  \{(\p , \ST_{1\role{p}} \,\blacktriangleright\, \ST_{2\role{p}}, \sigma_{\role p})\}_{\p\in \RolesG{\GT}\setminus\rset}\]

By $\mathtt{[Ctx1]}$ of the global semantics,
\[\infer{\gspairf{\GT_1} \trans{\lnewlab}\GT_1'}{\gspairf{\GtoP{  \GT_1}{{c},{n}}{\lset}{\rset}{\GT_2}} \trans{\nu \inst}\GtoP{  \GT_1'}{{c},{n}}{\lset}{\rset}{  \GT_2}}\]

By induction,
(where $Y_c$ are the configurations of roles in  $\RolesG{\GT}\cap\rset$ -- already committed on the right):
\begin{equation}\label{extra}
Y_{1} = \{(\p , \ST_{1\role{p}}, \sigma_{1\role p} ) \}_{\role p\in \RolesG{\GT}\setminus \rset}\cup  Y_{c}\trans{\vv{\lnewlab}} \{(\p , \ST_{1\role{p}}', \sigma_{1\role p} ) \}_{\role p\in \RolesG{\GT}\setminus \rset}\cup  Y_{c}\end{equation}
where for all $\p\in \RolesG{\GT}\setminus \rset$, $\sigma_{\role p} = \sigma_{1\role p }\circ\sigma_{2\role p }$ for some $\sigma_{2\role p}$.

By induction,
\begin{equation}\label{n2rr}
\forall \p \in \RolesG{\GT}\setminus \rset,~~~\pleft \vdash \GT_1' \mathbin{\proj} \p = \ST_{1\role{p}}', \sigma_{1\role p}
\end{equation}

The thesis for $\p \in \RolesG{\GT}\setminus \rset$ follows from observing that projection is defined inductively on the two sides and: (1) (\ref{n2rr}) and (\ref{extra}) for the LHS block; (2) the global and local types are unchanged hence the thesis for the RHS follows by hypothesis.
The queues are unchanged hence still $\sigma_{\role p} = \sigma_{1\role p }\circ\sigma_{2\role p }$.

%
%
%
$\mathtt{[LSnd]}$ Let $\GT = \GtoP{  \GT_1}{}{\lset}{\rset}{\GT_2}$.
By $\mathtt{[LSnd]}$
\begin{equation}\label{inf4}
\infer{\gspairf{\GT_1} \trans{\lsnd{p}{q}{a_k}} \GT_1'}{\gspairf{\GT} \trans{\lsnd{p}{q}{a_k}} \GtoP{  \GT_1'}{}{\lset}{\rset}{\GT_2}}\end{equation}


We have two cases: (a) $\p$ is uncommitted, (b) $\p$ is committed.

Case (a): If $\p$ is uncommitted, let $Y_1$ be the system obtained by projecting $\GT_1$ on all its roles.
By induction, by the premise in (\ref{inf4}) it follows, observing that the definition of `$<:$' for systems ignores the queues (see \Cref{asec:st}):
\begin{equation}\label{eqlsnd2}
Y_1\trans{\lsnd{p}{q}{\llab_k}} Y_1' ~~~\land~~~ Y_1'<:\GT_1'\proj
\end{equation}
From \Cref{eqlsnd2} we obtain (via definition of '$<:$') $Y_1' \blacktriangleright Y_2 <: Y'$.
Lookinq at the queues now, without loss of generality, assume the only queue that changed in the transition from $Y$ to $Y'$ is $\q$'s. Denoting the queue of $\q$ in $Y$ as $\sigma_{\q}$ and the queue of $\q$ in $Y'$ as $\sigma_{\q}'$ we have
$\sigma_{\q}' = \sigma_{\role q}[\p\mapsto \sigma_{\role q}(\p) \cdot (a_k, \pth)]$.
The queues of all other roles are unchanged.
By \Cref{nogarbage} (garbageless projections)
$Y$'s queues have no stale messages. Hence the only message that may cause the queues of $Y'$ be different from thos of $\GT'\proj$ is $(a_k, \pth)$. By \Cref{lem:stale} we have that $(a_k, \pth)$ is stale (can be garbage collected) if and only if it is not in in the queues of $\GT'\proj$, as desired.

Case (b) is essentially as case (a) except the projection of $\ST_{\role p }$ ignores the right-hand side process.

%
%
%
$\mathtt{[LRcv1]}$ Let $\GT = \GtoP{  \GT_1}{}{\lset}{\rset}{\GT_2}$.
By $\mathtt{[LRcv1]}$
\[
\infer{\gspairf{\GT_1} \trans{\lrcv{p}{q}{a_k}} \GT_1' \qquad \q\not\in \rset\land(\p\in\lset\lor \role{q}=\role{r})}{\gspairf{\GT} \trans{\lrcv{p}{q}{a_k}} \GtoP{  \GT_1'}{}{\lset\cup\{\q\}}{\rset}{\GT_2}}\]

Without loss of generality,
\begin{equation}\label{eqlrcv1}
Y_{\GT} = Y_{\q} ,Y_{rest}\end{equation}
where either (a) $Y_{\q} = (\q, \ST_{1\q}\,\blacktriangleright\, \ST_{2\q}, \sigma_{\q})$ or (b) $Y_q = (\q, \ST_{1\q}\,\blacktriangleright\, \bullet,  \sigma_{\q})$. In either case the thesis follows by projection where the LHS of the outmost MC is the projection of $\GT_1'$ by induction, whereas the RHS of the MC is just left out by the second projection rule for active MC.
%
%
%
In the case for $\mathtt{[LRcv2]}$ either $\role{q}$ is already committed or remains uncommitted. In the first case, $Y_{\q} = (\q, \ST_{1\q}\,\blacktriangleright\, \bullet,  \sigma_{\q})$ and in the second case $(\q, \ST_{1\q}\,\blacktriangleright\, \ST_{2\q}, \sigma_{\q})$.
The case is similar to $\mathtt{[LRcv1]}$, except the action is not a committing action hence $Y_{\q}$ and $Y_{\q}'$ have the same form.

%
%
%
The case for $\mathtt{[RSnd]}$ is symmetric to the case for  $\mathtt{[LSnd]}$ but simpler: no action can happen on the right hand side if $\lset$ is not empty so there is no need for garbage collection.
The case for $\mathtt{[RRcv]}$ is as $\mathtt{[RSnd]}$ since on the right-hand side both send and receive are committing.
\end{proof}

\subsection{Lemmata for \Cref{thm:oc}}

\newcommand{\wco}{\approx^{(1)}}
\newcommand{\wct}{\approx^{(2)}}

\PURPLE{

Theorem 4.4 requires to show that the following relation is a weak correspondence (\Cref{def:wc}). 

\begin{equation}
\label{R}
    R = \{(\GT,Y) ~| ~\GT \text{ is reachable from an initial, aware, balanced global type }.~Y <: \GT\proj\}
\end{equation}

\Cref{wc:top} and \Cref{thm:complete} gives us only one step of correspondence when $Y = \GT\proj$. We need to show the transitive closure holds with $Y <: \GT\proj$. 

Given a pair $(\GT,Y)\in R$, we use the notation $\GT\wco Y$ to denote condition (1) in  \Cref{def:wc} and $\GT\wct Y$ to denote condition (2) in \Cref{def:wc}. The fact that $R$ is a weak correspondence follows from \Cref{para} and \Cref{para3}.





}

\begin{lemma}\label{para1}
\PURPLE{If $\ST <: \ST'$ and $(\p, \ST, \sigma)\trans{\ell} (\p, \ST'', \sigma)$ then $(\p, \ST', \sigma)\trans{\ell} (\p, \ST''', \sigma)$ and $\ST'' <: \ST'''$. 
}
\end{lemma}

\begin{proof}(sketch)
The proof is by induction on the derivation of $\ST <: \ST'$ proceeding by case analysis on the last rules in \Cref{asec:st}. The base cases are trivial as $\tend$ and $\mathtt t$ do not make any transition. 
The case for 
\[
\inferrule{
\ST_k' \mathbin{<:} \ST_k \quad k\in I
}{
\p \, \&_{i \mathclose{\in} I} \, a_i.\ST_i'
\mathbin{<:}
\p \, \& \, a_k.\ST_k
}
\]
follows by induction since $\ST' = \p \, \&_{i \mathclose{\in} I} \, a_i.\ST_i'$
has more action options than $\ST = \p \, \& \, a_k.\ST_k$ and $\ST_k<:\ST_k'$ by inductive definition of `$<:$'.

The cases for rule 
\[\inferrule{
\ST_1' \mathbin{<:} \ST_1
\qquad
\ST_2' \mathbin{<:} \ST_2}
{
\ST_1' \mixi{c} \ST_2'
~\mathbin{<:}~
\ST_1 \lmixi{c} \ST_2
}\]
with $\ST = \ST_1' \mixi{c} \ST_2'$ and $\ST'=\ST_1 \lmixi{c} \ST_2$ follows by observing that the two terms in the conclusion can perform the same actions modulo a  MC instantiation (i.e., $\tau$) that $\ST_1' \mixi{c} \ST_2'$ can always perform by local semantic rule  \lrule{New}. Similarly for the rule below, by action \lrule{Rec}.
\[\inferrule{
 \ST' [\mu \mathtt{t}. \ST'/\mathtt{t}]\mathbin{<:} \ST
}{
\mu \mathtt{t}. \ST' \mathbin{<:} \ST
}\]
All other rules are straightforward by induction as left and right-hand side terms of `$<:$' in the conclusions are the same. 
\end{proof}

\begin{lemma}\label{para}
\PURPLE{If $(\GT,Y)\in R$ then 
$ \GT \wco Y$.}
\end{lemma}
\begin{proof}
By $R$ we have $\GT \proj~ <: Y$. If $\GT\trans{\ell}\GT'$ by \Cref{wc:top}
$\GT \proj\myreaches{\ell} Y''$ and $Y'' <: \GT' \proj$. 

We extend \Cref{para1} to configurations: observe that $\GT \proj$ and $Y$ can be decomposed into a finite number of configurations by rule $\mathtt{[PSystem]}$, recalled below for convenience with $Y = \{Y_{\role{r}}^2 \}_{\role r \in P}$ and $
\GT \proj =\{Y_{\role{r}}^1 \}_{\role r \in P}$.
\[\begin{array}{ll}
\inferrule{
\forall \role{r} \mathop{\in} P ~~ Y_\role{r}^2\mathbin{<:} Y_\role{r}^1
}{
\{Y_{\role{r}}^2 \}_{\role r \in P} \mathbin{<:} \{Y_{\role{r}}^1\}_{\role{r} \mathclose{\in} P}
}~~~~\mathtt{[PSystem]}
\end{array}\]
By applying \Cref{para1} to the pairs of corresponding (same participant) configurations, by hypothesis 
$\GT \proj~ <: Y$ we obtain $Y\myreaches{\ell} Y'$ and $Y' <: Y''$. 
Hence, by transitivity of `$<:$' we have $Y' <: \GT' \proj$ as desired. 
\end{proof}

\begin{lemma}\label{para2}
If $~\ST_1 <: \ST_2$, $(\p, \ST_1, \sigma) \trans{\ell}(\p, \ST_1', \sigma)$, and $\ell$ is not a receive action, then \\ $(\p, \ST_2, \sigma) \myreaches{\ell}(\p, \ST_2', \sigma)$ and $\ST_1' <: \ST_2'$.
\end{lemma}
\begin{proof}
First, we show that $(\p, \ST_2, \sigma) \trans{\ell}$ by induction on the derivation $\ST_1 <: \ST_2$ proceeding by case analysis on the last rule used. 

Since $(\p, \ST, \sigma)$ does not make a receive action, then the last rule used cannot be a receive/branching (first rule where $\wild=\&$ and second rule in \Cref{asec:st}). By inspection of the remaining rules, the thesis follows by inductive hypothesis. Action $\ell$ can be mimicked immediately unless $\ST$ is a MC definition, in which case it is mimicked after a MC initialization (i.e., a $\tau$ action). 
The fact that the continuation is still in a preorder relation is given by $\ST <: \ST'$ and by inductive definition of `$<:$'.
\end{proof}

For receiving actions, we have a weaker lemma which follows mechanically by induction on the derivation of `$<:$'. 
\begin{lemma}\label{para22}
If $~\ST_1 <: \ST_2$, $(\p, \ST_1, \sigma) \trans{\ell}(\p, \ST_1', \sigma)$ with $\ell$ a receive action, and $(\p, \ST_2, \sigma) \myreaches{\ell}(\p, \ST_2', \sigma)$ then $\ST_1' <: \ST_2'$.
\end{lemma}

\begin{lemma}\label{para3}
\PURPLE{If $(\GT, Y)\in R$ then $ \GT \wct Y$.}
\end{lemma}
\begin{proof}

By hypothesis we have $Y <: \GT \proj$. 
First, assume $Y\trans{\ell}Y'$ and $\ell$ is not a receive action. 
By \Cref{para2} we have $\GT\proj \myreaches{\ell} Y''$ and $Y' <: Y''$. 
By \Cref{thm:complete} $\GT \trans{\ell}\GT'$ and $Y'\myreaches{}Y'''$ with $Y'''<: \GT'\proj$. 
Again, since $Y'\myreaches{}Y'''$ and $Y\trans{\ell}Y'$, by \Cref{para2} we have 
$Y'\myreaches{}Y''''$ and $Y'''' <: Y'''$. By transitivity of `$<:$' we have $Y'''' <: \GT'\proj$ as desired. 
\footnote{Note, the step $Y'\myreaches{}Y''''$ is necessary when the role that is supposed to receive the message sent (for now stored in its queue) is getting a message that commits them to the opposite side wrt where the message is received. In this case, the queue of the receiver may include stale messages that need to be purged to restore equality of the queues between $Y'$ and $\GT'\proj$ (projected systems never have stale messages).}

The case for receive action is interesting. 
If $Y\trans{\ell}Y'$ by a receive action, say $\ell = \p\q?(a_k,\pi)$ then this action only involves one configuration in $Y$: the configuration of the receiver $(\q, \ST_{2\q}, \sigma_{2\q})$ and $\sigma_{2\q}(\p) = \vec{m_1}\cdot(a_k,\pi)\cdot\vec{m_2}$.
This case is interesting because the local type of $\GT\proj \q$ could be of the form $\p\, \& \, a_j.\ST_j'$ where $\q$ is already set to receive a specific label $j$ by the second rule of $<:$ in \Cref{asec:st} and this requires us to show that $k=j$ for any $k$ that $\q$ may have chosen in $Y$. 

We use hypothesis $Y <: \GT\proj$ focussing on the configurations $Y_\p$ of sender $\p$ and $Y_\q$ receiver $\q$. Let $\GT\proj \p = {\ST_{\p}}, \sigma_\p$ and $\GT\proj \q = {\ST_{\q}}, \sigma_\q$. By $\mathtt{[PConfig]}$ we have $Y_{\p} <: (\p, {\ST_{\p}}, \sigma_\p)$ and $Y_{\q} <: (\q, {\ST_{\q}}, \sigma_\q)$.

If $(a_k,\pi)$ is in the queue of $\q$ in $Y_\q$ then it is also in $\sigma_\q(\p)$, which means $\p$ has already sent the message $k$ in $\GT$. More precisely, $\GT$ must have a subterm which is a message in transit by $\p$ to $\q$ of label $a_k$ in the path $\pi$.


So, the branch selected by $\p$ in $Y$ must be the same as the one selected by $\GT\proj$ yielding $j=k$. Hence, also  $\GT\proj \trans{\p\q?(a_k,\pi)}Y''$ and by \Cref{para22} $Y' <: Y''$. This case now proceeds as the one above for non-receiving actions.

\end{proof}

\newcommand{\LAC}{\Ctname_{\ST \mathtt e}}
\newcommand{\GAC}{\Ctname_{\GT\mathtt e}}
\newcommand{\Cag}{\Ctname_\GT}
\newcommand{\Cal}{\Ctname_\ST}
\newcommand{\Path}[1]{\mathtt{Path}(#1)}

\section{Orphan Message Freedom (OMF)}\label{app:om}
\label{app:orphan}\label{asec:er}

To reason on orphan messages is presence of MC, it helps to define the notion of active context. Intuitively, an active context is a context where the hole is \emph{not} in a stale path.


\begin{definition}[Active Contexts]\label{def:acontexts}
We define below global active contexts $\Cag(\p)$.
Global active contexts are parameterized on a role $\p$. Parameterization is necessary as some contexts may be active only for some participants when other participants have already committed to the opposite side.  We write $\Cag$ (omitting the parameter for readability) when not relevant or clear from the context.

\[\begin{array}{rcll}
\Cag(\p)& ::= & [] \\
             & \mid &
             \GactShort{p}{q}{\St{}}\cup\{a.\Cag(\p)\}\\
  & \mid &   \GmsgShort{p}{q}{k}{\St{}}\cup\{a_k.\Cag(\p)\}\\
              & \mid &
             \Cag(\p) \blacktriangleright_{\lset,\rset} \GT & \quad (\lset\cup \rset \subseteq \Roles(\GT) \text{ and }\p\not\in\rset) \\
             & \mid & \GT \blacktriangleright_{\lset,\rset} \Cag(\p) & \quad ( \lset\cup\rset\subseteq \Roles(\GT) \text{ and }\p\not\in\lset) \\
%
\end{array}\]
We write $\Cag^\circ$ to denote a general global context obtained by the rules for $\Cag$ but omitting the side conditions $\p\not\in\rset$ and $\p\not\in\lset$ in the fourth and fifth grammar rules, respectively. Note that $\Cag^\circ$ may contain stale paths. Note also that $\Cag^\circ$ differs from $\Ctname$ defined in the main text of this paper as the former only involves runtime elements, no MC definitions and recursive types.
\end{definition}

The following definition is useful to correlate paths of global types with message paths of corresponding local types.

\begin{definition}[Path of global active contexts]
Define the path of an active context inductively as follows:
 \[\begin{array}{ll}
 \Path{[]} = \epsilon \\
 \Path{\GactShort{p}{q}{\St{}}\cup\{a.\Cag\}} = \Path{\Cag(}\\
 \Path{\GmsgShort{p}{q}{k}{\St{}}\cup\{a_k.\Cag\}} = \Path{\Cag}\\
 \Path{\Cag \blacktriangleright_{\lset,\rset}\GT} = \pleft . \Path{\Cag} \\
 \Path{ \GT \blacktriangleright_{\lset,\rset} \Cag} = \pright . \Path{\Cag}
 \end{array}\]
 \end{definition}

 \begin{definition}[Actions $\trans{\ell@\pi}$]
 We say that $\Cag[\GT]$ moves at $\pi$ with label $\ell$, written $\Cag[\GT]\trans{\ell@\pi}$ if $\Cag[\GT]\trans{\ell}\Cag' [\GT']$ with $\GT\neq \GT'$ and $\pi = \Path{\Cag} = \Path{\Cag' }$.
 \end{definition}

 \begin{definition}[Global staleness]
 We say that $\pi$ is stale in $\GT$ if there exists no $\Cag$ such that $\GT = \Cag[\GT'] $ and $\Path{\Cag} = \pi$.
 \end{definition}

\begin{proposition}\label{erg}
For all $\GT = \Cag(\q)[\GmsgShort{p}{q}{k}{\{a_i.\GT_i\}_{i\in I}}]$ reachable from an initial, aware, balanced global type, (with $\pi = \Path{\Cag}$),
there exists $\GT'$ reachable from $\GT$ such that either $\GT'\trans{\p\q?a_k@\pi}$ or $\pi$ is stale in $\GT'$.
\end{proposition}
\begin{proof}(sketch)
We reason by induction on the structure of $\Cag$.

Case $\Cag = []$. The thesis $\GT'\trans{\p\q?a_k@\pi}$ follows immediately by an action for \glab{Rcv} with $\GT=\GT'$.

Case $\Cag = \GactShort{s}{r}{\St{}}\cup\{a.\Cag ' \}$. By induction:
\[\Cag ' [\GmsgShort{p}{q}{k}{\{a_i.\GT_i\}_{i\in I}}] \trans{}^* \trans{\p\q?a_k@\pi'}\]
 with $\pi'=\Path{\Cag '}$. By two applications of rule \glab{Snd} and \glab{Rcv}, respectively,  $\GT\trans{\role s\role r!a, \role s\role r?a} \Cag ' [\GmsgShort{p}{q}{k}{\{a_i.\GT_i\}_{i\in I}}] $. Observe that $\Path{\Cag ' } = \Path{\Cag} = \pi$ and hence $\GT \trans{}^* \trans{\p\q?a_k@\pi}$ with the last action being at $\Path{\Cag }=\pi$ as required.
The case for $\GmsgShort{p}{q}{k}{\St{}}\cup\{a_k.\Cag\}$ is similar.

Case $ \Cag(\q) =  \Cag ' (\q) \blacktriangleright_{\lset,\rset} \GT_{\mathtt r} $ (with $\q\not\in\rset$). By induction
\begin{equation}\label{eq:omeq}
\Cag ' [\GmsgShort{p}{q}{k}{\{a_i.\GT_i\}_{i\in I}}] \trans{}^* \GT_{in}'
\end{equation}
and either (1) $\GT_{in}' \trans{\p\q?a_k@\pi'}$ with $\pi' = \Path{\Cag '}$ or (2) $\pi'$ is stale in $\GT_{in}'$.



We have two cases:
\begin{itemize}
\item if $\rset\neq \emptyset$ then, since there is an active dependency from the observer of the mixed choice and $\q$, $\q$ will commit to the RHS, which will make $\pi$ stale in the reached state as desired. Here it does not matter if $\pi'$ in the induction step is stale or not -- case 1 or 2 following (\ref{eq:omeq}), as the whole $\pi$ becomes stale by effect of an action on the RHS of the larger context.

\item if $\rset= \emptyset$ and (\ref{eq:omeq}) holds by case (1) then we need to show that all actions by $\Cag ' [\GmsgShort{p}{q}{k}{\{a_i.\GT_i\}_{i\in I}}] $ can also be performed by the larger context $\Cag ' [\GmsgShort{p}{q}{k}{\{a_i.\GT_i\}_{i\in I}}] \blacktriangleright_{\lset,\emptyset}^c \GT_{\mathtt r}$.

First, observe that no action on the LHS of MC $c$ can change the right-hand side set $c$ and hence $\Cag ' [\GmsgShort{p}{q}{k}{\{a_i.\GT_i\}_{i\in I}}] \blacktriangleright_{\lset,\emptyset} \GT_{\mathtt r}  \trans{}^*$ with lead to an outer context where $\rset=\emptyset$. We omit the (mechanical) inner induction on the length of $\trans{}^*\trans{\p\q?a_k@\pi'}$ and show that any action $\ell$ can be executed by its context proceeding by case analysis on the last transition rule used.
The possible cases are: (i) instantiation actions, which can occur in $\Cag$ by \glab{Ctx1} (if there are still uncommitted roles in $\Cag '$ then there are also in $\Cag$ hence the premise of \glab{Ctx1} holds for $\Cag$), (ii) non committing receive actions that can occur in the larger context by \glab{LRcv2}, and (iii) send actions that can occur by \glab{LSnd}. As to committing actions -- that can be only receive actions as no participant can commit on the LHS with a send action (iv): receive actions can  occur in $\Cag$ by \glab{LRcv1} and change the context into $\Cag{} ' \blacktriangleright_{\lset\cup\{\role r'\},\rset} \GT_{\mathtt r} $ where $\mathtt{r'}$ is the subject of the receive action. Observe that the path of the new context is still $\pi$: $\Path{\Cag ' \blacktriangleright_{\lset\cup\{\role r'\},\emptyset} \GT_{\mathtt r}} = \pi$. The previous four cases show that $\Cag ' [\GmsgShort{p}{q}{k}{\{a_i.\GT_i\}_{i\in I}}] \blacktriangleright_{\lset,\emptyset} \GT_{\mathtt r} \trans{}^* \trans{\p\q?a_k@\pleft.\pi'}$ where, by the structure of $\Cag$, $\pleft.\pi'= \pi$, as desired.

If $\rset= \emptyset$ and (\ref{eq:omeq}) holds then let, with no loss of generality, $\ell_i\in \vec{\ell}$ be the first action in the sequence $\Cag ' [\GmsgShort{p}{q}{k}{\{a_i.\GT_i\}_{i\in I}}] \trans{\vec{\ell}}\GT_{in}'$ by which $\pi'$ becomes stale. Using a similar reasoning to case (1) all actions of the inner context can be performed by the outer MC context, including $\ell_i$. The thesis follows by observing that a context with a stale subterm is also stale hence if $\pi'$ becomes stale then also $\pi$ does.

\end{itemize}
Case $ \Cag =  \GT_{\mathtt l} \blacktriangleright_{\emptyset,\rset} \Cag ' $ is symmetric.
\end{proof}

\begin{proposition}[Stale persistency]\label{foreverstale}
If $Y, (\q, \ST, \sigma)\trans{\ell} Y', (\q, \ST', \sigma')$ and $\mathtt{stale}(\pi,\ST)$ then $\mathtt{stale}(\pi,\ST')$.
\end{proposition}
The proof is mechanical by induction on the transition, reasoning by case analysis on the last rule applied. Observe that transitions never remove MC instances and can only change the tree structure of nested MC by either expanding the leaves with new instantiations, or committing hence \emph{adding} (never removing) stale sides $\bullet$.

\begin{proposition}[Local-global activeness]\label{inversionlem}
Let $Y, (\q, \ST, \sigma) = \GT \proj $ and  $\sigma[\p] = (a, \pi) \cdot \vec{m}$ then $\GT =  \Cag(\q)[\GmsgShort{p}{q}{k}{\{a_i.\GT_i\}_{i\in I}}]$ for some $\Cag$ where $\Path{\Cag} = \pi$.
\end{proposition}
\begin{proof}(sketch)
Observe that $\neg \mathtt{stale}(\pi,\ST)$ by garbageless projection (\Cref{nogarbage}).

By inspection of the projection rules, the only global type that can produce a message $(a,\pi)$ that is, by the second projection rule for messages in transit
$\GmsgShort{p}{q}{k}{\{a_i.\GT_i\}_{i\in I}}$.

Assume without loss of generality that $\GT = \Cag^\circ[\Cag'(\q)[\GmsgShort{p}{q}{k}{\{a_i.\GT_i\}_{i\in I}}]]$ (the inner active context may be just the hole). We need to show that $\Cag^\circ$ is an active context for $\q$ with path $\pi$.

The proof is by induction on the structure of $\pi$.

Case $\pi= \epsilon$.
Application of the second rule on  $\GmsgShort{p}{q}{k}{\{a_i.\GT_i\}_{i\in I}}$ can occur after a finite number of applications of instances of the first or the second projection rule (for communication and message in transit).  It results that
$\GT = \Cag^\circ[\GmsgShort{p}{q}{k}{\{a_i.\GT_i\}_{i\in I}}]$ where $\Cag^\circ$ is either a hole $[]$ or a finite concatenation of communication actions and messages in transit. In all these cases, since no MC occurs in $\Cag^\circ$ then $\Cag^\circ$ is an active context on $\q$ hence
$\GT = \Cag(\q)[\GmsgShort{p}{q}{k}{\{a_i.\GT_i\}_{i\in I}}]$
with $\Cag^\circ = \Cag(\q)$ and  $\Path{\Cag}=\pi=\epsilon$ as required.

If $\pi =  \pi'.\pleft$ then without loss of generality
$\GT = \Cag^\circ[ \Cag'(\q)[\GmsgShort{p}{q}{k}{\{a_i.\GT_i\}_{i\in I}} ]  ]$
where $\Cag'(\q)$ is the largest context with only send/receive actions, and since by definition it has no MC it is active. We now use the form of $\pi$ to decompose $\Cag^\circ$ and show it is active. By inspection of the projection rules (by either fourth rule, first and third case) we have observe that $\Cag^\circ$ has the following form:
\[\Cag^\circ = {\Cag^\circ} ' [
\Cag'(\q) \blacktriangleright_{\lset,\rset} \GT_2 ]
\]
If the fourth projection rule - first case - was used to project $\GT$, then $\role q\in\lset$ hence by coherence $\role q\not \in\rset$. If the fourth rule - third case - was used
again, by side condition $\q\not\in\lset\cup\rset$ we have $\q\not\in\rset$. Since in both cases $\q\not\in\rset$ then $
[] \blacktriangleright_{\lset,\rset} \GT_2$ is active, and hence $\Cag'(\q) \blacktriangleright_{\lset,\rset} \GT_2$ is also active, as obtained by plugging an active context into another active one.

By induction, we have ${\Cag^\circ} '$ is active for $\q$ with path $\pi'$. Hence,
$\Cag^\circ$ is active because it is obtained by plugging context $\Cag'(\q) \blacktriangleright_{\lset,\rset} \GT_2$ that is active for $\p$ with $\pi$ into another context ${\Cag^\circ} '$ that is also active for $\q$ with $\pi$ (this is by definition of active contexts).

The case for $\pi = \pi'.\pright$ is symmetric.

\end{proof}

\Cref{inversion} follows by generalizing \Cref{inversionlem}, by mechanical induction on the length of $\vec{m_1}$ and on the derivation of $Y <: \GT\proj$.

\begin{proposition}[Local-global activeness (general)]\label{inversion}
Let $Y, (\q, \ST, \sigma) <: \GT \proj $ and  $\sigma[\p] =\vec{m_1}\cdot (a, \pi) \cdot \vec{m_2}$ then $\GT =  \Cag(\q)[\GmsgShort{p}{q}{k}{\{a_i.\GT_i\}_{i\in I}}]$ for some $\Cag$ where $\Path{\Cag} = \pi$.
\end{proposition}

\enabledreceptionlocal*
\begin{proof}
If $\mathtt{stale}(\pi,\ST)$ the thesis follows immediately.
If $\neg\mathtt{stale}(\pi,\ST)$ then by \Cref{inversion} $Y, (\q, \ST, \sigma) <:  \GT\proj$ for some $\GT =  \Cag[\GmsgShort{p}{q}{k}{\{a_i.\GT_i\}_{i\in I}}]$ with $\Path{\Cag} = \pi$.

By \Cref{erg} there exists $\GT'$ reachable from $\GT$ such that $\GT'\trans{\p\q?a_k@\Path{\Cag}}$.

By \Cref{thm:oc} since $\GT\approx \GT\proj$ we obtain
$\GT\proj\trans{}^*\trans{\p\q?(a_k,\pi)}$.

By $Y, (\q, \ST, \sigma) <:  \GT\proj$ and \Cref{para} we obtain
$Y, (\q, \ST, \sigma)\trans{}^* \trans{\p\q?(a_k\pi)}$, with
$\pi=\Path{\Cag}$ as desired.
\end{proof}

\end{document}